%% file: manuscript.tex
\documentclass{article}
\usepackage{fqi_preprint,times}
\usepackage[T1]{fontenc}
\usepackage[utf8]{inputenc}
\usepackage{microtype}
\usepackage{amsmath,amssymb,amsthm,mathtools,bm}
\usepackage{graphicx,booktabs,tabularx,array,longtable}
\usepackage{placeins}
\usepackage{algorithm,algpseudocode,enumitem}
\usepackage{caption,pdflscape}
\usepackage{tikz}
\usepackage[normalem]{ulem}
\usetikzlibrary{arrows.meta,positioning,calc}
\usepackage{xcolor,xurl,hyperref}
\usepackage{wrapfig}

\hypersetup{colorlinks=true,linkcolor=blue,citecolor=blue,urlcolor=blue,
  pdftitle={Finite-Sample Theory for Fitted Q-Iteration When Actions Are Functions},
  pdfauthor={Gefei Lin, Rui Miao, Xiaoke Zhang}}
\newcolumntype{L}[1]{>{\raggedright\arraybackslash}p{#1}}
\newtheorem{theorem}{Theorem}
\newtheorem{proposition}{Proposition}
\newtheorem{lemma}{Lemma}
\newtheorem{corollary}{Corollary}
\theoremstyle{definition}

\newtheorem{assumption}{Assumption}
\newtheorem{example}{Example}
\newtheorem*{examplecontinued}{Example~\ref{ex:two-state-control} (continued)}
\newtheorem{remark}{Remark}

\DeclareMathOperator{\clip}{clip}

\newcommand{\E}{\mathbb{E}}
\newcommand{\norm}[1]{\left\lVert #1 \right\rVert}
\newcommand{\ind}{\mathbf{1}}
\newcommand{\eps}{\varepsilon}
\newcommand{\nobs}{n}
\newcommand{\Neff}{N_{\mathrm{eff}}}
\newcommand{\Vmax}{V_{\max}}
\newcommand{\nuplus}{\nu_+}
\newcommand{\nuX}{\nu_X}
\newcommand{\law}{\mathcal L}
\newcommand{\Penv}{P_{\mathrm{env}}}
\newcommand{\Pnext}{P_S}
\newcommand{\ProjOp}{\mathbf U}
\newcommand{\WeightOp}{\mathbf W}
\newcommand{\SpecOp}{\mathbf E}
\newcommand{\InterpOp}{\mathbf J}
\newcommand{\BaseWeightOp}{\mathbf W}
\newcommand{\CovOp}{\mathbf C}
\newcommand{\BlockCovOp}{\mathbf K}

\newcommand{\dAda}{d_{\mathrm{Ada}}}
\newcommand{\qLtwo}[1]{\left\lVert #1 \right\rVert_{L^2(\nuX)}}

\title{Finite-Sample Theory for Fitted Q-Iteration When Actions Are Functions}
\author{Gefei Lin\textsuperscript{1}\qquad Rui Miao\textsuperscript{2}\qquad Xiaoke Zhang\textsuperscript{1}\\[0.4em]
  \normalfont\small\textsuperscript{1}Department of Statistics, The George Washington University\\
  \normalfont\small\textsuperscript{2}Department of Mathematical Sciences, The University of Texas at Dallas}
\date{}
\begin{document}
\addtocontents{toc}{\protect\setcounter{tocdepth}{-1}}
\maketitle
\pagestyle{fancy}
\input{sections/00_abstract.tex}
\input{sections/01_introduction.tex}
\input{sections/02_problem_setup.tex}

\input{sections/03_coverage.tex}
\input{sections/04_value_guarantee.tex}
\input{sections/05_experiments.tex}
\input{sections/06_discussion.tex}

\label{main-text-end}
\clearpage
\flushbottom
\input{sections/07_statements.tex}

\bibliography{references}
\bibliographystyle{fqi_preprint}
\newpage
\input{appendix.tex}

\end{document}

%% file: sections/00_abstract.tex
\begin{abstract}
Offline reinforcement learning seeks optimal decision rules from previously collected data.
In some applications, a decision can be an entire function, such as a fluence map in radiation therapy or a smooth movement trajectory in robotics.
In this paper we study the finite-sample theory for fitted Q-iteration (FQI) with functional actions in a discounted infinite-horizon setting.
Three major difficulties arise in this setting: first, the absence of a Lebesgue probability density for functional actions complicates coverage descriptions; second, conventional coverage requirements can be restrictive; and third, the large functional action space makes greedy optimization in FQI challenging.
To address these difficulties, we study smoothness-regularized policy search under a critic-relative coverage condition.
This condition measures how well logged data distinguish relevant action-value differences without requiring an action density.
Our main theorem gives finite-sample guarantees for learned-policy regret relative to the best value within a fixed smooth class of functional-action policies.
The results allow trajectory lengths to be either bounded or growing and Q-functions to be fitted by either functional-input kernel ridge regression or adaptive functional neural networks. 
For a few examples, we can obtain polynomially decaying regret bounds in the number of logged transitions, up to logarithmic factors, with logarithmically many FQI iterations.
Numerical experiments show gains of learned functional-action policies over constant-action policies and support our adoption of 
a critic-relative coverage condition. 
\end{abstract}

%% file: sections/01_introduction.tex
\section{Introduction}

Reinforcement learning (RL) learns policies that map states to actions to maximize cumulative reward \citep{SuttonBarto2018}.
When new interactions are costly, one may use offline RL to learn optimal policies from fixed trajectories generated by earlier decision rules, known as behavior policies \citep{LevineEtAl2020OfflineRL,MiaoQiZhang2022OPE}. The action space determines what a decision can express. In the RL literature, the action is typically a categorical, discrete, or continuous variable. 
Recently in some applications, it is required to choose
an entire function as the action at each decision time, 
such as a radiation-intensity fluence map \citep{RomeijnEtAl2003} or a smooth robot trajectory \citep{GasparettoEtAl2012}. Such decision is referred to as a \emph{functional action}, of which 
argument is distinct from the sequence of RL decisions. RL work for functional actions is extremely limited. \citet{PanEtAl2018} study function-valued controls for partial differential equations in an online RL setting. For offline RL, \citet{LinEtAl2026} integrate functional actions into fitted Q-iteration (FQI), a popular offline RL method that learns the optimal policy 
by repeated regression on logged transitions \citep{ErnstEtAl2005}, but they do not provide any theoretical guarantee.

In this paper, we fill this theoretical gap and study the finite-sample theory for FQI with functional actions in a discounted infinite-horizon setting. %
A common challenge in offline FQI theory is to transfer control of critic error from the logged state--action distribution to actions selected during policy improvement. Establishing such a transfer typically requires a coverage condition linking the behavior distribution to policy-selected actions. For example, for finite-dimensional continuous actions, \citet{AntosEtAl2007} obtain this link by combining smoothness and regular action geometry with a behavior density bounded away from zero. However, this density-based argument does not extend directly to functional actions since there is no Lebesgue probability density for random functions \citep{DelaigleHall2010}, neighborhood probabilities can decay faster than any power of their radius \citep{Mas2012SmallBall}, and smooth policy actions may lie a positive distance from the behavior support. Consequently, it does not by itself yield the error transfer needed in our setting.

To address this issue, we introduce a coverage alternative for functional-action FQI, motivated by the idea in earlier work that coverage should be measured relative to the error class relevant to the learner. \citet{UeharaSun2022PartialCoverage} study model-based pessimistic learning and estimate an entire transition model to transfer transition-model errors, whereas our model-free FQI procedure fits critics through repeated Bellman regressions. \citet{XieEtAl2021BellmanPessimism} formulate function-class-relative coverage through Bellman residuals in finite state and action spaces. By restricting attention to Bellman residuals relevant to the learner, their coefficient is no larger than the corresponding state--action density ratio and can be much smaller. Because its definition does not depend on action geometry, it can also be stated formally for functional actions. However, for functional-action FQI, we also seek to quantify how smoothness of the selected action affects the transfer of critic error. Therefore, we develop coverage conditions that connect this smoothness, controlled by a reproducing kernel Hilbert space (RKHS) penalty on the policy, to critic-specific information in the logged data.
These conditions further yield a uniform H\"older bound from logged prediction error to critic error at policy-selected actions, even when the corresponding uniform linear error ratio is infinite.

Building on this transfer, we combine continuous-action policy-search analysis \citep{AntosEtAl2007} with error-propagation arguments \citep{MunosSzepesvari2008} to establish a finite-sample bound on the value gap of the policy learned by functional-action FQI. The bound accounts separately for policy-class approximation error, critic error, policy-improvement error, and finite-iteration error, while allowing reuse of the same data and dependence within trajectories. We further derive convergence rates for the learned policy's value gap when the critics are fitted by either functional-input kernel ridge regression \citep{WangWongZhangChan2026FFTEE} or adaptive functional neural networks \citep{YaoMuellerWang2021AdaFNN}, under the respective coverage and complexity conditions, including an extension to continuous state spaces. A running two-state example verifies the conditions throughout the analysis and, with suitable tuning and logarithmically many iterations, yields a value-regret rate of $(NT)^{-2/9}$ up to logarithmic factors for $N$ independent subjects with $T$ transitions each, despite a positive distance between policy-selected actions and the behavior support. Finally, experiments based on the OpenAI Gym Pendulum benchmark \citep{BrockmanEtAl2016} compare FQI with functional and constant actions and examine how action proximity relates to critic-relative diagnostics.

%% file: sections/02_problem_setup.tex
\section{Problem setup for functional-action fitted Q-iteration}\label{sec:main-method}\label{sec:main-preliminaries}
Consider a discounted, time-homogeneous Markov decision process \citep{Puterman1994} over the infinite horizon $t=0,1,\ldots$, with discount factor $0<\gamma<1$, Borel state space $\mathcal S$, and Borel action space $\mathcal A\subset L^2([0,1])$. Each $a=a(\cdot)\in\mathcal A$ is an action function, and $a(u)$ is its value at $u\in[0,1]$. We write $\Penv(\cdot,\cdot\mid s,a)$ for the conditional distribution of the reward $R_t\in[0,1]$ and next state $S_{t+1}$ given $S_t=s$ and $A_t=a$, and $\Pnext(\cdot\mid s,a)$ for its next-state marginal.
We analyze measurable stationary deterministic policies, represented by maps $\pi:\mathcal S\to\mathcal A$, so that $\pi(s)$ is the entire action function selected whenever state $s$ is visited. This choice is motivated by discounted optimality theory that a measurable selector attaining the optimal Bellman maxima defines a stationary deterministic optimal policy \citep[Section~6.2]{Puterman1994}. 

For initial-state law $\rho$, define the value of a given $\pi$ as
$J(\pi)=\E^\pi_{S_0\sim\rho}\sum_{t\ge0}\gamma^tR_t$, which is upper bounded by $\Vmax=(1-\gamma)^{-1}$ as $R_t\in [0,1]$.
Let $\mathcal H$ be a separable RKHS of measurable maps from $\mathcal S$ to $L^2([0,1])$, with operator-valued kernel $\Gamma:\mathcal S\times\mathcal S\to\mathcal B\{L^2([0,1])\}$ \citep{MicchelliPontil2005}, where $\mathcal B(E)$ denotes the bounded linear operators on $E$.
We work with the feasible policy class and value benchmark
\begin{equation}
 \Pi=\{\pi\in\mathcal H:\pi(s)\in\mathcal A\ \text{for all }s\},
 \qquad J^\star=\sup_{\pi\in\Pi}J(\pi).
 \label{eq:main-policy-class}
\end{equation}
We observe $\mathcal D=\{(S_{i,t},A_{i,t},R_{i,t},S_{i,t+1}):i\le N,t\le T_i\}$ from $N$ independent subjects, each contributing a trajectory of length $T_i$. For the main analysis, without loss of generality, we assume $T_i=T=T_N$, so
the total number of transitions is $n=NT$. See Appendix~\ref{sec:ti-cluster} for handling unequal trajectory lengths.
We aim to develop the finite-sample bound for $J^\star-J(\widehat\pi_M)$ where $\widehat\pi_M$ is obtained by Algorithm~\ref{alg:ffqi-main}. 

Algorithm~\ref{alg:ffqi-main} is an extension of standard FQI \citep{ErnstEtAl2005} to functional actions. 
The key feature of Algorithm~\ref{alg:ffqi-main} is that its policy-improvement step adopts the average-Q objective for continuous actions \citep{AntosEtAl2007} and additionally regularizes policy smoothness through the RKHS-norm penalty $\lambda_{\pi,n}\|\pi\|_{\mathcal H}^2$.
In particular, if the zero policy $\pi_0(s)\equiv0$ is feasible and the approximate maximization has regularized-objective error at most $\bar\tau_{\pi,n}$, clipping the critic to $[0,\Vmax]$ gives
$\|\widehat\pi_m\|_{\mathcal H}^2
 \le (\Vmax+\bar\tau_{\pi,n})/{\lambda_{\pi,n}}$.
Under the spectral evaluation condition in Section~\ref{sec:main-spectral}, this norm control also bounds smoothness within each selected action. 

\begin{wrapfigure}[17]{r}{0.62\textwidth}
\vspace{-0.8cm}
\begin{minipage}{\linewidth}
\begin{algorithm}[H]
\caption{FQI with functional actions}\label{alg:ffqi-main}\label{alg:general-functional-fqi}
\begin{algorithmic}[1]
\Require Data $\mathcal D$, action-feasible policy domain $\Pi$, critic fitting rule, coefficient $\lambda_{\pi,n}>0$, discount $\gamma$, iterations $M$.
\State Set $\widehat Q_0=0$; obtain $\widehat\pi_0$ by the update in step~\ref{line:main-policy-update}.
\For{$m=1,\ldots,M$}
\State Set $Y_{i,t,m}=R_{i,t}+\gamma\widehat Q_{m-1}\{S_{i,t+1},\widehat\pi_{m-1}(S_{i,t+1})\}$.\label{line:main-bellman-label}
\State Regress $Y_{i,t,m}$ on $(S_{i,t},A_{i,t})$ to obtain $\widetilde Q_m$; set $\widehat Q_m=\clip_{[0,\Vmax]}(\widetilde Q_m)$.\label{line:main-critic-fit}
\State $\displaystyle \widehat\pi_m\gets\operatorname*{arg\,max}_{\pi\in\Pi}
\begin{aligned}[t]
&n^{-1}\sum_{i,t}\widehat Q_m\{S_{i,t+1},\pi(S_{i,t+1})\}\\
&\phantom{n^{-1}}-\lambda_{\pi,n}\|\pi\|_{\mathcal H}^2
\end{aligned}$\label{line:main-policy-update}
\EndFor
\State \Return $\widehat\pi_M$.
\end{algorithmic}
\end{algorithm}
\end{minipage}
\end{wrapfigure}
The penalty-induced bound above places each returned policy in an RKHS norm ball. Sections~\ref{sec:main-spectral} and~\ref{sec:main-guarantees} use a fixed such ball for the coverage and value analyses. For $B<\infty$, define
\[
 \Pi_{\mathcal H}(B)=\{\pi\in\Pi:\|\pi\|_{\mathcal H}\le B\}.
\]

Appendix~\ref{lem:rev-calibration} gives sufficient conditions for one fixed radius $B=B_{\mathcal H}<\infty$ to contain every returned policy throughout the iterations.

The best value in the fixed norm ball is $J_{\mathcal H}^\star=\sup_{\pi\in\Pi_{\mathcal H}(B_{\mathcal H})}J(\pi)$. Its approximation gap relative to the best feasible-policy value $J^\star$ is $a_{\mathcal H}=J^\star-J_{\mathcal H}^\star\ge0$.
Then we can decompose $J^\star-J(\widehat\pi_M) =a_{\mathcal H}+\{J_{\mathcal H}^\star-J(\widehat\pi_M)\}$. Based on the coverage condition in Section~\ref{sec:main-spectral}, we develop a bound for $J_{\mathcal H}^\star-J(\widehat\pi_M)$ in Section~\ref{sec:main-guarantees}, with $a_{\mathcal H}$ accounting for restriction to the fixed RKHS ball.

%% file: sections/03_coverage.tex
\section{A coverage alternative for functional actions}\label{sec:main-spectral}
Coverage is a central difficulty for functional-action FQI. Without proper conditions, a small prediction error on logged data does not necessarily imply a small error at the actions selected by a learned policy. However, standard action-density and neighborhood conditions are difficult to apply to functional-action FQI because there exists no Lebesgue probability density for functional actions and actions can easily lie outside the behavior support. In this section, we circumvent this difficulty by adopting critic-relative coverage conditions that combine directional information in logged data with policy smoothness to control evaluation error.

\subsection{From logged prediction to policy evaluation}\label{subsec:main-transfer-norms}
FQI analyses commonly distinguish the distribution used for critic regression from the state distribution used for policy evaluation \citep{AntosEtAl2007,MunosSzepesvari2008}. Accordingly, let $\nu_X$ be the average logged state--action law and $\nu_+$ be the average next-state law respectively:
\[\nu_X=\frac1{NT}\sum_{i=1}^N\sum_{t=1}^T\law(S_{i,t},A_{i,t}),
 \qquad \nu_+=\frac1{NT}\sum_{i=1}^N\sum_{t=1}^T\law(S_{i,t+1}).
\]
These pooled laws may depend on 
$n=NT$.
FQI fits critics under $\nu_X$ and evaluates them at actions selected by $\Pi_{\mathcal H}(B_{\mathcal H})$. We therefore compare logged prediction error with a policy-evaluation norm covering these actions.

In Algorithm~\ref{alg:ffqi-main}, for a preceding critic $Q$ and policy $\pi$, the Bellman label function is $y_{Q,\pi}(s,a,r,s')=r+\gamma Q\{s',\pi(s')\}$.
For a fixed label function $y$, its population regression target is the conditional mean $Q_y(s,a)=\int y(s,a,r,s')\Penv(dr,ds'\mid s,a)$.
With exact functional inputs, the total error of a critic $f$ fitted to labels $y$ is $f-Q_y$. Let $f_y^\circ$ be a structured approximant of $Q_y$. Let $\mathcal V_n$ be a deterministic family, possibly depending on the sample size $n$, containing all corresponding pairs of possible fitted critics and target approximants. For $v=(f,f_y^\circ)\in\mathcal V_n$, write $h_v=f-f_y^\circ$. We measure this critic difference by its \emph{logged prediction error} $e_v$ and its \emph{policy-evaluation error} $\upsilon_v$:
\begin{equation}
 e_v=\|h_v\|_{L^2(\nu_X)},\qquad
 \upsilon_v=\|h_v\|_{\nu_+,\Pi_{\mathcal H}(B_{\mathcal H}),2}
 :=\Big[\int\sup_{\pi\in\Pi_{\mathcal H}(B_{\mathcal H})}|h_v\{s,\pi(s)\}|^2\nu_+(ds)\Big]^{1/2}.
 \label{eq:main-graph-norm}
\end{equation}
Both quantities measure the fitted critic's error relative to its target approximant: $e_v$ uses the logged design, whereas $\upsilon_v$ uses policy-selected actions. The supremum in $\upsilon_v$ is taken separately at each state, so it controls a statewise error envelope. The target-approximation error $f_y^\circ-Q_y$ is controlled separately on the full state--action domain because policy actions may lie outside behavior support.

We seek conditions under which a small $e_v$ implies a small $\upsilon_v$. A uniform linear inequality $\upsilon_v\le Ce_v$ can fail, so we instead allow a H\"older inequality $\upsilon_v\le C_{\rm tr,n}e_v^\vartheta$, $0<\vartheta\le1$, which still transfers consistency under suitable control of $C_{\rm tr,n}$. In particular, when $\vartheta<1$, this requires controlling the coefficient radius of the analyzed critic differences. The mechanism combines information in the logged actions with smooth evaluation at the actions selected by the localized policy class.

\subsection{Directional information and Spectral Transfer}
For illustration, we first suppress state dependence and start with an action-linear critic difference $h(a)=\langle d,a\rangle$. Fix an orthonormal basis $\{\phi_j\}$ of $L^2([0,1])$ and let $\BaseWeightOp_\alpha\phi_j=j^{-2\alpha}\phi_j$, $\alpha>0$. Logged-design variation, measured through the second-moment condition $\mathbb E h(A)^2\ge c\|\BaseWeightOp_\alpha^{1/2}d\|^2$ for $c>0$, controls the critic coefficient in an information-weighted norm. For $\zeta>0$, we quantify action smoothness by
$
 \|a\|_{\zeta}^2:=\sum_{j\ge1}j^{2\zeta}|\langle a,\phi_j\rangle|^2.
$
The $\zeta$-norm emphasizes higher-order spectral variations. For linear features, actions controlled under this norm reduce sensitivity to critic differences along directions weakly observed under the logged design for compatible transfer exponents satisfying $\alpha\vartheta\le\zeta$. See Appendix~\ref{sec:general-oracle} for detailed explanations. 
The policy RKHS supplies action smoothness, and general critic features require the corresponding directional-information and smooth-evaluation conditions stated below.

\begin{assumption}[Spectrally regular policy evaluation]\label{ass:main-policy-regularity}
There is a measurable $\kappa_\zeta:\mathcal S\to[0,\infty)$ such that, on one common full-$\nu_+$-measure set,
\begin{equation}
 \|\pi(s)\|_\zeta\le\kappa_\zeta(s)\|\pi\|_{\mathcal H}
 \quad\text{for every }\pi\in\mathcal H,
 \quad \text{and} \quad
 K_\zeta^2:=\int\kappa_\zeta(s)^2\nu_+(ds)<\infty.
\label{eq:main-policy-spectral-evaluation}
\end{equation}
\end{assumption}
Assumption~\ref{ass:main-policy-regularity} requires the policy-RKHS norm to control smoothness within each selected action, beyond its $L^2$ magnitude. It gives the fixed comparison ball the classwise spectral envelope
\begin{equation}
 \int\sup_{\pi\in\Pi_{\mathcal H}(B_{\mathcal H})}\|\pi(s)\|_{\zeta}^2\nu_+(ds)
   \le K_\zeta^2B_{\mathcal H}^2
       =:{B}_{\Pi,\zeta}^2.
 \label{eq:main-policy-radius}
\end{equation}
The supremum is inside the state integral, so the bound covers every policy action available at each state. The $\zeta$ describes action regularity, while $B_{\mathcal H}$ controls the size of the policy class. Appendix~\ref{subsec:penalty-smooth-actions} gives weighted spectral and Sobolev RKHS constructions satisfying Assumption~\ref{ass:main-policy-regularity}.

For more general critics, the feature need not equal the action. We may associate each $v$ with a coefficient Hilbert space $\mathcal E_v$, a coefficient $d_v$, and a measurable feature map $\psi_v:\mathcal S\times\mathcal A\to\mathcal E_v$ such that
$
 h_v(s,a)=\langle d_v,\psi_v(s,a)\rangle_{\mathcal E_v}.
$
Define the feature subspace $\mathcal E_v^{\mathrm{feat}}:=\overline{\operatorname{span}}\{\psi_v(s,a):(s,a)\in\mathcal S\times\mathcal A\}\subseteq\mathcal E_v$, and let $\ProjOp_v$ be the orthogonal projection in $\mathcal E_v$ onto $\mathcal E_v^{\mathrm{feat}}$. The projected coefficient $\ProjOp_vd_v$ retains the directions that can affect the critic anywhere on the full domain, including at policy actions outside behavior support. This projection is determined by the full feature map, not the observed sample span. If the features span $\mathcal E_v$, then $\ProjOp_v$ is the identity.
Let $\WeightOp_v:\mathcal E_v\to\mathcal E_v$ be a positive, injective, self-adjoint contraction assigning weights to coefficient directions. The information condition below compares logged variation with these weights. Appendix~\ref{app:spectral-power-notation} defines the inverse powers of $\WeightOp_v$.

\begin{assumption}[Directional information and smooth evaluation]\label{ass:main-secant}
For constants $c_n>0$, $0\le B_{\psi,n}<\infty$ and $0<\vartheta\le1$, require the following for every $v\in\mathcal V_n$:
\begin{enumerate}[label=(\roman*),leftmargin=*,nosep]
\item \emph{Feature integrability.} $\int\|\psi_v(s,a)\|_{\mathcal E_v}^2\nu_X(ds,da)<\infty$.
\item \emph{Directional information.} $e_v^2\ge c_n\langle\ProjOp_vd_v,\WeightOp_v\ProjOp_vd_v\rangle$.
\item \emph{Smooth evaluation.} $\int\sup_{\pi\in\Pi_{\mathcal H}(B_{\mathcal H})}\|\WeightOp_v^{-\vartheta/2}\psi_v\{s,\pi(s)\}\|^2\nu_+(ds)\le{B}_{\psi,n}^2$.
\end{enumerate}
\end{assumption}
In Assumption~\ref{ass:main-secant}, Part~\textup{(i)} ensures finite logged feature second moments. Part~\textup{(ii)} requires logged variation to distinguish the relevant critic differences, with smaller weights allowing less information in some coefficient directions. Part~\textup{(iii)} limits how strongly policy-selected actions can amplify critic differences along those weakly observed directions. Policy smoothness alone does not verify this critic-relative condition for an arbitrary feature map: one must also establish how action smoothness controls the weighted feature norm. Appendices~\ref{app:direct-krr} and~\ref{app:ada-regression} give sufficient conditions for this verification in KRR and AdaFNN, respectively.

The restriction to relevant errors follows the model-relative coverage approach of \citet{UeharaSun2022PartialCoverage}, while \citet[Definition~1]{XieEtAl2021BellmanPessimism} use ratios of squared Bellman-residual norms. Here the errors are differences between fitted critics and their target approximants. Coverage is uniform over this declared error family, and evaluation takes a statewise supremum over the comparison policy ball. The resulting H\"older transfer can remain useful when the uniform linear error ratio is infinite. In continuous-action FQI, \citet[Assumptions~A2 and A6--A8]{AntosEtAl2007} combine a positive behavior-density lower bound with regular action geometry and Lipschitz smoothness. Instead, our conditions in Assumption~\ref{ass:main-secant} require directional information and smooth evaluation, allowing policy actions outside behavior support. Appendix~\ref{app:related-work} gives a detailed comparison.

We now state the main transfer result in Proposition~\ref{thm:main-spectral}, which bounds policy-evaluation error by logged prediction error. 
\begin{proposition}[Spectral transfer]\label{thm:main-spectral}
Under Assumption~\ref{ass:main-secant},
 $\upsilon_v\le{B}_{\psi,n}c_n^{-\vartheta/2} \|d_v\|^{1-\vartheta}e_v^\vartheta$.
\end{proposition}
The proof of Proposition~\ref{thm:main-spectral} is given in Appendix~\ref{sec:general-oracle}. 
Proposition~\ref{thm:main-spectral} ensures that a small logged prediction error leads to a small error at smooth policy actions. The exponent $\vartheta$ describes the loss of accuracy in this transfer: $\vartheta=1$ preserves the prediction-error order, while $0<\vartheta<1$ gives a bound that decays more slowly. For a fixed weighting operator, stronger logged information (larger $c_n$) and less amplification by policy-selected features (smaller $B_{\psi,n}$) tighten the bound. When $\vartheta<1$, the transfer bound also depends on the size of the critic-difference coefficients. Despite this slower decay, the H\"older bound can still support consistent evaluation without nearby logged actions or a finite uniform linear error ratio. 
Next we provide an example that 
admits H\"older transfer even though all actions selected by the localized policy class lie outside the support of the logged actions.

\begin{example}[Two-state functional-action support gap]\label{ex:two-state-control}
Let $\mathcal S=\{1,2\}$ and $\mathcal A=\{a:\|a\|_{L^2}\le{B}_A\}$. Take $\mathcal G=H^2(0,1):=\{a\in L^2(0,1):a',a''\in L^2(0,1)\}$ with $\|a\|_{\mathcal G}^2=\|a\|_{L^2}^2+\|a''\|_{L^2}^2$, and let $\mathcal H=\mathcal G\oplus\mathcal G$ with $\|\pi\|_{\mathcal H}^2=\|\pi(1)\|_{\mathcal G}^2+\|\pi(2)\|_{\mathcal G}^2$. The feasible domain remains the product action ball $\Pi=\{\pi\in\mathcal H:\|\pi(s)\|_{L^2}\le B_A,\ s=1,2\}$. For the normalized shifted Legendre basis, the Sobolev embedding in Proposition~\ref{prop:curvature-spectral} in Appendix~\ref{sec:bspline-approximation} verifies~\eqref{eq:main-policy-spectral-evaluation} with $\zeta=2$, and every action selected by $\Pi_{\mathcal H}(B_{\mathcal H})$ has spectral norm at most $C B_{\mathcal H}$.
At each decision time, logged actions are independently sampled as
$A=\varsigma_A\sum_{j\ge1}j^{-9/4}\xi_j\phi_j$, where $\xi_j$ are independent symmetric signs and $\varsigma_A$ is small enough for feasibility. We also set $\nu_X=\nu_+\otimes\law(A)$ with $\nu_+$ uniform on the two states and adopt the functional linear critic kernel 
$K\{(s,a),(s',a')\}=\ind\{s=s'\}(1+\langle a,a'\rangle)$.
This model satisfies Assumption~\ref{ass:main-secant} with $\vartheta=8/9$. The logged actions lack order-two spectral smoothness, with an infinite $\zeta$-norm at $\zeta=2$, and their support lies a positive $L^2$ distance from all actions selected by $\Pi_{\mathcal H}(B_{\mathcal H})$. 
For the bounded critic differences $\Delta_j(s,a)=\kappa_0\ind\{s=1\}\langle a,\phi_j\rangle$ with fixed sufficiently small $\kappa_0>0$, 
we have 
$e_j\asymp j^{-9/4}$ and 
$\upsilon_j\asymp j^{-2}\asymp e_j^{8/9}$. 
Hence logged-design consistency still implies policy-evaluation consistency through Proposition~\ref{thm:main-spectral}, even though every uniform linear transfer coefficient is infinite. The attained $8/9$ exponent also shows that a uniformly faster transfer rate requires stronger information or smoothness conditions.
See Appendix~\ref{app:verification-examples} for covariance, smoothness and support-gap verifications and Section~\ref{sec:main-guarantees} for rewards and transitions for value analysis. 
\end{example}

\subsection{Transfer at a prescribed accuracy}\label{subsec:main-accuracy-transfer}
A large error ratio is not necessarily problematic when policy-evaluation error is already below the desired accuracy. For $\varepsilon>0$, define amplification
\begin{equation}
 C_n\{\varepsilon;\Pi_{\mathcal H}(B_{\mathcal H})\}
 =\sup_{v\in\mathcal V_n:\upsilon_v>\varepsilon}\upsilon_v^2/e_v^2,
 \label{eq:main-accuracy-transfer}
\end{equation}
so that $\upsilon_v\le \varepsilon+\sqrt{C_n\{\varepsilon;\Pi_{\mathcal H}(B_{\mathcal H})\}}\,e_v$. In this definition, 
a positive number divided by zero is taken as $+\infty$ and an empty supremum is taken as zero.

For a generic logged prediction bound $\delta\ge0$, define
\begin{equation} \label{eq:main-accuracy-transfer-function}
 \omega_n(\delta)=\inf_{\varepsilon>0:\,C_n\{\varepsilon;\Pi_{\mathcal H}(B_{\mathcal H})\}<\infty}
       [\varepsilon+\sqrt{C_n\{\varepsilon;\Pi_{\mathcal H}(B_{\mathcal H})\}}\,\delta].
\end{equation}
Thus $e_v\le\delta$ implies $\upsilon_v\le\omega_n(\delta)$.
When $\omega_n(\delta)\lesssim\delta^{\vartheta_*}$ uniformly in $n$, write $\vartheta_*\in(0,1]$ for the resulting transfer exponent. With uniformly bounded transfer factors in Proposition~\ref{thm:main-spectral}, one may take $\vartheta_*=\vartheta$. Proposition~\ref{thm:smooth-spectral-transfer} in Appendix~\ref{sec:general-oracle} gives an upper bound for $C_n\{\varepsilon;\Pi_{\mathcal H}(B_{\mathcal H})\}$ and extends the transfer result to accuracy-dependent directional information.

%% file: sections/04_value_guarantee.tex
\section{Finite-sample value guarantee}\label{sec:main-guarantees}
Section~\ref{sec:main-spectral} provides a transfer control of logged prediction error to critic error at smooth policy actions. In this section we combine this transfer with policy-improvement error and state propagation to bound the regret of Algorithm~\ref{alg:ffqi-main}.

\subsection{Assumptions and error events}\label{subsec:one-step-errors}
Before observing the data, we specify a clipped critic class $\mathcal Q_n$ and a label class $\mathcal Y_n$ containing every Bellman label $y_{Q,\pi}$ with $Q\in\mathcal Q_n$ and $\pi\in\Pi_{\mathcal H}(B_{\mathcal H})$. The deterministic family $\mathcal V_n$ of pairs of fitted critics and target approximants is defined in Section~\ref{subsec:main-transfer-norms}.

\paragraph{Uniform error event.}
For each $y\in\mathcal Y_n$, let $\widehat f_y$ be the clipped output of the regression rule in step~\ref{line:main-critic-fit} of Algorithm~\ref{alg:ffqi-main}, applied to labels $y$ on $\mathcal D$. For the transfer analysis, take $F_y$ to be the raw fit for KRR and the clipped fit evaluated with exact functional scores for AdaFNN. Here $\clip$ acts pointwise on critic outputs, with $\clip(z)=\min\{\Vmax,\max\{0,z\}\}$, so $\widehat f_y-\clip F_y$ measures numerical input error. For $Q\in\mathcal Q_n$, let $\widehat\pi_Q$ be its policy update from step~\ref{line:main-policy-update}, and define its empirical and population objectives by
\[
 \widehat\Phi_Q(\pi)=\frac1n\sum_{i,t}Q\{S_{i,t+1},\pi(S_{i,t+1})\},\qquad
 \Phi_Q(\pi)=\E_{\nu_+}Q\{S,\pi(S)\}.
\]
The regularized objective is $\widehat\Psi_Q(\pi)=\widehat\Phi_Q(\pi)-\lambda_{\pi,n}\|\pi\|_{\mathcal H}^2$. 
Let $\delta_n$, $e_{Q,n}$, $u_n$, and $\tau_{\pi,n}$ be nonnegative deterministic bounds for prediction relative to the target approximant, numerical input error, objective sampling error, and policy optimization error, respectively. Define the joint event
\begin{equation}
 \mathcal U_n=\left\{
 \begin{aligned}
 &(F_y,f_y^\circ)\in\mathcal V_n,\quad
 \|F_y-f_y^\circ\|_{L^2(\nu_X)}\le\delta_n,
 \quad\forall y\in\mathcal Y_n,\\[2pt]
 &\|\widehat f_y-\clip F_y\|_\infty\le e_{Q,n}
 \quad\forall y\in\mathcal Y_n,\\[2pt]
 &\sup_{Q\in\mathcal Q_n,\,\pi\in\Pi_{\mathcal H}(B_{\mathcal H})}
 |\widehat\Phi_Q(\pi)-\Phi_Q(\pi)|\le u_n,\\[2pt]
 &
 \sup_{Q\in\mathcal Q_n, \pi\in\Pi_{\mathcal H}(B_{\mathcal H})}\widehat\Psi_Q(\pi)
       -\widehat\Psi_Q(\widehat\pi_Q)\le\tau_{\pi,n}
 \end{aligned}
 \right\}.
 \label{eq:main-uniform-error-event}
\end{equation}
Each supremum norm above is taken over the full state--action domain, and $\tau_{\pi,n}$ includes representation and solver errors. For the prescribed, possibly $n$-dependent iteration count $M=M_n$, define the membership and localization event $\mathcal R_n=\bigcap_{m=0}^{M} \{\widehat Q_m\in\mathcal Q_n,\ \widehat\pi_m\in\Pi_{\mathcal H}(B_{\mathcal H})\}$ and the joint analysis event $\mathcal E_n=\mathcal R_n\cap\mathcal U_n$.
On $\mathcal E_n$, every realized Bellman label lies in $\mathcal Y_n$, and its fitted critic and target approximant satisfy $e_v\le\delta_n$ and hence $\upsilon_v\le\omega_n(\delta_n)$. The uniform bounds permit substitution of all fitted iterates despite data reuse. 
See the full argument in Appendix~\ref{sec:samedata-fqi}. 

For a critic $Q$ and policy $\pi$, define the policy Bellman operator by
$[\mathbb B^\pi Q](s,a)=\E[R+\gamma Q\{S',\pi(S')\}\mid s,a]$.
At iteration $m$, define the statewise greedy defect
$g_m(s)=\sup_{\pi\in\Pi_{\mathcal H}(B_{\mathcal H})}\widehat Q_m\{s,\pi(s)\}-\widehat Q_m\{s,\widehat\pi_m(s)\}$.
It measures how much critic value is lost by using the returned policy at state $s$. To bound the final value gap, we require uniform control of critic-update and policy-improvement errors across iterations.

\begin{assumption}[One-step accuracy]\label{ass:main-one-step}
Let $\varepsilon_{Q,n}$ and $\varepsilon_{\pi,n}$ be bounds on critic-update error and policy-improvement error, respectively. On $\mathcal E_n$, we assume
$$
 \max_{1\le m\le M}\|\widehat Q_m-\mathbb B^{\widehat\pi_{m-1}}\widehat Q_{m-1}\|_{\nu_+,\Pi_{\mathcal H}(B_{\mathcal H}),2}
 \le\varepsilon_{Q,n},\quad \text{and} \quad \max_{0\le m\le M}\|g_m\|_{L^2(\nu_+)}\le\varepsilon_{\pi,n}.
$$
\end{assumption}
Assumption~\ref{ass:main-one-step} bounds the critic-update and policy-improvement errors uniformly over FQI iterations.
The prediction and numerical-input bounds in~\eqref{eq:main-uniform-error-event}, together with target approximation and coverage, supply $\varepsilon_{Q,n}$ while the objective sampling and optimization bounds supply $\varepsilon_{\pi,n}$. Both conversions hold on the same $\mathcal E_n$. Appendices~\ref{subsec:regression-policy-transfer} and~\ref{sec:policy-update} give the details and localization conditions.

To turn these errors under $\nu_+$ into a value guarantee, we must also control how they propagate to the states visited by the returned policy.
To describe state visitation, let $\law_\rho^\pi(S_t)$ be the distribution of $S_t$ starting from $S_0\sim\rho$ and following $\pi$ under the transition kernel $\Pnext$. The discounted state occupancy is the probability law
$d_\rho^\pi=(1-\gamma)\sum_{t\ge0}\gamma^t\law_\rho^\pi(S_t)$.
Define the pointwise Bellman envelope by
$[\mathbb BQ](s,a)=\E[R+\gamma\sup_{\pi\in\Pi_{\mathcal H}(B_{\mathcal H})}Q\{S',\pi(S')\}\mid s,a]$, and let $Q^\dagger$ be its bounded fixed point. This envelope supplies an upper value comparator for every policy in the comparison class; it need not be attained by one policy in that coupled class.

\begin{assumption}[State transfer]\label{ass:main-state-transfer}
For $Q,Q'\in\mathcal Q_n\cup\{Q^\dagger\}$, let $\bar\pi$ be any reference policy in $\Pi_{\mathcal H}(B_{\mathcal H})$. Define the nonnegative state functions
 $h_{Q,Q'}(s)=\sup_{\pi\in\Pi_{\mathcal H}(B_{\mathcal H})}|Q\{s,\pi(s)\}-Q'\{s,\pi(s)\}|$,
 $g_{Q,\bar\pi}(s)=\sup_{\pi\in\Pi_{\mathcal H}(B_{\mathcal H})}Q\{s,\pi(s)\}-Q\{s,\bar\pi(s)\}$.
For a state function $h$, write $(\Pnext h)(s,a)=\int h(s')\Pnext(ds'\mid s,a)$. There are constants $0<C_B,C_{\rm occ}(\pi)<\infty$ such that the following bounds hold for every $h=h_{Q,Q'}$ and every $h=g_{Q,\bar\pi}$, over all the stated choices of $Q,Q',\bar\pi$:
\begin{align}
 \E_{d_\rho^\pi}h&\le C_{\rm occ}(\pi)\|h\|_{L^2(\nu_+)},
 \label{eq:main-state-occupancy}\\
 \int\sup_{\pi'\in\Pi_{\mathcal H}(B_{\mathcal H})}|\Pnext h\{s,\pi'(s)\}|^2\nu_+(ds)
 &\le C_B^2\|h\|_{L^2(\nu_+)}^2.
 \label{eq:main-state-propagation}
\end{align}
The occupancy bound~\eqref{eq:main-state-occupancy} must hold for every possible returned policy $\pi$. It controls the average error along that policy's discounted state visits by its $L^2(\nu_+)$ size.
\end{assumption}
In Assumption~\ref{ass:main-state-transfer}, the function $h_{Q,Q'}$ measures the largest absolute critic difference among policy-selected actions at each state, and $g_{Q,\bar\pi}$ measures the loss from using the reference policy $\bar\pi$ relative to the statewise supremum. Taking $Q=\widehat Q_m$ and $\bar\pi=\widehat\pi_m$ recovers $g_m$. The propagation bound~\eqref{eq:main-state-propagation} limits how much one state transition can amplify either error, with factor $C_B$.
Eq.~\eqref{prop:coverage-domination} in Appendix~\ref{sec:bellman} gives state-density conditions sufficient for both~\eqref{eq:main-state-occupancy} and~\eqref{eq:main-state-propagation}, and hence for Assumption~\ref{ass:main-state-transfer}, as in concentrability-based FQI analyses \citep{MunosSzepesvari2008,FarahmandMunosSzepesvari2010}.

\subsection{Main value guarantee and examples}\label{subsec:value-propagation}\label{subsec:critic-specializations}
We first establish a finite-sample bound on the value gap $J^\star-J(\widehat\pi_M)$ and then derive convergence rates for three examples.

\begin{theorem}[Finite-sample value guarantee under stable propagation]\label{thm:main-value}
Under the data and measurability conditions above as well as Assumptions~\ref{ass:main-one-step} and~\ref{ass:main-state-transfer}, suppose that $C_{\rm occ}$ is uniformly bounded over all possible returned policies and $\gamma C_B < 1$ by a fixed positive margin independent of $N$, $T$, and $M$. There exists a constant $C$, which is independent of $N$, $T$, and $M$, such that after $M\ge1$ FQI iterations, on $\mathcal E_n$, we have
\begin{equation}
 J^\star-J(\widehat\pi_M)
 \le a_{\mathcal H}+C\Big\{
 (\gamma C_B)^M
 +\varepsilon_{Q,n}+\varepsilon_{\pi,n}\Big\}.
 \label{eq:main-value-stable}
\end{equation}
\end{theorem}
Theorem~\ref{thm:main-value} separates the bound of $J^\star-J(\widehat\pi_M)$ into four components.
The approximation gap $a_{\mathcal H}$ defined in Section \ref{sec:main-method} is the value lost by restricting policies to the fixed RKHS ball. The bounds $\varepsilon_{Q,n}$ and $\varepsilon_{\pi,n}$ defined in Assumption~\ref{ass:main-one-step} measure the accuracy of critic fitting and policy improvement.
The finite-iteration error $(\gamma C_B)^M$ bounds the remaining influence of the initial critic $\widehat Q_0=0$. It decreases geometrically with the number of FQI updates: $\gamma$ discounts future errors while $C_B$ from Assumption~\ref{ass:main-state-transfer} controls their propagation through one state transition. Keeping $\gamma C_B$ uniformly below one also prevents repeated critic and policy errors from accumulating without bound. Once the finite-iteration error falls below $\varepsilon_{Q,n}+\varepsilon_{\pi,n}$, further iterations have little effect on the guarantee. If $\Pr(\mathcal E_n)\to1$, vanishing one-step errors and increasing $M$ yield convergence in probability to the best value within the fixed smooth class. Theorem~\ref{thm:finite-value-general} in Appendix~\ref{sec:finite-sample} gives the general finite-$M$ bound for every finite $\gamma C_B$ and the full proof.

To derive convergence rates in terms of $n=NT$, we further impose a geometric $\beta$-mixing condition on the complete transition records $Z_{i,t}=(S_{i,t},A_{i,t},R_{i,t},S_{i,t+1})$. Let $\beta_{\rm mix}(k)$ bound their lag-$k$ absolute-regularity coefficients uniformly over subjects and sample sizes, as defined in Appendix~\ref{subsec:mix-definition}.

\begin{assumption}[Geometric mixing of transition records]\label{ass:main-mixing}
The $N$ subjects are independent, and there are constants $C_{\rm mix},c_{\rm mix}>0$, such that
 $\beta_{\rm mix}(k)\le C_{\rm mix}e^{-c_{\rm mix}k}$, $1\le k<T$.
\end{assumption}
Assumption~\ref{ass:main-mixing} controls complete transition records uniformly over subjects and sample sizes, while permitting time-varying and subject-specific marginals.
In discounted infinite-horizon RL, exponential $\beta$-mixing supports finite-sample analyses of fitted Q-iteration \citep[Assumption~A2]{AntosEtAl2007} and fitted policy iteration \citep[Assumption~2]{AntosSzepesvariMunos2008}. Related geometric-ergodicity conditions are used for finite-time temporal-difference learning \citep[Assumption~1]{BhandariRussoSingal2021} and value-function inference when trajectory length grows \citep[Assumption~A3(ii)]{ShiEtAl2022InfiniteHorizon}. 

The following three examples impose Assumption~\ref{ass:main-mixing} and the premises of Theorem~\ref{thm:main-value}, with $\Pr(\mathcal E_n)\to1$ and $M\ge c\log(NT)$ for a sufficiently large fixed constant $c$. Appendix~\ref{subsec:mix-rate} gives the probability bounds and logarithmic tuning factors underlying the $NT$-based rates below. The notation $\widetilde O$ suppresses logarithmic factors in $n=NT$.

\begin{examplecontinued}
In the two-state model with functional linear KRR, coverage, fixed-radius policy localization, one-step accuracy, and stable state propagation are jointly verified under the stated tuning in Appendix~\ref{app:verification-examples}. Since the fixed policy ball contains a globally optimal policy, we have $a_{\mathcal H}=0$ and
$J^\star-J(\widehat\pi_M)=\widetilde O_p((NT)^{-2/9})$.
Thus the returned policy is value-consistent despite a positive action-support gap and an infinite uniform linear error ratio for the analyzed critic differences. See Appendix~\ref{subsec:mix-scope} for verifications of the temporal conditions.
\end{examplecontinued}

The following KRR and AdaFNN examples rely on the stated approximation, information, localization, and optimization conditions. We assume $\omega_n(\delta)\lesssim\delta^{\vartheta_*}$ uniformly over the critic classes, with $\vartheta_*$ defined in Section~\ref{sec:main-spectral}. Define the policy-compatibility gap by
 $\Delta(Q)=\E_{\nu_+}\sup_{\pi\in\Pi_{\mathcal H}(B_{\mathcal H})}Q\{S,\pi(S)\}
 -\sup_{\pi\in\Pi_{\mathcal H}(B_{\mathcal H})}\E_{\nu_+}Q\{S,\pi(S)\}$, and let $\Delta_n=\sup_{Q\in\mathcal Q_n}\Delta(Q)$.
The gap vanishes when each covered critic has a measurable statewise maximizing policy in the comparison ball. Appendix~\ref{sec:policy-update} gives the decomposition.

\begin{example}[Functional-input kernel ridge regression]\label{subsec:main-krr}
KRR \citep{WangWongZhangChan2026FFTEE} uses exact functional inputs and a bounded kernel, with Bellman targets in a fixed bounded critic-RKHS ball. Suppose that the full deterministic Bellman-label and policy-objective log covering numbers are $\widetilde O((NT)^{p_K})$ with $0\le p_K<1$, at the uniform-norm covering resolutions specified in Appendix~\ref{subsec:krr-value-rate}. Then we have 
\[
 J^\star-J(\widehat\pi_M)
 \le a_{\mathcal H}+\widetilde O_p\!\left(
 (NT)^{-\vartheta_*(1-p_K)/4}+\sqrt{\Delta_n}\right).
\]
For logarithmic covering complexity $p_K=0$ and $\Delta_n=0$, the excess over $a_{\mathcal H}$ is $\widetilde O_p((NT)^{-\vartheta_*/4})$.
\end{example}

\begin{example}[Adaptive functional neural networks]\label{subsec:main-joint-rate}
AdaFNN learns $d$ functional scores and an outer neural network \citep{YaoMuellerWang2021AdaFNN}. Here we focus on finitely many states and $\mathcal H=\bigoplus_{s\in\mathcal S}H^2(0,1)$. Let $s_h\ge1$ and $s_b>0$ be the smoothness orders of the outer Bellman-target functions and their score representers. At approximation accuracy $\varepsilon_{\mathrm{app}}>0$, assume the entire fitted class has action-Lipschitz constants $O(\varepsilon_{\mathrm{app}}^{-p_{\rm Lip}})$, $p_{\rm Lip}\ge0$, and define
$D_{\rm Ada}=\max\{d/s_h,\ 1/s_b,\ 1+p_{\rm Lip}/2\}$.
Under the conditions in Appendices~\ref{sec:structural-rates} and~\ref{subsec:mix-rate}, we have 
\[
 J^\star-J(\widehat\pi_M)
 \le a_{\mathcal H}+\widetilde O_p\!\left(
 (NT)^{-\min\{\vartheta_*,1/2\}/(2+D_{\rm Ada})}+\sqrt{\Delta_n}\right).
\]
For the one-score case with first-order smoothness and uniformly bounded action-Lipschitz constants, $d=s_h=s_b=1$ and $p_{\rm Lip}=0$. If $\vartheta_*\ge1/2$ and $\Delta_n=0$, the excess over $a_{\mathcal H}$ is $\widetilde O_p((NT)^{-1/6})$.
Appendix~\ref{subsec:mix-rate} extends the analysis to continuous states, giving an excess value rate over $a_{\mathcal H}$ of $\widetilde O_p((NT)^{-1/14})$ under the stated conditions.
\end{example}

%% file: sections/05_experiments.tex
\raggedbottom
\section{Experiments: offline Pendulum control}\label{sec:main-experiments}
We modify the Pendulum environment \citep{BrockmanEtAl2016} by replacing its scalar torque action with a bounded torque function $a:[0,1]\to[-2,2]$ applied over each decision interval. Then we run experiments in this modified environment to provide empirical motivation for two aspects of the theory: the benefit of selecting functional-action policies over constant-action policies,
and assessing logged information through the critic differences that matter for policy improvement. The experiment illustrates these ideas but does not verify all conditions of the value guarantee. Appendix~\ref{app:experiment-design} gives details of implementation and evaluation.

For each sample size $n=NT\in\{2000,4000,8000,16000,32000\}$, with $T=20$ and $\gamma=0.95$, we separately train FQI policies selecting smooth torque functions over each decision interval and state-dependent constant torques on the same 20 independently generated datasets. We estimate each policy value using 1,000 independent Monte Carlo trajectories, separately from those used to select the smoothness penalty. Figure~\ref{fig:main-sample-behavior}(a) shows that as the sample size increases,
the mean normalized value increases from $0.587$ to $0.702$ for functional actions, versus $0.531$ to $0.560$ for constant actions, %
illustrating the benefit of retaining within-action shape. 

To examine coverage at learned actions, for each dataset we obtain $q=NT=n$ diagnostic pairs $(S_j^{\rm ev},A_j^{\rm ev})$ from an independent cohort of $N$ behavior-policy trajectories with $T=20$ decisions each. Because states are continuous, we compare each learned action and its paired behavior action with logged actions at the same 32 nearest training states, approximating a same-state comparison. Figure~\ref{fig:main-sample-behavior}(b) shows larger median nearest-action $L^2$ distances for learned actions than for paired behavior actions, suggesting weaker finite-sample neighborhood coverage at learned actions.
Figure~\ref{fig:adafnn-neighbor-sensitivity} 
in Appendix~\ref{app:empirical-spectral} reveals the same pattern for other neighborhood sizes.

Our critic-relative coverage condition concerns error transfer from the logged design to policy-selected actions. As a descriptive finite-sample diagnostic, we use the adjacent fitted-critic difference $h=\widehat Q_{20}-\widehat Q_{19}$ and define the critic-energy ratio $\mathsf{ER}_h=\bigl[q^{-1}\sum_{j=1}^q h\{S_j^{\rm ev},\widehat\pi_{20}(S_j^{\rm ev})\}^2\bigr]/\allowbreak\bigl[n^{-1}\sum_{i=1}^N\sum_{t=1}^T h(S_{i,t},A_{i,t})^2\bigr]$. The numerator and denominator average squared $h$ at learned actions and logged pairs, respectively. 
In Figure~\ref{fig:main-sample-behavior}(c), the median ratios over the 20 datasets remain near 1 across sample sizes,
indicating comparable energy under both designs. This suggests that weaker local action matching does not necessarily imply poor empirical transfer for the observed critic difference.

Together, these findings motivate critic-relative coverage based on error transfer rather than action proximity alone.
Figure~\ref{fig:main-sample-behavior} only focuses on AdaFNN. 
For KRR (Appendix~\ref{app:additional-ada}), the critic-energy ratio is rising over $n$ despite the same qualitative learned-versus-behavior distance gap, which further illustrates why coverage should be assessed relative to the fitted critic class.

\begin{figure}[H]
\centering
\includegraphics[width=\linewidth]{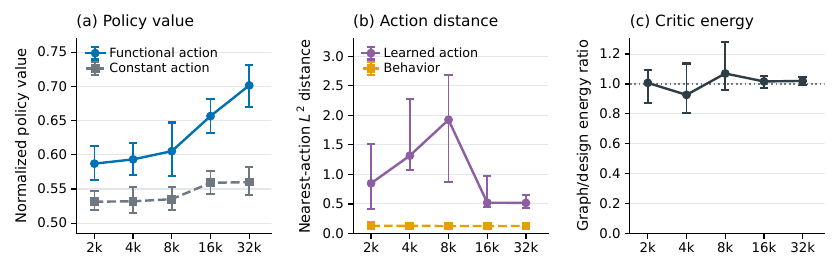}
\caption{AdaFNN policy value and critic-relative diagnostics versus the number $n$ of offline transitions. (a) Mean normalized policy value for functional- and constant-action policies. (b) Median nearest-action $L^2$ distance. (c) Median critic-energy ratio $\mathsf{ER}_h$; the dotted line marks equal energy. Bars are pointwise 95\% percentile-bootstrap intervals across 20 independent datasets.}
\label{fig:main-sample-behavior}
\end{figure}

\FloatBarrier

%% file: sections/06_discussion.tex
\section{Discussion}
In this paper, we establish finite-sample value guarantees for offline functional-action FQI by combining critic-relative information, policy smoothness, and stable state propagation. Due to the unique nature of functional actions, we propose to use an alternative critic-relative coverage condition, under which we bound learned-policy regret relative to the best value within a fixed smooth policy class.
Under suitable conditions for a few special examples, we can obtain a polynomially decaying bound in the number of logged transitions, up to logarithmic factors, with logarithmically many FQI iterations. Numerical experiments show gains of learned functional-action policies over constant-action policies and support our adoption of a critic-relative coverage condition. A future direction is to extend critic-relative coverage to distributional policy evaluation with functional actions.

%% file: sections/07_statements.tex
\subsection*{AI use statement}
The authors used AI assistance to review mathematical arguments for potential logical gaps, search for counterexamples, check draft derivations, improve the exposition, and polish the language. The authors take full responsibility for all mathematical claims and the final content.

\subsection*{Ethics statement}
The experiments use a synthetic control environment and do not involve human participants or private personal data.

\subsection*{Reproducibility statement}
Proofs and experimental details are provided in the appendices. Code for reproducing the experiments is available at \url{https://github.com/gefeilin/Functional-Fitted-Q-Learning/tree/theory}.

%% file: appendix.tex
\appendix
\setcounter{figure}{0}
\renewcommand{\thefigure}{S\arabic{figure}}
\renewcommand{\theHfigure}{appendix.\arabic{figure}}
\numberwithin{theorem}{section}
\numberwithin{proposition}{section}
\numberwithin{lemma}{section}
\numberwithin{corollary}{section}
\numberwithin{definition}{section}
\numberwithin{assumption}{section}
\numberwithin{example}{section}
\numberwithin{remark}{section}
\section*{Appendix}
\addtocontents{toc}{\protect\setcounter{tocdepth}{2}}
Appendices~\ref{sec:data-model}--\ref{sec:bspline-approximation} give the policy domain, Bellman identities, RKHS embeddings and conforming spline approximation.
Appendices~\ref{sec:policy-update}--\ref{sec:general-oracle} establish full-class policy errors, subject concentration, deterministic label covers and spectral transfer.
Appendix~\ref{sec:finite-sample} gives the general finite-iteration value bound and proves its stable-propagation specialization in the main text.
KRR and AdaFNN fitting bounds are in Appendices~\ref{app:direct-krr} and~\ref{app:ada-regression}.
Appendices~\ref{app:closed-balanced-rates}--\ref{sec:mix-complete} give implementation, quantitative approximation and temporal-dependence refinements.
Appendix~\ref{app:verification-examples} verifies the same control model used throughout the main text. Experimental details begin in Appendix~\ref{app:experiment-design}.

\begingroup
\hypersetup{linktoc=all}
\renewcommand{\contentsname}{Contents}
\small
\tableofcontents
\endgroup
\clearpage

\section{Technical setting and algorithmic objects}\label{sec:data-model}
\subsection{One RKHS policy domain}\label{subsec:rkhs-benchmark}
We use $\mathcal S,\mathcal A,\mathcal H,\Pi$, $\Pi_{\mathcal H}(B_{\mathcal H})$, $J^\star$ and $J_{\mathcal H}^\star$ from Section~\ref{sec:main-method}.
The policy RKHS has kernel $\Gamma:\mathcal S\times\mathcal S\to\mathcal B\{L^2([0,1])\}$ and is chosen so that its evaluation maps obey~\eqref{eq:main-policy-spectral-evaluation}. Concrete weighted-spectral and Sobolev constructions are given in Appendix~\ref{subsec:penalty-smooth-actions}.
The scalar critic RKHS $\mathcal F$ is a separate space on $\mathcal S\times\mathcal A$.
Every returned policy belongs to the base class $\Pi$. The event $\mathcal R_n$ from Section~\ref{subsec:one-step-errors} additionally places the initialization and every returned update in the fixed comparison ball $\Pi_{\mathcal H}(B_{\mathcal H})$ and their critics in $\mathcal Q_n$. The value analysis uses $\mathcal E_n=\mathcal R_n\cap\mathcal U_n$ from Section~\ref{subsec:one-step-errors}, including localization failure in $\Pr(\mathcal E_n^c)$. Standalone statistical results below give sufficient events for the four bounds defining $\mathcal U_n$ in~\eqref{eq:main-uniform-error-event}; their failure probabilities are combined with the localization failure probability to bound $\Pr(\mathcal E_n^c)$.

\begin{assumption}[Measurability and conditional compatibility]\label{ass:mdp-compatibility}\label{ass:reward-regularity}
The environment kernel $\Penv(dr,ds'\mid s,a)$ is measurable on the standard Borel spaces, has rewards in $[0,1]$, conditional mean reward $\bar r(s,a)$ and next-state marginal $\Pnext$. Its conditional law agrees with every observed reward--next-state pair given the current pair. Policies and critics are jointly measurable. Policy envelopes of the critics, their absolute differences, and the successive upper-Bellman iterates starting at zero are measurable. These properties extend to their uniform limits.
\end{assumption}
The common kernel makes the conditional Bellman target the same function at every observed pair, as in the time-homogeneous control model of \citet{AntosEtAl2007}. Measurable policy envelopes also make the full RKHS suprema legitimate random variables; subject independence and temporal dependence enter separately through the sampling conditions.
Here and below $\law(Z)$ denotes a probability law; norms without a subscript use the declared Hilbert or Euclidean space.
The main data have fixed common length $T=T_N$, $n=NT$, and independent subjects. Observed marginals can vary with time.
Unequal and random lengths are confined to Appendix~\ref{sec:ti-cluster}'s extension. In that extension, conditioning on lengths must preserve subject independence, the common conditional model, and the uniform bounds used in the proof.

A useful sufficient measurability construction is a separable policy class in the topology of uniform $L^2$ evaluation, jointly measurable evaluation maps, and critics continuous in the action at each state. Then a countable dense policy subset gives the same pointwise suprema. It also suffices that these properties hold in an invariant, sup-norm-closed class containing zero and all upper-Bellman iterates. For example, a Feller reward--transition model on compact metric state space with uniformly continuous critic envelopes has this property. More generally, the explicit measurable-envelope assumption above is the premise. We construct the fixed point by the measurable iterates from zero, without claiming that every bounded measurable critic has a measurable supremum over an uncountable policy class.

For $V^\pi(s)=\E_s^\pi\sum_{t\ge0}\gamma^tR_t$,
$Q^\pi=\bar r+\gamma \Pnext V^\pi$, and $J(\pi)=\E_\rho V^\pi$.
The smooth-class analysis first bounds $J_{\mathcal H}^\star-J(\widehat\pi_M)$; adding $a_{\mathcal H}=J^\star-J_{\mathcal H}^\star$ gives the reported target $J^\star-J(\widehat\pi_M)$.
No attainment of the supremum defining $J^\star$ is needed.

\subsection{Numerical policy representations and the norm penalty}\label{subsec:bspline-class}
Algorithm~\ref{alg:ffqi-main} optimizes over the full action-feasible domain $\Pi$. For its numerical verification, we later use a deterministic feasible subset $\mathcal C_n\subset\Pi$ to bound representation and solver errors; $\mathcal C_n$ is not part of the theoretical algorithm and need not impose an RKHS-radius constraint.
For the B-spline vector $\mathbf b_n^{\rm spl}=(b_{1,n}^{\rm spl},\ldots,b_{p_n,n}^{\rm spl})^\top$ and measurable coefficient maps $c_\theta(s)$, its policies have the form $\pi_\theta(s)(u)=\mathbf b_n^{\rm spl}(u)^\top c_\theta(s)$.
Both functional-action regularity and membership in the function-valued policy RKHS must be checked. Appendix~\ref{sec:bspline-approximation} supplies such constructions; membership is not inferred from a finite parameter count.

Write
\[
 \begin{aligned}
 \widehat\Phi_Q(\pi)&=n^{-1}\sum_{i,t}Q\{S_{i,t+1},\pi(S_{i,t+1})\},\qquad
 \Phi_Q(\pi)=\E_{\nu_+}Q\{S,\pi(S)\},\\
 \widehat\Psi_Q(\pi)&=\widehat\Phi_Q(\pi)-\lambda_{\pi,n}\|\pi\|_{\mathcal H}^2.
 \end{aligned}
\]
The RKHS-norm penalty $\lambda_{\pi,n}\|\pi\|_{\mathcal H}^2$ is deterministic and is the only regularization used in the theoretical policy update; hence its policy proof requires no empirical-to-population comparison for this penalty. Uniform numerical discrepancies are added to the full-class objective gap in Appendix~\ref{sec:policy-update}. The reported finite-dimensional experiment replaces direct computation of the complete norm by the empirical total curvature documented in Appendix~\ref{app:implemented-penalty}; it does not add that criterion to the RKHS-norm objective. Appendix~\ref{app:closed-balanced-rates} proves that this surrogate is coverage-compatible because its localized spline classes satisfy the spectral envelope used by the transfer argument.
Algorithm~\ref{alg:ffqi-main} reuses all $n$ transitions and returns the last policy.
Clipping and $1+\gamma\Vmax=\Vmax$ keep labels and conditional targets in $[0,\Vmax]$.

\section{Bellman envelopes and state transfer}\label{sec:bellman}
For the fixed smooth policy class $\Pi_{\mathcal H}(B_{\mathcal H})$, define
\[
 (\mathbb B^\pi Q)(s,a)=\bar r(s,a)+\gamma\int Q\{s',\pi(s')\}\Pnext(ds'\mid s,a),
 \quad
 \mathbb BQ=\bar r+\gamma \Pnext[\sup_{\pi\in\Pi_{\mathcal H}(B_{\mathcal H})}Q(\cdot,\pi(\cdot))].
\]
The policy-evaluation norm is~\eqref{eq:main-graph-norm}. It obeys the triangle inequality because the pointwise policy supremum does, followed by Minkowski's inequality. It is generally a seminorm: a zero value need not identify functions outside the state--action pairs induced by $\Pi_{\mathcal H}(B_{\mathcal H})$.
For $\pi\in\Pi_{\mathcal H}(B_{\mathcal H})$ write
$g_{Q,\pi}(s)=\sup_{\pi'\in\Pi_{\mathcal H}(B_{\mathcal H})}Q\{s,\pi'(s)\}-Q\{s,\pi(s)\}$.

\begin{lemma}[Upper fixed point and occupancy identity]\label{lem:bellman-contractions}\label{lem:occupancy}
There is a bounded measurable fixed point $Q^\dagger\in[0,\Vmax]$ of $\mathbb B$, unique among bounded functions with measurable policy envelopes. It satisfies $Q^\pi\le Q^\dagger$ for every $\pi\in\Pi_{\mathcal H}(B_{\mathcal H})$.
With $V^\dagger(s)=\sup_{\pi\in\Pi_{\mathcal H}(B_{\mathcal H})}Q^\dagger\{s,\pi(s)\}$ and
$d_\rho^\pi=(1-\gamma)\sum_{t\ge0}\gamma^t\law_\rho^\pi(S_t)$,
\[
 J_{\mathcal H}^\star\le\E_\rho V^\dagger,\qquad
 \E_\rho V^\dagger-J(\pi)
 =\frac1{1-\gamma}\E_{d_\rho^\pi}
           [V^\dagger(S)-Q^\dagger\{S,\pi(S)\}].
\]
\end{lemma}
The Bellman contraction and occupancy expansion are the error-propagation tools used in \citet{MunosSzepesvari2008,FarahmandMunosSzepesvari2010}. Here the envelope may choose different members of a coupled RKHS class at different states, so its fixed point provides an upper value comparator; the actual returned policy is compared to it through its nonnegative greedy defect.
\begin{proof}
For admissible measurable envelopes,
$|\mathbb BQ-\mathbb BQ'|\le\gamma \Pnext\sup_{\pi\in\Pi_{\mathcal H}(B_{\mathcal H})}|(Q-Q')(\cdot,\pi(\cdot))|
\le\gamma\|Q-Q'\|_\infty$.
The iterates of zero stay in $[0,\Vmax]$ and are uniformly Cauchy. Their measurable limit and the uniform limit of their envelopes give the fixed point; the same inequality proves uniqueness.
Monotonicity implies $(\mathbb B^\pi)^kQ^\dagger\le Q^\dagger$, and the left side converges uniformly to $Q^\pi$. Consequently
$J(\pi)\le\E_\rho V^\dagger$ for every $\pi\in\Pi_{\mathcal H}(B_{\mathcal H})$, and taking the supremum proves $J_{\mathcal H}^\star\le\E_\rho V^\dagger$ directly.

Let $\Pnext^\pi f(s)=\Pnext f\{s,\pi(s)\}$ and
$g_\pi=V^\dagger-Q^\dagger(\cdot,\pi(\cdot))\ge0$.
Subtracting the policy Bellman equation gives
$V^\dagger-V^\pi=g_\pi+\gamma \Pnext^\pi(V^\dagger-V^\pi)$.
Iterate $k$ times. The remainder has sup norm at most $\gamma^k\Vmax$ and tends to zero. Integrating the nonnegative sum against $\rho$ and using monotone convergence proves the identity.
\end{proof}

\begin{assumption}[Residual families for state transfer]\label{ass:state-conc}\label{ass:bellman-coverage}
Let $\mathcal Q_n$ be the deterministic fitted class and include $Q^\dagger$ in the following residual families:
\[
 \begin{aligned}
 h(s)&=\sup_{\pi\in\Pi_{\mathcal H}(B_{\mathcal H})}|Q\{s,\pi(s)\}-Q'\{s,\pi(s)\}|,
       &&Q,Q'\in\mathcal Q_n\cup\{Q^\dagger\},\\
 h(s)&=g_{Q,\pi_0}(s),
       &&Q\in\mathcal Q_n\cup\{Q^\dagger\},\quad\pi_0\in\Pi_{\mathcal H}(B_{\mathcal H}) .
 \end{aligned}
\]
For every possible output $\pi$, $\E_{d_\rho^\pi}h\le C_{\rm occ}(\pi)\|h\|_{L^2(\nu_+)}$.
For all residuals in these families,
$\|\Pnext h\|_{\nu_+,\Pi_{\mathcal H}(B_{\mathcal H}),2}\le C_B\|h\|_{L^2(\nu_+)}$.
\end{assumption}
The future-state concentrability of \citet{MunosSzepesvari2008} and the $L^2$ density norms of \citet{FarahmandMunosSzepesvari2010} motivate the two roles below. We require the inequalities on the displayed state residuals, separating final-policy visitation from one-step propagation of the uniform policy-evaluation envelope.
These are the families used in Assumption~\ref{ass:main-state-transfer}.
Sufficient conditions are
\begin{equation}
 \left\|\frac{dd_\rho^\pi}{d\nu_+}\right\|_{L^2(\nu_+)}\le C_{\rm occ}(\pi),
 \qquad
 \sup_{s,\pi}\left\|\frac{dP(\cdot\mid s,\pi(s))}{d\nu_+}\right\|_{L^2(\nu_+)}
 \le C_B,
 \label{prop:coverage-domination}
\end{equation}
with absolute continuity of all displayed state laws.
The density-ratio bound for $d_\rho^\pi$ in~\eqref{prop:coverage-domination} follows by change of measure and Cauchy--Schwarz. Applying Cauchy--Schwarz at each $(s,\pi(s))$, taking the policy supremum, and integrating over $\nu_+$ proves the transition-density bound in the same display.
These state-side sufficient conditions connect with \citet{MunosSzepesvari2008,FarahmandMunosSzepesvari2010,ChenJiang2019,XieJiang2020,RashidinejadEtAl2021,ZhanEtAl2022}.
The fixed-point contraction in sup norm uses only $\gamma<1$.
The additional condition $\gamma C_B<1$ is used for stable propagation in the policy-evaluation norm.

\section{RKHS regularity and conforming spline approximation}\label{sec:bspline-approximation}
\subsection{Policy evaluation and functional-action spectral energy}\label{subsec:penalty-smooth-actions}
Condition~\eqref{eq:main-policy-spectral-evaluation} makes the policy geometry itself compatible with the stronger action norm used in coverage transfer. For every $\pi\in\Pi_{\mathcal H}(B_{\mathcal H})$,
\[
 \|\pi(s)\|_\zeta^2\le\kappa_\zeta(s)^2B_{\mathcal H}^2.
\]
Taking the classwise supremum and integrating proves~\eqref{eq:main-policy-radius}. The following constructions show that this requirement can be built directly into an $L^2$-valued policy RKHS.

\begin{proposition}[Policy RKHS constructions with spectral evaluation]\label{prop:functional-action-rkhs}
Let $\phi_j(u)=\sqrt{2j-1}\,\mathrm{Leg}_{j-1}(2u-1)$ be normalized shifted Legendre polynomials, where $\mathrm{Leg}_k$ denotes the degree-$k$ Legendre polynomial.
Let $\mathcal H_S$ be a scalar RKHS on $\mathcal S$ with kernel $k_S$ and suppose
$\int k_S(s,s)\nu_+(ds)<\infty$. The $L^2$-valued space of maps
\[
 \pi(s)=\sum_{j\ge1}f_j(s)\phi_j,
 \qquad
 \|\pi\|_{\mathcal H}^2=\sum_{j\ge1}j^{2\zeta}\|f_j\|_{\mathcal H_S}^2
\]
is a separable vector-valued RKHS with operator-valued kernel
\[
 \Gamma_\zeta(s,t)=k_S(s,t)T_\zeta,
 \qquad T_\zeta\phi_j=j^{-2\zeta}\phi_j.
\]
It satisfies~\eqref{eq:main-policy-spectral-evaluation} with
$\kappa_\zeta(s)=k_S(s,s)^{1/2}$.

Alternatively, let $\mathcal G=H^2(0,1)$ with
$\langle a,b\rangle_{\mathcal G}=\langle a,b\rangle_{L^2}+\langle a'',b''\rangle_{L^2}$. Then $\mathcal G$ is an RKHS and $\|a\|_{\zeta}^2\le c_{\mathcal G}\|a\|_{\mathcal G}^2$ for $0<\zeta\le2$. Consequently, any $\mathcal G$-valued policy RKHS, viewed through the continuous embedding $\mathcal G\hookrightarrow L^2$, that satisfies
$\|\pi(s)\|_{\mathcal G}\le\kappa_{\mathcal G}(s)\|\pi\|_{\mathcal H}$ and
$\kappa_{\mathcal G}\in L^2(\nu_+)$ verifies~\eqref{eq:main-policy-spectral-evaluation} with
$\kappa_\zeta=c_{\mathcal G}^{1/2}\kappa_{\mathcal G}$.
\end{proposition}
The first construction encodes the coverage norm directly in the operator-valued kernel. The second uses a Sobolev action component, which is convenient for spline computation and for the running example. In both cases the deterministic policy norm, rather than a second population penalty, supplies the spectral radius.
\begin{proof}
For the first construction, scalar RKHS evaluation gives
\[
 \|\pi(s)\|_\zeta^2
 =\sum_jj^{2\zeta}|f_j(s)|^2
 \le k_S(s,s)\sum_jj^{2\zeta}\|f_j\|_{\mathcal H_S}^2.
\]
The coefficient-space construction is complete and evaluation into $L^2$ is bounded because the same display with the weights removed is no larger. Its reproducing operator is $\Gamma_\zeta$: expanding both arguments in the $\phi_j$ basis gives the reproducing identity coordinate by coordinate. Integration of the displayed inequality proves the required square-integrability.

For the Sobolev construction, the elementary decomposition
$a(u)=a(0)+ua'(0)+\int_0^u(u-t)a''(t)\,dt$
bounds $\|a'\|_{L^2}$ by $C(\|a\|_{L^2}+\|a''\|_{L^2})$:
the integral remainder has norm at most $\|a''\|_{L^2}/\sqrt{12}$, and the norm of a linear polynomial controls its slope. The resulting norm is equivalent to the usual complete $H^2$ norm. Also
$|a(u)|\le\|a\|_{L^2}+\|a'\|_{L^2}$ up to a universal constant, uniformly in $u$.
Evaluation is bounded, so its Riesz representers define a reproducing kernel.
The spectral embedding is proved next. Composing it with the assumed $\mathcal G$-valued evaluation inequality gives the stated $\kappa_\zeta$.
\end{proof}

\begin{proposition}[Sobolev control of a common functional spectral norm]\label{prop:curvature-spectral}
Let $\mathrm{Leg}_k$ be the degree-$k$ Legendre polynomial and
$\phi_j(u)=\sqrt{2j-1}\,\mathrm{Leg}_{j-1}(2u-1)$, $j\ge1$.  For every $g\in H^2(0,1)$ and $0<\zeta\le2$,
\begin{equation}
 \sum_{j\ge1}j^{2\zeta}\langle g,\phi_j\rangle^2
 \le C_S\{\|g\|_{L^2}^2+\|g''\|_{L^2}^2\}.
 \label{eq:curvature-spectral-bound}
\end{equation}
No endpoint condition is imposed on $g$: neither $g$ nor $g'$ is required to vanish at $0$ or $1$. Consequently, the spectral bound~\eqref{eq:curvature-spectral-bound} applies to cubic spline actions with the continuity needed for a weak second derivative. The spectral basis is an analytical representation, not a replacement of the policy's B-spline basis.
\end{proposition}
The shifted Legendre equation \citep[Section~14.2]{NISTDLMFLegendre} supplies the spectral identity used below. The estimate verifies that an $H^2$ component of the policy-RKHS norm controls the weighted feature norm in the transfer assumption, while leaving the endpoint values of the action function unconstrained.
\begin{proof}
Let $w(u)=u(1-u)$ and $\operatorname{LegOp}g=-(wg')'$. The shifted Legendre equation gives $\operatorname{LegOp}\phi_j=\varpi_j\phi_j$, with $\varpi_j=(j-1)j$ \citep[Section~14.2]{NISTDLMFLegendre}. For smooth $g$, integration by parts twice gives
\begin{align*}
 \langle \operatorname{LegOp}g,\phi_j\rangle
 &=[-wg'\phi_j+wg\phi_j']_0^1-\int_0^1g(w\phi_j')'\,du
 =\varpi_j\langle g,\phi_j\rangle.
\end{align*}
Both boundary terms vanish because $w(0)=w(1)=0$, not because of a boundary condition on $g$. Since $j^{2\zeta}\le j^4\le4(1+\varpi_j^2)$, Parseval's identity yields
\begin{align*}
 \sum_jj^{2\zeta}\langle g,\phi_j\rangle^2
 &\le4\|g\|_{L^2}^2+4\|\operatorname{LegOp}g\|_{L^2}^2.
\end{align*}
To bound $\operatorname{LegOp}g$, write $g(u)=g(0)+ug'(0)+v(u)$, where $v(u)=\int_0^u(u-t)g''(t)\,dt$. Cauchy--Schwarz gives $\|v\|_{L^2}\le\|g''\|_{L^2}/\sqrt{12}$. Also,
\begin{align*}
 \|a+bu\|_{L^2}^2&=(a+b/2)^2+b^2/12,\\
 |g'(0)|&\le\sqrt{12}\|g-v\|_{L^2}
 \le\sqrt{12}\|g\|_{L^2}+\|g''\|_{L^2},\\
 \left\|\int_0^{\cdot}g''(t)\,dt\right\|_{L^2}&\le\|g''\|_{L^2}/\sqrt2.
\end{align*}
Consequently $\|g'\|_{L^2}\le\sqrt{12}\|g\|_{L^2}+(1+1/\sqrt2)\|g''\|_{L^2}$. Using $\|w\|_\infty=1/4$ and $\|w'\|_\infty=1$ gives
\begin{align*}
 \|\operatorname{LegOp}g\|_{L^2}&\le\sqrt{12}\|g\|_{L^2}+2\|g''\|_{L^2},\\
 \sum_jj^{2\zeta}\langle g,\phi_j\rangle^2
 &\le100\|g\|_{L^2}^2+32\|g''\|_{L^2}^2
 \le200(\|g\|_{L^2}^2+\|g''\|_{L^2}^2).
\end{align*}
For general $g\in H^2(0,1)$, choose smooth $g_k\to g$ in $H^2$. Apply the bound to each finite coordinate sum, pass to $k\to\infty$ by continuity of the coordinates, and then let the number of coordinates increase by monotone convergence. This proves the same inequality for $g$.
\end{proof}

\subsection{An energy projection with a uniform action error}\label{prop:conforming-spline}
Let a quasi-uniform partition have $d_u$ intervals and maximum width $h\asymp d_u^{-1}$.
Use the space of piecewise cubic, globally $C^1$ splines. This is an $H^2$-conforming B-spline space with double interior knots. Write $\ProjOp_{d_u}$ for its $\mathcal G$-orthogonal projection.
The use of an energy projection is part of this verified numerical construction.

\begin{proposition}[Feasible functional-action and objective approximation]\label{prop:spline-objective}
For every $a\in H^2(0,1)$,
\[
 \|\ProjOp_{d_u}a\|_{\mathcal G}\le\|a\|_{\mathcal G},\qquad
 \|a-\ProjOp_{d_u}a\|_{L^2}\le C d_u^{-2}\|a\|_{\mathcal G}.
\]
For ${B}_A>0$ and $\|a\|_{L^2}\le{B}_A$, define
$\widetilde a={B}_A\ProjOp_{d_u}a/
\max\{B_A,\|\ProjOp_{d_u}a\|_{L^2}\}$.
Then $\|\widetilde a\|_{L^2}\le{B}_A$,
$\|\widetilde a\|_{\mathcal G}\le\|a\|_{\mathcal G}$ and
\[
 \|a-\widetilde a\|_{L^2}\le2C d_u^{-2}\|a\|_{\mathcal G}.
\]
If $Q(s,\cdot)$ is $\mathrm{Lip}_Q$-Lipschitz, replacing $a$ by $\widetilde a$ in
$Q(s,a)-\lambda\|a\|_{\mathcal G}^2$ incurs loss at most
$\mathrm{Lip}_Q\|a-\widetilde a\|_{L^2}$.
\end{proposition}
Spline approximation with derivative-penalty control is developed, for example, by \citet[Proposition~2.1]{HuangSu2021PenalizedSplines}. The energy projection here additionally preserves the fixed action-RKHS budget, and the radial adjustment preserves action feasibility; the final inequality converts these properties into a policy optimization error.
\begin{proof}
Orthogonal projection is a contraction. Put $e=a-\ProjOp_{d_u}a$.
Riesz representation gives $z\in H^2$ satisfying
$\langle z,w\rangle_{\mathcal G}=\langle e,w\rangle_{L^2}$ for every $w\in H^2$.
Testing with $z$ gives $\|z\|_{\mathcal G}\le\|e\|_{L^2}$.
For compactly supported tests the weak equation is $z''''+z=e$.
A fourfold primitive of $e-z\in L^2$ is in $H^4$; its difference from $z$ has fourth distributional derivative zero and is a polynomial of degree at most three.
The coefficients of that polynomial are controlled by its $L^2$ norm.
It follows that $\|z\|_{H^4}\le C(\|z\|_{L^2}+\|e-z\|_{L^2})\le C\|e\|_{L^2}$.
Integration by parts in the weak identity gives the natural conditions
$z''(0)=z''(1)=z'''(0)=z'''(1)=0$.
These conditions apply to the dual solution, not to the action $a$.

For clarity, the required spline approximation can be obtained by cubic Hermite interpolation of $z$ and $z'$ at the mesh nodes.
On the unit reference interval subtract its third-order Taylor polynomial.
The remainder and its first three derivatives are integrals of $z^{(4)}$, so their sup norms are bounded by a constant times $\|z^{(4)}\|_{L^2}$.
The Hermite interpolant of this remainder has coefficients bounded by the same endpoint trace bounds.
Write $\InterpOp_hz$ for the resulting piecewise Hermite interpolant.
Scaling to an interval $I$ of width $h_I$ gives, for $j=0,1,2$,
\[
 \|(z-\InterpOp_hz)^{(j)}\|_{L^2(I)}
 \le C h_I^{4-j}\|z^{(4)}\|_{L^2(I)}.
\]
The pieces share value and first derivative at nodes and thus belong to the conforming cubic space.
Summing squared bounds gives
$\|z-\InterpOp_hz\|_{\mathcal G}\le Ch^2\|z\|_{H^4}$.
Projection orthogonality now yields
\[
 \|e\|_{L^2}^2
 =\langle e,z-\ProjOp_{d_u}z\rangle_{\mathcal G}
 \le C h^2\|e\|_{\mathcal G}\|e\|_{L^2}
 \le C h^2\|a\|_{\mathcal G}\|e\|_{L^2}.
\]
Division when $e\ne0$ proves the error estimate; the zero case is immediate.

Radial scaling is a contraction in $\mathcal G$ and changes the projected action function by
$(\|\ProjOp_{d_u}a\|_{L^2}-{B}_A)_+\le\|\ProjOp_{d_u}a-a\|_{L^2}$.
The triangle inequality proves the factor two.
The RKHS-norm penalty cannot increase because $\|\widetilde a\|_{\mathcal G}\le\|a\|_{\mathcal G}$; action Lipschitz continuity therefore proves the last assertion.
\end{proof}

For the finite-state product class of Example~\ref{ex:two-state-control}, apply this construction at each state.
The resulting policy belongs to $\mathcal C_n\subset\Pi$. If $T_{\mathcal G}:L^2\to\mathcal G\subset L^2$ is defined by
$\langle T_{\mathcal G}a,b\rangle_{\mathcal G}=\langle a,b\rangle_{L^2}$, its $L^2$-valued operator kernel is
$\Gamma(s,t)=\ind\{s=t\}T_{\mathcal G}$.
Uniform averaging of the pointwise objective bound controls the full-class-to-spline empirical objective gap, with no value-level approximation term.

\subsection{A compatible continuous-state implementation}\label{subsec:continuous-rkhs-implementation}
Here is an independent implementation verification on $\mathcal S=[0,1]$, not another control model.
Let $\mathcal H=H^1([0,1];\mathcal G)$ with norm
\[
 \|\pi\|_{\mathcal H}^2=\int_0^1
       \{\|\pi(s)\|_{\mathcal G}^2+\|\partial_s\pi(s)\|_{\mathcal G}^2\}\,ds .
\]
Bounded Hilbert-valued evaluation makes this a $\mathcal G$-valued RKHS.
For the fixed joint ball with $\|\pi(s)\|_{L^2}\le{B}_A$, take the tensor product of an $H^1$ energy projection onto continuous piecewise linear state functions and the functional-action projection $\ProjOp_{d_u}$.
Both projections are contractions in $\mathcal H$.

If the state mesh width is $h_s$, its uniform evaluation error is at most
$C h_s^{1/2}\|\pi\|_{\mathcal H}$ in $\mathcal G$.
To verify this, the dual equation for the state energy projection is $z-z''=e$ with natural Neumann boundary conditions; a twofold-primitive argument gives $\|z\|_{H^2(\mathcal G)}\le C\|e\|_{L^2(\mathcal G)}$.
Piecewise linear interpolation and energy orthogonality give
$\|e\|_{L^2(\mathcal G)}\le C h_s\|\pi\|_{\mathcal H}$.
The one-dimensional interpolation inequality
$\|e\|_{\infty,\mathcal G}^2\le C\|e\|_{L^2(\mathcal G)}\|e\|_{H^1(\mathcal G)}$
then gives the claim. Indeed, choose $t_0\in[0,1]$ with
$\|e(t_0)\|_{\mathcal G}\le\|e\|_{L^2(\mathcal G)}$.
Absolute continuity of $\|e(\cdot)\|_{\mathcal G}^2$ and Cauchy--Schwarz give
\[
 \|e\|_{\infty,\mathcal G}^2
 \le\|e\|_{L^2(\mathcal G)}^2
       +2\|e\|_{L^2(\mathcal G)}\|e'\|_{L^2(\mathcal G)}
 \le3\|e\|_{L^2(\mathcal G)}\|e\|_{H^1(\mathcal G)}.
\]
Applying Proposition~\ref{prop:spline-objective} after the state projection yields uniform $L^2$ action error
$C B_{\mathcal H}(h_s^{1/2}+d_u^{-2})$ on $\Pi_{\mathcal H}(B_{\mathcal H})$.

A single scalar shrinkage by
${B}_A/\max\{B_A,\sup_s\|\pi_{\rm proj}(s)\|_{L^2}\}$
preserves the tensor spline space and the $\mathcal H$ budget and at most doubles this error.
Because both projections and the final shrinkage contract the $\mathcal H$ norm, the deterministic RKHS-norm penalty cannot increase. Thus the full regularized empirical representation loss is bounded by
\[
 C\,\mathrm{Lip}_Q B_{\mathcal H}(h_s^{1/2}+d_u^{-2}).
\]
This verifies membership, action feasibility and an objective bound for a continuous-state implementation.
The coupled class still retains $\Delta(Q)$.
Discontinuous state-cell policies are not used as approximants in this $\mathcal H$.

\section{Policy improvement and subject concentration}\label{sec:policy-update}
The approximation gap $a_{\mathcal H}$ concerns restriction to the fixed policy ball, whereas $\Delta_n$ concerns compatibility of statewise maximization with a single policy in that ball. Policy localization separately requires all returned policies to remain in the ball on $\mathcal E_n$. Example~\ref{ex:two-state-control} establishes localization with $a_{\mathcal H}=0$ and controls the greedy defect directly, without requiring $\Delta_n=0$.
A positive penalty coefficient $\lambda_{\pi,n}$ alone does not establish a sample-size-independent radius. Appendix~\ref{lem:rev-calibration} gives a sufficient low-norm comparator calibration.
\subsection{Policy optimization error and the actual greedy defect}\label{subsec:relative-penalty-calibration}
For $Q:\mathcal S\times\mathcal A\to[0,\Vmax]$, define
\begin{align*}
 \Delta(Q)&=\E_{\nu_+}\sup_{\pi\in\Pi_{\mathcal H}(B_{\mathcal H})}Q\{S,\pi(S)\}
              -\sup_{\pi\in\Pi_{\mathcal H}(B_{\mathcal H})}\Phi_Q(\pi),\\
 g_Q^{\rm sel}(s)&=\sup_{\pi\in\Pi_{\mathcal H}(B_{\mathcal H})}Q\{s,\pi(s)\}
              -Q\{s,\widehat\pi_Q(s)\}.
\end{align*}
A bound on the policy optimization error is
$\tau_{\pi,n}\ge\sup_Q[\sup_{\pi\in\Pi_{\mathcal H}(B_{\mathcal H})}\widehat\Psi_Q(\pi)-\widehat\Psi_Q(\widehat\pi_Q)]$.
The theory uses this policy optimization error directly. For a later numerical verification based on $\widehat\pi_Q\in\mathcal C_n\subset\Pi$, it can be bounded by
\begin{align}
 \sup_{\pi\in\Pi_{\mathcal H}(B_{\mathcal H})}\widehat\Psi_Q(\pi)-\widehat\Psi_Q(\widehat\pi_Q)
 &\le\left[\sup_{\pi\in\Pi_{\mathcal H}(B_{\mathcal H})}\widehat\Psi_Q(\pi)
                       -\sup_{\pi\in\mathcal C_n}\widehat\Psi_Q(\pi)\right]_+\nonumber\\
 &\quad+\{\sup_{\pi\in\mathcal C_n}\widehat\Psi_Q(\pi)
                              -\widehat\Psi_Q(\widehat\pi_Q)\}.
 \label{eq:full-policy-objective-decomposition}
\end{align}
The brackets are numerical representation and solver errors. If the computed objective differs uniformly from $\widehat\Psi_Q$ by at most $\epsilon_{\rm num}$, add $2\epsilon_{\rm num}$ to its solver bound.

\begin{lemma}[Objective-to-defect conversion]\label{lem:functional-objective-transfer}
Suppose every returned policy is in $\Pi_{\mathcal H}(B_{\mathcal H})$,
$\sup_{Q,\pi\in\Pi_{\mathcal H}(B_{\mathcal H})}|\widehat\Phi_Q(\pi)-\Phi_Q(\pi)|\le u_n$,
and the policy optimization error is at most $\tau_{\pi,n}$.
Then, for every covered $Q$,
\[
 \E_{\nu_+}g_Q^{\rm sel}\le\Delta(Q)+2u_n+\tau_{\pi,n}+\lambda_{\pi,n}B_{\mathcal H}^2,
 \quad
 \|g_Q^{\rm sel}\|_{L^2(\nu_+)}^2
 \le\Vmax\{\Delta(Q)+2u_n+\tau_{\pi,n}+\lambda_{\pi,n}B_{\mathcal H}^2\}.
\]
\end{lemma}
For the deterministic fitted-critic class $\mathcal Q_n$, let $u_n$ bound uniform sample-to-population objective deviation and let $\tau_{\pi,n}$ bound the policy optimization error, including the representation and numerical errors described above. A uniform policy-improvement error bound is
\begin{equation}
 \varepsilon_{\pi,n}^2=\Vmax\Big\{
 \underbrace{\sup_{Q\in\mathcal Q_n}\Delta(Q)}_{\text{policy-compatibility gap}}
 +\underbrace{2u_n+\tau_{\pi,n}}_{\text{sampling and computation}}
 +\underbrace{\lambda_{\pi,n}B_{\mathcal H}^2}_{\text{RKHS-norm penalty}}
 \Big\}.
 \label{eq:main-policy-error}
\end{equation}
The algorithm optimizes one average-objective policy, whereas the greedy defect compares actions separately at each state. The policy-compatibility gap measures this difference. The square root converts the average-objective bound to the required $L^2(\nu_+)$ greedy-defect bound.
The policy-compatibility gap plays the role of the approximate-greedy error in \citet[Section~3.2.3]{AntosEtAl2007}. Here $\Delta(Q)$ records the compatibility of statewise choices with one RKHS policy, while the remaining terms quantify estimation, numerical representation, optimization and norm regularization within that same policy class.
\begin{proof}
For every $\pi\in\Pi_{\mathcal H}(B_{\mathcal H})$, empirical near-optimality gives
\[
 \Phi_Q(\pi)-\Phi_Q(\widehat\pi_Q)
 \le2u_n+\tau_{\pi,n}
      +\lambda_{\pi,n}\{\|\pi\|_{\mathcal H}^2
                          -\|\widehat\pi_Q\|_{\mathcal H}^2\}
 \le2u_n+\tau_{\pi,n}+\lambda_{\pi,n}B_{\mathcal H}^2.
\]
Take the supremum over $\pi$ and add $\Delta(Q)$.
The displayed $L^2(\nu_+)$ bound follows from $0\le g_Q^{\rm sel}\le\Vmax$.
This proves~\eqref{eq:main-policy-error}, including initialization, whenever the stated bounds and policy localization hold. In the value analysis, these conditions are verified on $\mathcal E_n$.
\end{proof}

If the class is a product of compact action sets over finitely many states, a continuous critic admits independent statewise maximizers and $\Delta(Q)=0$. In a coupled RKHS ball this equality requires verification.
More generally, for a measurable statewise selector $a_Q^\star$ and policies $\pi\in\Pi_{\mathcal H}(B_{\mathcal H})$, action Lipschitz continuity gives
\[
 \Delta(Q)\le \mathrm{Lip}_Q\inf_{\pi\in\Pi_{\mathcal H}(B_{\mathcal H})}
       \{\E_{\nu_+}\|a_Q^\star(S)-\pi(S)\|_{L^2}^2\}^{1/2}.
\]
If the pointwise loss instead has a verified quadratic upper bound
$Q(s,a_Q^\star(s))-Q(s,\pi(s))\le \mathrm{Lip}_{Q,2}\|a_Q^\star(s)-\pi(s)\|_{L^2}^2/2$,
the corresponding squared approximation bound applies.
A constrained boundary maximizer by itself does not make the first derivative vanish.

\subsection{What norm regularization supplies}\label{lem:rev-calibration}
The positive deterministic penalty first gives a sample-size-dependent radius without any concentration argument. Suppose the zero policy belongs to $\Pi$ and the computed update has regularized-objective gap at most $\bar\tau_{\pi,n}$ relative to zero. Since $0\le Q\le\Vmax$,
\begin{equation}
 \|\widehat\pi_Q\|_{\mathcal H}^2
 \le B_{n,\mathrm{pre}}^2
 :=\frac{\Vmax+\bar\tau_{\pi,n}}{\lambda_{\pi,n}}.
 \label{eq:preliminary-policy-radius}
\end{equation}
This preliminary ball is deterministic once the tuning and numerical tolerance are fixed, but it grows if $\lambda_{\pi,n}\downarrow0$.

A fixed-radius refinement follows from low-norm near-optimal comparators. Suppose that, for every covered $Q$, there is a policy $\pi_Q^0\in\Pi$ such that
\[
 \|\pi_Q^0\|_{\mathcal H}\le B_0,
 \qquad
 \sup_{\pi\in\Pi}\Phi_Q(\pi)-\Phi_Q(\pi_Q^0)\le z_n,
\]
and the numerical representation and solver guarantees imply
\[
 \widehat\Psi_Q(\widehat\pi_Q)
 \ge\widehat\Psi_Q(\pi_Q^0)-\bar\tau_{\pi,n}.
\]
Assume $B_0\le B_{n,\mathrm{pre}}$ and
\[
 \sup_{Q,\pi\in\Pi_{\mathcal H}(B_{n,\mathrm{pre}})}
 |\widehat\Phi_Q(\pi)-\Phi_Q(\pi)|\le\bar u_n.
\]
Then
\begin{equation}
 \|\widehat\pi_Q\|_{\mathcal H}^2
 \le B_0^2+\frac{z_n+2\bar u_n+\bar\tau_{\pi,n}}{\lambda_{\pi,n}}.
 \label{eq:calibrated-policy-radius}
\end{equation}
Indeed, regularized near-optimality bounds the returned squared norm by $B_0^2$ plus the empirical objective advantage of $\widehat\pi_Q$ over $\pi_Q^0$, divided by $\lambda_{\pi,n}$, and the latter advantage is at most $z_n+2\bar u_n$. If the right side is at most $B_{\mathcal H}^2$, every returned policy belongs to the fixed ball $\Pi_{\mathcal H}(B_{\mathcal H})$. Thus errors of order $\lambda_{\pi,n}$ and a uniformly bounded comparator norm are sufficient for the localization event used by the coverage and value proofs. All classes in this argument, including the preliminary ball, are declared before fitting.

\subsection{Equal-length subjects and joint covering inputs}\label{subsec:policy-empirical-process}
\begin{assumption}[Independent subjects]\label{ass:bounded-clusters}
The subject-level data have a common trajectory length and satisfy the data conditions in Section~\ref{sec:main-method}; arbitrary dependence within a subject is permitted.
\end{assumption}
Each subject is one independent sampling unit for the concentration argument, so the unconditional statistical scale is $N$. This sampling model complements the stationary mixing trajectory of \citet[Assumption~A2]{AntosEtAl2007}: temporal information is used here through the separate mixing refinement, while the subject-level bounds allow arbitrary within-trajectory dependence.
For a deterministic function class $\mathcal O$, write
${H}_{\mathcal O}(\epsilon)=\log \operatorname{Cov}(\epsilon,\mathcal O,\|\cdot\|_\infty)$.
A covering number is the minimum number of radius-$\epsilon$ balls needed to cover the class.
Unpenalized policy objectives on the localization event form
$\mathcal O_n=\{s\mapsto Q(s,\pi(s)):Q\in\mathcal Q_n,\pi\in\Pi_{\mathcal H}(B_{\mathcal H})\}$.
If critics have action Lipschitz bound $\mathrm{Lip}_n$, their product cover gives
\begin{equation*}
 {H}_{\mathcal O_n}(\epsilon_Q+\mathrm{Lip}_n\epsilon_\pi)
 \le{H}_{\mathcal Q_n}(\epsilon_Q)
       +H_{\Pi_{\mathcal H}(B_{\mathcal H})}(\epsilon_\pi),\qquad
 d_\Pi(\pi,\pi')=\sup_s\|\pi(s)-\pi'(s)\|_{L^2}.
\end{equation*}
Bellman label covers use these unpenalized objectives.
The deterministic RKHS penalty has no deviation term. Optional empirical derivative diagnostics use their own density class or matrix event.

For $0\le f\le B$, a deterministic net and independent subject averages give
\begin{equation}
 \sup_{f\in\mathcal O}|(\widehat{\mathbb P}-\mathbb P)f|
 \le2\epsilon+B\sqrt{\frac{H_{\mathcal O}(\epsilon)+\log(2/\delta)}{2N}}
 \label{eq:equal-objective}
\end{equation}
with probability at least $1-\delta$. In particular, this supplies $u_n$ for $\mathcal O_n$.
For clipped fits and bounded labels with conditional targets $Q_y$, the same subject-level argument with relative Bernstein concentration yields
\begin{equation}
 \sup_y\|\widehat f_y-Q_y\|_{L^2(\nu_X)}^2
 \le C\left\{\eta_Q^2+\tau_Q+\Vmax(\epsilon+\eta)
 +\frac{\Vmax^2[{H}_F(\epsilon)+{H}_Y(\eta)+\log(2/\delta)]}{N}\right\}.
 \label{eq:equal-prediction}
\end{equation}
The full proof, including the different observed marginals within subjects, is the weighted calculation below with weights $1/N$.
There, all three subject scales are exactly $N$.
For a derivative state representation with rank $r_{\rm st}$ and standardized feature envelope $Cr_{\rm st}$, its relative Gram event has failure at most $2r_{\rm st}\exp(-cN/r_{\rm st})$.
These events can be intersected with failure at most their sum; they need not be independent.

\subsection{Extension: unequal and random trajectory lengths}\label{sec:ti-cluster}
The main presentation has equal weights $1/N$. This extension retains the original unequal-length concentration scales. No within-trajectory independence, marginal stationarity, mixing, subsampling or change to the empirical objective is introduced.

\paragraph{Length arrays and the pooled design.}
For $N$ subjects, let $T_{i,N}\ge1$ be the number of logged transitions from subject $i$. Put
\begin{align*}
 n&=\sum_{i=1}^N T_{i,N},&
 \overline T_N&=n/N,&
 T_{\max,N}&=\max_iT_{i,N},&
 w_{i,N}&=T_{i,N}/n.
\end{align*}
The imbalance ratio remains $C_{T,n}=T_{\max,N}/\overline T_N\ge1$. Lengths are deterministic, or probability statements are conditional on their vector under the preceding convention: conditional on that vector, subjects remain independent, each observed reward--transition conditional law remains compatible with the common kernel, and the claimed constants are uniform. Conditioning on arbitrary informative stopping times is not asserted to preserve these properties. For readability omit the subscript $N$ on $T_i,w_i$ within the proof.

For $Z_{it}=(S_{it},A_{it},R_{it},S_{i,t+1})$, write
\begin{align*}
 \widehat{\mathbb P} f
 &=\frac1n\sum_{i,t}f(Z_{it})
  =\sum_iw_i\overline f_i,&
 \overline f_i&=\frac1{T_i}\sum_{t=1}^{T_i}f(Z_{it}),&
 \mathbb P f&=\mathbb E\widehat{\mathbb P} f.
\end{align*}
Thus neither fitting nor policy improvement becomes an equally weighted mean of subject means. The pooled laws remain
\begin{align*}
 \nu_{X,N}&=\sum_iw_i\frac1{T_i}\sum_t\law(S_{it},A_{it}),&
 \nu_N^+&=\sum_iw_i\frac1{T_i}\sum_t\law(S_{i,t+1}).
\end{align*}
They may depend on the length array and on $N$. All approximation, membership and localization, information, state-transfer and state-leverage conditions used below refer to these laws and to the same classes fixed before fitting. A rate for the lengths alone does not establish any of those other conditions.

Define two deterministic cluster scales
\begin{align*}
 \mathcal N_{\infty}&=\frac1{\max_iw_i}=\frac n{T_{\max}}=\frac N{C_{T,n}},&
 \mathcal N_2&=\frac1{\sum_iw_i^2}=\frac{n^2}{\sum_iT_i^2}.
\end{align*}
Because $\sum_iw_i=1$, $\sum_iw_i^2\le\max_iw_i$ and $\sum_iw_i^2\ge1/N$,
\begin{align*}
 N_\star=\frac N{C_{T,n}^2}
 \le\mathcal N_{\infty}\le\mathcal N_2\le N.
\end{align*}
Unlike $N_\star$, $\mathcal N_{\infty}$ is always at least one. It diverges exactly when no single subject has a nonvanishing fraction of the transition weight. The next proof uses $\mathcal N_2$ for bounded objective deviations but only $\mathcal N_{\infty}$ as a common scale for prediction and relative penalty control. The larger $\mathcal N_2$ is not substituted throughout without a separate argument.

\paragraph{Three simultaneous statistical inputs.}
\begin{lemma}[Independent subjects with unequal trajectory lengths]\label{lem:ti-concentration}
Assume independent subjects and the common observed conditional model. All classes and net centers below are deterministic, conditional on the lengths when needed.

For any class $\mathcal O$ with $0\le f\le B$ and logarithmic sup-norm covering number ${H}_{\mathcal O}(\epsilon)$, with probability at least $1-\delta$,
\begin{align}
 \sup_{f\in\mathcal O}|(\widehat{\mathbb P}-\mathbb P)f|
 &\le2\epsilon+B\sqrt{\frac{H_{\mathcal O}(\epsilon)+\log(2/\delta)}{2\mathcal N_2}}
 \le2\epsilon+CB\sqrt{\frac{H_{\mathcal O}(\epsilon)+\log(2/\delta)}{\mathcal N_{\infty}}}.
 \label{eq:ti-objective}
\end{align}

For full-label clipped least-squares regression, suppose $f,Y_y,Q_y\in[0,V_{\max}]$, $Q_y(X_{it})=\mathbb E(Y_y(Z_{it})\mid X_{it})$ at every observed pair, and the uniform full-domain approximation error is at most $\eta_Q$. Let $H(\epsilon,\eta)$ be the sum of the log covers of the fit and label classes at sup-norm accuracies $\epsilon,\eta$. If the returned fit is $\tau_Q$-optimal for the full empirical squared loss, then, on one full-label event of probability at least $1-\delta$,
\begin{align}
 \sup_y\|\widehat f_y-Q_y\|_{L^2(\nu_{X,N})}^2
 \le C\left\{\eta_Q^2+\tau_Q+V_{\max}(\epsilon+\eta)
       +\frac{V_{\max}^2\{H(\epsilon,\eta)+\log(2/\delta)\}}{\mathcal N_{\infty}}\right\}.
 \label{eq:ti-prediction}
\end{align}

For any deterministic state-feature map $w_{\rm st}$, let $\Sigma_{\rm st}=\mathbb E_{\nu_N^+}w_{\rm st}w_{\rm st}^\top$ and
$\widehat\Sigma_{\rm st}=n^{-1}\sum_{i,t}w_{\rm st}(S_{i,t+1})w_{\rm st}(S_{i,t+1})^\top$, with rank $r_{\rm st}\le d_S$. If
$\operatorname{ess\,sup}_{\nu_N^+} w_{\rm st}^\top \Sigma_{\rm st}^\dagger w_{\rm st}\le Cd_S$, then
\begin{align}
 \Pr\left\{\tfrac12\Sigma_{\rm st}\npreceq\widehat\Sigma_{\rm st}
           \ \text{or}\ \widehat\Sigma_{\rm st}\npreceq\tfrac32\Sigma_{\rm st}\right\}
 \le2r_{\rm st}\exp\{-c\mathcal N_{\infty}/d_S\}.
 \label{eq:ti-gram}
\end{align}
The zero-rank case has equality almost surely and failure probability zero. In particular, $\mathcal N_{\infty}\ge Cd_S\log(2d_S/\delta)$ gives the displayed relative Gram comparison. Each bound uses whole subjects as independent units.
\end{lemma}
Weighted subject averages let the unequal-range inequality of \citet[Theorem~2]{Hoeffding1963} retain $\mathcal N_2$ for bounded objectives, whereas the relative-risk and Gram bounds use the largest subject weight through $\mathcal N_{\infty}$. The matrix step applies \citet[Theorem~5.1.1]{Tropp2015} to whole-subject feature sums, preserving all logged transitions in the fitted objective.
\begin{proof}
For a fixed $f$, $w_i\overline f_i$ are independent and lie in intervals of lengths $w_iB$. The unequal-range Hoeffding inequality gives
\begin{align*}
 \Pr\{|(\widehat{\mathbb P}-\mathbb P)f|>x\}
 \le2\exp\left\{-\frac{2x^2}{B^2\sum_iw_i^2}\right\}.
\end{align*}
This is \citet[Theorem~2]{Hoeffding1963} applied to weighted subject averages. A union bound over a deterministic $\epsilon$-net and an off-net cost $2\epsilon$ prove \eqref{eq:ti-objective}. It applies to unpenalized policy objectives by using $f(s')=Q\{s',\pi(s')\}$ and their joint cover.

For a fixed fit/label pair, put
\begin{align*}
 d_{f,y}(Z)&=(f(X)-Y_y(Z))^2-(Q_y(X)-Y_y(Z))^2,\\
 D_i&=T_i^{-1}\sum_td_{f,y}(Z_{it}),\\
 \mu_i&=\mathbb ED_i=T_i^{-1}\sum_t\mathbb E(f-Q_y)^2(X_{it})\ge0,\\
 \epsilon_{f,y}^2&=\sum_iw_i\mu_i=\|f-Q_y\|_{L^2(\nu_{X,N})}^2.
\end{align*}
Boundedness and conditional compatibility imply
$|d_{f,y}|\le V_{\max}^2$ and $d_{f,y}^2\le4V_{\max}^2(f-Q_y)^2$.
Jensen's inequality within the subject, without independence, gives
\begin{align*}
 |w_i(D_i-\mu_i)|&\le2w_iV_{\max}^2,\\
 \sum_i\operatorname{Var}(w_iD_i)
 &\le\sum_iw_i^2\mathbb ED_i^2
 \le4V_{\max}^2\sum_iw_i^2\mu_i
 \le\frac{4V_{\max}^2}{\mathcal N_{\infty}}\epsilon_{f,y}^2.
\end{align*}
The scalar Bernstein inequality for the independent centered terms $w_i(D_i-\mu_i)$, followed by Young's inequality, yields, for $x>0$, with failure at most $2e^{-x}$,
\begin{align*}
 |(\widehat{\mathbb P}-\mathbb P)d_{f,y}|
 &\le\sqrt{\frac{8V_{\max}^2\epsilon_{f,y}^2x}{\mathcal N_{\infty}}}
       +\frac{4V_{\max}^2x}{3\mathcal N_{\infty}}
 \le\tfrac14\epsilon_{f,y}^2+\frac{CV_{\max}^2x}{\mathcal N_{\infty}}.
\end{align*}
One may use the one-dimensional specialization of \citet[Theorem~6.1.1]{Tropp2015}. The linear remainder $4V_{\max}^2x/(3\mathcal N_{\infty})$ in the Bernstein bound is harmless. Take a product fit/label net and $x=H(\epsilon,\eta)+\log(2/\delta)$. The common conditional kernel gives $\|Q_y-Q_{y'}\|_\infty\le\|y-y'\|_\infty$. For representatives $(f',y')$ with $\|f-f'\|_\infty\le\epsilon$ and $\|y-y'\|_\infty\le\eta$, difference of squares gives
\[
 \|d_{f,y}-d_{f',y'}\|_\infty\le
     2V_{\max}\epsilon+6V_{\max}\eta=:a_1,\qquad
 |\epsilon_{f,y}^2-\epsilon_{f',y'}^2|\le2V_{\max}(\epsilon+\eta)=:a_2.
\]
The centered empirical difference costs at most $2a_1$, and replacing the risk in the relative term costs $a_2/4$. Thus the same event extends to all pairs with a remainder
$C\{V_{\max}(\epsilon+\eta)+V_{\max}^2x/\mathcal N_{\infty}\}$.
Net centers may be chosen in the deterministic classes; a factor-two change of radius for a restricted subclass is absorbed into $C$.

Let $f_y^\circ$ approach the deterministic approximation oracle. On that event the approximate empirical minimization inequality gives
\begin{align*}
 \tfrac34\epsilon_{\widehat f_y,y}^2
 \le\tfrac54\epsilon_{f_y^\circ,y}^2+\tau_Q
   +C\left\{V_{\max}(\epsilon+\eta)
           +V_{\max}^2x/\mathcal N_{\infty}\right\}.
\end{align*}
Let the near-minimizer tolerance decrease to zero. This proves \eqref{eq:ti-prediction}. No common subject marginal distribution was required: each $\mu_i$ may differ. Consequently the step $\sum_iw_i^2\mu_i\le(\sum_iw_i^2)\epsilon_{f,y}^2$ is not valid in general, which is why \eqref{eq:ti-prediction} does not claim the scale $\mathcal N_2$.

For \eqref{eq:ti-gram}, every observed feature vector belongs to $\operatorname{Range}(\Sigma_{\rm st})$ almost surely. Indeed, each marginal is a positive-weight component of $\nu_N^+$; apply the zero second moment to a basis of $\ker \Sigma_{\rm st}$. On the nonzero eigenspace let $\widetilde w=\Sigma_{\rm st}^{\dagger/2}w_{\rm st}$ and define
\begin{align*}
 \Xi_i=\frac1n\sum_{t=1}^{T_i}\widetilde w(S_{i,t+1})\widetilde w(S_{i,t+1})^\top.
\end{align*}
These matrices are independent across subjects, positive semidefinite, and satisfy
\begin{align*}
 0\preceq \Xi_i\preceq C d_Sw_i\operatorname{Id}
 \preceq\frac{Cd_S}{\mathcal N_{\infty}}\operatorname{Id},
 \qquad\sum_i\mathbb E\Xi_i=\operatorname{Id}.
\end{align*}
The matrix Chernoff inequality \citep[Theorem~5.1.1]{Tropp2015} at relative error $1/2$ proves the two eigenvalue tails in \eqref{eq:ti-gram}. Transform back on $\operatorname{Range}(\Sigma_{\rm st})$ and retain the common null space. If an empirical finite-feature curvature diagnostic is used, integrating the same quadratic-form inequalities against its derivative coefficients gives the corresponding comparison; this optional consequence is not needed for the deterministic RKHS penalty.
\end{proof}
The three events can be intersected with failure at most the sum of their stated probabilities. In particular, this weighted refinement does not rely on independence between the policy update, fitted critics and Gram matrix.

\section{Deterministic classes and RKHS policy entropy}\label{sec:samedata-fqi}
For a critic $Q$ and policy $\pi$, write
\[
 y_{Q,\pi}(s,a,r,s')=r+\gamma Q\{s',\pi(s')\}.
\]
Step~\ref{line:main-bellman-label} uses the realized function
$y_m(\mathcal D)=y_{\widehat Q_{m-1},\widehat\pi_{m-1}}$ to form
$Y_{i,t,m}=y_m(\mathcal D)(S_{i,t},A_{i,t},R_{i,t},S_{i,t+1})$.
The following deterministic construction covers these data-dependent labels and their fitted-critic/target pairs.
\begin{assumption}[Classes fixed before fitting]\label{ass:q-policy-reachable}
Fix a deterministic clipped critic class $\mathcal Q_n$ and a deterministic label class
\[
 \mathcal Y_n\supset
 \{(s,a,r,s')\mapsto r+\gamma Q\{s',\pi(s')\}:Q\in\mathcal Q_n,\ \pi\in\Pi_{\mathcal H}(B_{\mathcal H})\}.
\]
The membership and localization event is $\mathcal R_n$ defined in Section~\ref{subsec:one-step-errors}: it requires $\widehat Q_m\in\mathcal Q_n$ and $\widehat\pi_m\in\Pi_{\mathcal H}(B_{\mathcal H})$ for $0\le m\le M$.
The event $\mathcal U_n$ is defined by the four uniform bounds in~\eqref{eq:main-uniform-error-event}, and $\mathcal E_n=\mathcal R_n\cap\mathcal U_n$. Sufficient concentration and computational events are established on these deterministic classes before fitted functions are substituted.
\end{assumption}
This is the fit-and-target uniformity used to justify data reuse in \citet{NguyenTangEtAl2022Besov}. The label family ranges over the deterministic norm ball containing every realized policy on the membership and localization event, so a fitted policy and critic can be substituted together after all iterations have reused the data.
For every deterministic $y\in\mathcal Y_n$, define its conditional target by
\[
 Q_y(s,a)=\int y(s,a,r,s')\Penv(dr,ds'\mid s,a),
 \qquad (s,a)\in\mathcal S\times\mathcal A.
\]
The environment kernel defines this function on the full admissible domain, including actions outside behavior support. Choose deterministic target approximants $f_y^\circ$ satisfying
\[
 \sup_{y\in\mathcal Y_n}\|f_y^\circ-Q_y\|_\infty\le \eta_{Q,n},
\]
where the supremum norm is over that full domain. Approximation only on behavior support would not control the approximation error at policy-selected actions.
For exact functional inputs, let $\mathcal F_n(y)$ be a deterministic class containing every critic that step~\ref{line:main-critic-fit} may return for label $y$. The pair family in Section~\ref{sec:main-spectral} is
\[
 \mathcal V_n=\{(f,f_y^\circ):y\in\mathcal Y_n,\ f\in\mathcal F_n(y)\}.
\]
With numerical functional scores, use the corresponding exact-score fit family in this definition and control the discrepancy between implemented and exact-score critics separately by $e_{Q,n}$, as in Corollary~\ref{cor:common-regression-transfer}. For KRR, Appendix~\ref{app:direct-krr} applies the secant argument to raw fits before clipping.
On $\mathcal R_n$, $y_m(\mathcal D)\in\mathcal Y_n$. Thus, on $\mathcal E_n$,
\[
 \max_m{E}_{\mathcal D}(y_m(\mathcal D))
       \le\sup_{y\in\mathcal Y_n}{E}_{\mathcal D}(y).
\]
Here ${E}_{\mathcal D}(y)$ denotes any error functional controlled uniformly in $y$ on $\mathcal U_n$. This is a pathwise substitution, and uses no new independence of fitted labels; the fitted critics are paired with the approximants of their realized labels in the same fixed $\mathcal V_n$.
If $\mathcal Q_n$ is enlarged to include idealized counterparts that evaluate functional integrals exactly, or other numerical approximants, the corresponding label cover is enlarged before fitting as well.

\begin{proposition}[Policy entropy from RKHS evaluation tails]\label{prop:rkhs-entropy}
Suppose $\{e_\ell\}$ is an orthonormal basis of the policy RKHS, uniform $L^2$ evaluation has norm at most $\kappa_{\mathcal H}$, and for some $p_{\rm eig}>0$,
\[
 \sup_s\sum_{\ell>m}\|e_\ell(s)\|_{L^2}^2\le C m^{-2p_{\rm eig}}.
\]
For every subset $\Pi$ of its radius-${B}_\Pi$ ball,
\[
 \log \operatorname{Cov}(\epsilon,\Pi,d_\Pi)
 \le C\epsilon^{-1/p_{\rm eig}}
                   \log(1+C/\epsilon),\qquad 0<\epsilon\le1,
\]
with constants depending on the radius and the displayed evaluation bounds.
\end{proposition}
The reproducing evaluation bound of \citet[Proposition~2.1]{MicchelliPontil2005} controls the finite coefficient block, and the assumed uniform tail controls the remaining directions. Their combination gives a cover of the entire RKHS ball; the truncation is a covering argument and leaves the policy domain unchanged.
\begin{proof}
Write $\pi=\sum_\ell\theta_\ell e_\ell$, $\sum_\ell\theta_\ell^2\le{B}_\Pi^2$.
Cauchy--Schwarz gives uniform tail error at most
$C^{1/2}{B}_\Pi m^{-p_{\rm eig}}$.
Choose $m$ so this is at most $\epsilon/4$.
On the first $m$ coefficients, evaluation is $\kappa_{\mathcal H}$-Lipschitz from Euclidean norm to $d_\Pi$.
A radius-$\epsilon/(4\kappa_{\mathcal H})$ Euclidean net in the radius-${B}_\Pi$ ball has at most
$(1+8\kappa_{\mathcal H}{B}_\Pi/\epsilon)^m$ points.
The truncation and cover therefore give error at most $\epsilon/2$.
For a restricted feasible class, replace each ball intersecting $\Pi$ by a representative from that intersection, at cost at most a factor two in radius. Taking logarithms proves the bound.
\end{proof}
Bounded kernel diagonal supplies the evaluation bound, whereas the tail condition supplies uniform approximation of the whole RKHS ball. These are separate properties.
A Mercer eigenvalue bound alone does not supply the uniform eigenfunction tail in the proposition.

For action-Lipschitz critics, the label map satisfies
\[
 \|y_{Q,\pi}-y_{Q',\pi'}\|_\infty
 \le\gamma\{\|Q-Q'\|_\infty+\mathrm{Lip}_n d_\Pi(\pi,\pi')\}.
\]
Combining this with the preceding cover gives full label and objective entropy over $\Pi$.
A finite parameter count may be used only for a verified numerical subclass; it is not the entropy of the full infinite-dimensional policy domain.

For the two-state fixed comparison ball, which is contained in a product of two fixed $H^2$ balls, Proposition~\ref{prop:spline-objective} also gives
$\log \operatorname{Cov}\{\epsilon,\Pi_{\mathcal H}(B_{\mathcal H}),d_\Pi\}\le C\epsilon^{-1/2}\log(1+C/\epsilon)$.
Indeed use $d_u\asymp\epsilon^{-1/2}$, cover the projected functional actions in an $L^2$-orthonormal basis of their $O(d_u)$-dimensional spline space, and take a product of the two covers.
The radius of that Euclidean ball is bounded by the continuous embedding of $\mathcal G$ into $L^2$; moving centers into the feasible class has the same factor-two cost.
For its actual Bellman labels an even smaller continuation-cube cover is available, as proved in Appendix~\ref{app:verification-examples}.

\begin{proposition}[Entropy of a full continuous-state policy RKHS ball]\label{prop:continuous-full-entropy}
Let $\mathcal S=[0,1]$, $\mathcal G=H^2(0,1)$ with the Sobolev norm of Proposition~\ref{prop:functional-action-rkhs}, and $\mathcal H=H^1([0,1];\mathcal G)$. Fix
\[
 \Pi_{\mathcal H}(B_{\mathcal H})
 =\{\pi\in\mathcal H:\|\pi\|_{\mathcal H}\le B_{\mathcal H},
          \sup_s\|\pi(s)\|_{L^2}\le{B}_A\}.
\]
For this entire fixed ball and $d_\Pi(\pi,\pi')=\sup_s\|\pi(s)-\pi'(s)\|_{L^2}$,
\begin{equation}
 \log \operatorname{Cov}\{\epsilon,\Pi_{\mathcal H}(B_{\mathcal H}),d_\Pi\}
 \le C\epsilon^{-5/2}\log(1+C/\epsilon),\qquad 0<\epsilon\le1.
 \label{eq:continuous-full-entropy}
\end{equation}
The constant depends on the fixed radii and the continuous embeddings, not on a computational dimension.
\end{proposition}
The Hilbert-valued RKHS evaluation property of \citet[Proposition~2.1]{MicchelliPontil2005} supplies the dimension-independent evaluation bound. Here the approximation in Appendix~\ref{subsec:continuous-rkhs-implementation} supplies a quantitative cover of the whole ball. The finite-rank spaces below are a device for proving that cover; policies, Bellman labels and the value benchmark still range over $\Pi_{\mathcal H}(B_{\mathcal H})$.
\begin{proof}
Let $\ProjOp_{h_s,d_u}$ be the tensor energy projection from Appendix~\ref{subsec:continuous-rkhs-implementation}. The state and functional-action projections commute and are contractions in $\mathcal H$, so
\[
 \|\ProjOp_{h_s,d_u}\pi\|_{\mathcal H}\le B_{\mathcal H},
 \qquad d_\Pi(\pi,\ProjOp_{h_s,d_u}\pi)
 \le C B_{\mathcal H}(h_s^{1/2}+d_u^{-2}).
\]
Choose $h_s\asymp\epsilon^2$ and $d_u\asymp\epsilon^{-1/2}$, with constants giving approximation error at most $\epsilon/8$. The image space has dimension
$D\le C h_s^{-1}d_u\le C\epsilon^{-5/2}$.
An $\mathcal H$-orthonormal coordinate system identifies its radius-$B_{\mathcal H}$ ball with a Euclidean ball. Uniform state evaluation gives
$d_\Pi(f,g)\le\kappa_{\mathcal H}\|f-g\|_{\mathcal H}$,
with fixed $\kappa_{\mathcal H}$. A Euclidean net of radius $\epsilon/(8\kappa_{\mathcal H})$ has at most $(1+16\kappa_{\mathcal H}B_{\mathcal H}/\epsilon)^D$ centers. Combining the projection error and this net covers $\Pi_{\mathcal H}(B_{\mathcal H})$ by ambient balls of radius $\epsilon/4$. For every ball meeting the feasible class, choose a center in that intersection; the radius at most doubles. Taking logarithms proves the claim, including the action-feasibility restriction.
\end{proof}

\section{Smooth policy actions and moment-based secant identification}\label{sec:general-oracle}

Norm localization and~\eqref{eq:main-policy-spectral-evaluation} supply the spectral radius in~\eqref{eq:main-policy-radius}, while observed actions in $L^2[0,1]$ satisfy the regression moment or envelope conditions. Identification is relative to critic differences and their second moments, with no local action probability condition. State transfer in Section~\ref{sec:bellman} remains a separate requirement.

Use the deterministic smooth class $\Pi_{\mathcal H}(B_{\mathcal H})$ defined in Section~\ref{sec:main-method}. All fitted critics and Bellman labels are covered on their full deterministic classes. Write
\[
 \|h\|_{\nuplus,\Pi_{\mathcal H}(B_{\mathcal H}),2}^2
 =\int\sup_{\pi\in\Pi_{\mathcal H}(B_{\mathcal H})}|h\{s,\pi(s)\}|^2\nuplus(ds).
\]
For a fixed orthonormal basis $\{\phi_j\}$ of $L^2([0,1])$, define
\[
 \|a\|_{\zeta}^2=\sum_{j\ge1}j^{2\zeta}|\langle a,\phi_j\rangle|^2,
 \qquad \BaseWeightOp_\alpha\phi_j=j^{-2\alpha}\phi_j,\quad \alpha>0.
\]
The spectral weight operator $\BaseWeightOp_\alpha$ sets the scale assigned to functional coordinates in the transfer argument. It does not assert that the observed action covariance is diagonal or has these eigenvalues. The directional comparison~\eqref{eq:secant-moment-lower} must be verified separately, and eigenvalue decay alone does not imply it. If an infinite-dimensional covariance dominates $\BaseWeightOp_\alpha$ on all of $L^2([0,1])$, finite second moments force $\alpha>1/2$.

For the state-free linear illustration $h(a)=\langle d,a\rangle$ in Section~\ref{sec:main-spectral}, the feature is the action itself. For $0<\vartheta\le1$ with $\alpha\vartheta\le\zeta$,
\[
 \|\BaseWeightOp_\alpha^{-\vartheta/2}a\|^2
 =\sum_{j\ge1}j^{2\alpha\vartheta}|\langle a,\phi_j\rangle|^2
 \le\|a\|_\zeta^2.
\]
Thus actions with $\|a\|_\zeta\le B$ have smooth-evaluation factor at most $B$. If the logged second-moment condition is $e^2:=\E h(A)^2\ge c\|\BaseWeightOp_\alpha^{1/2}d\|^2$ for $c>0$, weighted Cauchy--Schwarz and spectral interpolation give
\[
 |h(a)|\le B\|\BaseWeightOp_\alpha^{\vartheta/2}d\|
 \le B\|d\|^{1-\vartheta}\|\BaseWeightOp_\alpha^{1/2}d\|^\vartheta
 \le Bc^{-\vartheta/2}\|d\|^{1-\vartheta}e^\vartheta.
\]
This is the action-linear specialization of Proposition~\ref{thm:smooth-spectral-transfer}; smoothness controls evaluation, while the logged second-moment condition supplies information.

\subsection{Actual secants and evaluation at a prescribed accuracy}\label{subsec:actual-secant-accuracy}

All critic pairs in this section belong to a family $\mathcal V_n$ fixed before fitting. It contains every possible fitted critic and its population target approximant for each Bellman label under analysis. When numerical functional scores are used, the corresponding exact-score critic is the idealized version obtained by evaluating those integrals exactly rather than by quadrature; the numerical discrepancy is added separately. Neither observed prediction errors nor the fitted path are used to discard pairs. For $v=(f,f^\circ)\in\mathcal V_n$, put
\[
 \upsilon_v=\|f-f^\circ\|_{\nuplus,\Pi_{\mathcal H}(B_{\mathcal H}),2},\qquad
 e_v=\|f-f^\circ\|_{L^2(\nuX)}.
\]
For an evaluation accuracy $\varepsilon>0$, define
\begin{equation}
 C_n\{\varepsilon;\Pi_{\mathcal H}(B_{\mathcal H})\}
 =\sup_{v\in\mathcal V_n:\upsilon_v>\varepsilon}\frac{\upsilon_v^2}{e_v^2}.
 \label{eq:finite-accuracy-secant}
\end{equation}
A positive numerator divided by zero is $+\infty$, and the supremum of an empty set is zero.  This amplification concerns actual critic differences, rather than every vector in a feature space.  Accuracy $\varepsilon$ is an analysis parameter, not a modification of either algorithmic penalty.  Uniform control as $\varepsilon\downarrow0$ is not required.

For finite $C_\varepsilon=C_n\{\varepsilon;\Pi_{\mathcal H}(B_{\mathcal H})\}$, the definition separates two cases:
\begin{align*}
 \upsilon_v\le \varepsilon&\quad\Longrightarrow\quad \upsilon_v\le \varepsilon+\sqrt{C_\varepsilon}\,e_v,\\
 \upsilon_v>\varepsilon&\quad\Longrightarrow\quad e_v>0,\quad
               \upsilon_v\le\sqrt{C_\varepsilon}\,e_v.
\end{align*}
Consequently
\begin{equation}
 \upsilon_v\le \varepsilon+\sqrt{C_n\{\varepsilon;\Pi_{\mathcal H}(B_{\mathcal H})\}}\,e_v
 \qquad(v\in\mathcal V_n),
 \label{eq:finite-accuracy-transfer}
\end{equation}
and a deterministic bound $e_v\le\delta$ gives
\begin{equation}
 \upsilon_v\le\omega_n(\delta):=
 \inf_{\varepsilon>0:\,C_n\{\varepsilon;\Pi_{\mathcal H}(B_{\mathcal H})\}<\infty}
 \left[\varepsilon+\sqrt{C_n\{\varepsilon;\Pi_{\mathcal H}(B_{\mathcal H})\}}\,\delta\right].
 \label{eq:finite-accuracy-modulus}
\end{equation}
An empty set of pairs with $\upsilon_v>\varepsilon$ gives $C_\varepsilon=0$ and $\upsilon_v\le \varepsilon$ for every pair.
An empty admissible-accuracy set gives $\omega_n=+\infty$.
Along a sequence of fits on the fixed policy class, the prediction-error sequence $\delta_{Q,n}$ must satisfy $\omega_n(\delta_{Q,n})\to0$ for consistency; finiteness at each fixed accuracy alone is insufficient.

The coefficient in \eqref{eq:finite-accuracy-secant} is not assumed to be small without a verification.  The following directional moment and smooth-feature calculations supply quantitative sufficient conditions for it.

\subsection{Directional moments and smooth functional evaluation}

For each $v=(f,f^\circ)\in\mathcal V_n$, put $h_v=f-f^\circ$. Let a separable Hilbert space $\mathcal E_v$, a coefficient $d_v\in \mathcal E_v$, and a measurable feature map give the exact identity
\begin{equation*}
 h_v(s,a)=\langle d_v,\psi_v(s,a)\rangle,\qquad
 \Sigma_v=\mathbb E_{\nuX}[\psi_v(X)\otimes\psi_v(X)],\qquad
 e_v^2=\langle d_v,\Sigma_vd_v\rangle.
\end{equation*}
Here $(u\otimes w)z=u\langle w,z\rangle$ is a rank-one operator; thus $\Sigma_v z=\E_{\nuX}[\psi_v(X)\langle\psi_v(X),z\rangle]$. No feature mean is subtracted.
Assume $\mathbb E\|\psi_v(X)\|^2<\infty$. Let $\ProjOp_v$ be the orthogonal projection onto the closed span of feature values on the entire domain relevant to fitting and evaluation, not onto the span of the observed sample. Thus $\ProjOp_v^*=\ProjOp_v$, $\ProjOp_v^2=\ProjOp_v$, $\|\ProjOp_vz\|\le\|z\|$, and $\psi_v(s,a)=\ProjOp_v\psi_v(s,a)$, including at policy-selected actions outside behavior support.

\begin{assumption}[Information along actual secants and smooth features]\label{ass:secant-moment-smoothness}
For every pair, let $\WeightOp_v$ be the fixed smoothness-weight operator on $\mathcal E_v$, specified by the functional coordinates and critic representation. It is positive, injective and self-adjoint with norm at most one, so every nonzero direction receives a positive weight no larger than one. Uniformly on $\mathcal V_n$, require
\begin{equation}
 \langle d_v,\Sigma_vd_v\rangle
 \ge c_n\langle \ProjOp_vd_v,\WeightOp_v\ProjOp_vd_v\rangle,\qquad c_n>0.
 \label{eq:secant-moment-lower}
\end{equation}
For some $0<\vartheta\le1$, the candidate policies satisfy the smooth-feature bound
\begin{equation}
 \int\sup_{\pi\in\Pi_{\mathcal H}(B_{\mathcal H})}
       \|\WeightOp_v^{-\vartheta/2}\psi_v\{s,\pi(s)\}\|^2\nuplus(ds)
 \le {B}_{\psi,n}^2.
 \label{eq:secant-smooth-feature}
\end{equation}
The inverse is understood on its domain. All families, constants and suprema are fixed before fitting. If they are controlled on an event, the bound is uniform over a fixed class containing every possible fitted realization.
\end{assumption}
As in the Bellman-residual transfer coefficient of \citet[Definition~1]{XieEtAl2021BellmanPessimism}, information is required on the error class used by the analysis. Here that class consists of actual critic secants, and the smooth-feature condition identifies which evaluations at functional policy actions can be controlled from their behavior second moments.

\paragraph{Spectral powers.}\label{app:spectral-power-notation}
Let $\SpecOp_{\WeightOp}$ be the spectral resolution of the positive injective
contraction $\WeightOp=\WeightOp_v$, and $\mu_z(B)=\langle z,\SpecOp_{\WeightOp}(B)z\rangle$.
No bounded inverse is assumed. Spectral calculus gives
\begin{align*}
 \WeightOp^p&=\int_{(0,1]}t^p\SpecOp_{\WeightOp}(dt)\quad(p\ge0),\\
 \operatorname{Dom}(\WeightOp^{-\vartheta/2})
 &=\{z:\int_{(0,1]}t^{-\vartheta}d\mu_z(t)<\infty\},\qquad
 \|\WeightOp^{-\vartheta/2}z\|^2=\int_{(0,1]}t^{-\vartheta}d\mu_z(t).
\end{align*}
The last equality holds on that domain. In an eigenbasis it becomes
$\sum_j\varpi_j^{-\vartheta}|\langle z,e_j\rangle|^2$,
so the feature condition requires finite energy after upweighting directions with small smoothness weights.

Condition~\eqref{eq:secant-moment-lower} concerns the actual coefficient $d_v$ of each critic difference, rather than unrelated parameter directions. It is a joint state--action information condition, separate from policy smoothness and the eigenvalues of $\mathbb E[A\otimes A]$. Propositions~\ref{prop:krr-feature-smoothness} and~\ref{prop:ada-functional-secant} verify the feature bound from smoothness of policy-selected actions. These actions may lie outside behavior support.

A distribution-family-free sufficient verification of the directional moment condition is
\begin{equation}
 \Pr_{\nuX}\!\left\{
 |h_v(X)|\ge b_n^{\mathrm{dir}}
      \|\WeightOp_v^{1/2}\ProjOp_vd_v\|\right\}
 \ge p_n^{\mathrm{dir}}>0
 \label{eq:directional-moment-verifier}
\end{equation}
for pairs with nonzero weighted coefficient energy. It gives
$c_n=p_n^{\mathrm{dir}}(b_n^{\mathrm{dir}})^2$. Indeed, put $a_v=\|\WeightOp_v^{1/2}\ProjOp_vd_v\|$ and let $\mathsf{Obs}_v^{\mathrm{dir}}$ be the event inside \eqref{eq:directional-moment-verifier}. Nonnegativity of $h_v^2$ gives
\[
 e_v^2=\E h_v(X)^2
 \ge\E[h_v(X)^2\ind_{\mathsf{Obs}_v^{\mathrm{dir}}}]
 \ge(b_n^{\mathrm{dir}})^2a_v^2\Pr(\mathsf{Obs}_v^{\mathrm{dir}})
 \ge p_n^{\mathrm{dir}}(b_n^{\mathrm{dir}})^2a_v^2.
\]
When $a_v=0$ the required lower bound is immediate. This event measures whether the data distinguish two critics; the distinguishing observations need not be close to any policy-selected action. The quadratic-form comparison can instead be checked directly. Normality, independent action coordinates, and coordinatewise densities are not required.

\begin{proposition}[Secant transfer under uniform or accuracy-dependent information]
\label{thm:smooth-spectral-transfer}\label{prop:scale-directional-information}
\textup{(i) Uniform information.} Under Assumption~\ref{ass:secant-moment-smoothness},
\begin{equation*}
 \upsilon_v\le {B}_{\psi,n}c_n^{-\vartheta/2}
          \|d_v\|^{1-\vartheta}e_v^\vartheta.
\end{equation*}
If $\|d_v\|\le {B}_{Q,n}<\infty$, then, for $\varepsilon>0$,
\begin{equation}
 C_n\{\varepsilon;\Pi_{\mathcal H}(B_{\mathcal H})\}
 \le {B}_{\psi,n}^{2/\vartheta}c_n^{-1}
 {B}_{Q,n}^{2(1-\vartheta)/\vartheta}
 \varepsilon^{-2(1-\vartheta)/\vartheta}.
 \label{eq:derived-accuracy-leverage}
\end{equation}
At $\vartheta=1$ no uniform coefficient radius is required, and the exact evaluation-to-prediction ratio is at most ${B}_{\psi,n}^2/c_n$.  If all $d_v=0$, all secants vanish; expressions with exponent zero are then unnecessary.

\textup{(ii) Accuracy-dependent information.} Retain the smooth-feature condition and $\|d_v\|\le {B}_{Q,n}$, but replace the uniform moment condition by the following weaker, accuracy-dependent requirement.  For a fixed $\varepsilon_0>0$ and $0<\varepsilon\le \varepsilon_0$, only pairs with $\upsilon_v>\varepsilon$ must satisfy
\begin{equation}
 e_v^2\ge c_n(\varepsilon)
              \langle \ProjOp_vd_v,\WeightOp_v\ProjOp_vd_v\rangle,
 \qquad c_n(\varepsilon)>0.
 \label{eq:scale-directional-information}
\end{equation}
Then \eqref{eq:derived-accuracy-leverage} holds with $c_n(\varepsilon)$ in place of $c_n$. Also, for every pair with $e_v\le\delta$,
\begin{equation}
 \upsilon_v\le\inf_{0<\varepsilon\le \varepsilon_0}
 \left\{\varepsilon+
 {B}_{\psi,n}c_n(\varepsilon)^{-\vartheta/2}
 {B}_{Q,n}^{1-\vartheta}\delta^\vartheta\right\}.
 \label{eq:scale-directional-modulus}
\end{equation}
If $c_n(\varepsilon)\ge c_*\varepsilon^{\beta_{\mathrm{info}}}$ on a fixed accuracy interval, where $\beta_{\mathrm{info}}\ge0$ is the information-decay exponent, and $c_*^{-1}$, ${B}_{\psi,n}$ and ${B}_{Q,n}$ are bounded, this is
$O(\delta^{\vartheta_*})$ as $\delta\downarrow0$, where
\begin{equation*}
 \vartheta_*:=\frac{2\vartheta}{2+\beta_{\mathrm{info}}\vartheta}.
\end{equation*}
The directional probability verification \eqref{eq:directional-moment-verifier} need only hold for the same pairs, and its constants may depend on $\varepsilon$.  This gives $c_n(\varepsilon)=p_n^{\mathrm{dir}}(\varepsilon)
\{b_n^{\mathrm{dir}}(\varepsilon)\}^2$ without a local-positivity argument.
\end{proposition}
The exponent follows the Hilbert-scale interpolation principle in \citet[equation~(1.4)]{EggerHofmann2018}, applied here to the coefficient and smooth policy-evaluation feature of each critic difference. Accuracy-dependent information permits weaker identification of differences whose policy-evaluation effect is already small, and the balance in part~(ii) quantifies the resulting change in the transfer exponent.
\begin{proof}
Fix $v$, and write $\WeightOp=\WeightOp_v$, $z=\ProjOp_vd_v$ and $\psi=\psi_v(s,a)$. The full-domain projection gives $h_v(s,a)=\langle z,\psi\rangle$ and $\|z\|\le\|d_v\|$.

First justify the weighted inner product. For the spectral projections $\SpecOp_k=\SpecOp_{\WeightOp}([1/k,1])$, the inverse power is bounded on $\operatorname{Range}(\SpecOp_k)$, and
\begin{align*}
 \langle \SpecOp_kz,\SpecOp_k\psi\rangle
 &=\langle \WeightOp^{\vartheta/2}\SpecOp_kz,\WeightOp^{-\vartheta/2}\SpecOp_k\psi\rangle.
\end{align*}
Injectivity gives $\SpecOp_kz\to z$, and the assumed inverse-power domain gives $\WeightOp^{-\vartheta/2}\SpecOp_k\psi\to \WeightOp^{-\vartheta/2}\psi$. Passing to the limit and applying Cauchy--Schwarz yields
\begin{align*}
 |h_v(s,a)|^2
 &\le\langle z,\WeightOp^\vartheta z\rangle\|\WeightOp^{-\vartheta/2}\psi_v(s,a)\|^2.
\end{align*}
For $0<\vartheta<1$, H\"older's inequality with exponents $1/\vartheta$ and $1/(1-\vartheta)$ in the spectral measure $\mu_z$ gives
\begin{align*}
 \langle z,\WeightOp^\vartheta z\rangle
 &=\int t^\vartheta\,d\mu_z(t)
 \le\left(\int t\,d\mu_z\right)^\vartheta
      \left(\int1\,d\mu_z\right)^{1-\vartheta}\\
 &=\langle z,\WeightOp z\rangle^\vartheta\|z\|^{2(1-\vartheta)}
 \le c_n^{-\vartheta}e_v^{2\vartheta}\|d_v\|^{2(1-\vartheta)}.
\end{align*}
Taking the policy supremum before integrating the pointwise bound now gives
\begin{align*}
 \upsilon_v^2
 &\le\langle z,\WeightOp^\vartheta z\rangle
       \int\sup_{\pi\in\Pi_{\mathcal H}(B_{\mathcal H})}\|\WeightOp^{-\vartheta/2}\psi_v(s,\pi(s))\|^2\nuplus(ds)\\
 &\le {B}_{\psi,n}^2c_n^{-\vartheta}e_v^{2\vartheta}\|d_v\|^{2(1-\vartheta)}.
\end{align*}
This proves the transfer inequality. At $\vartheta=1$, use $\langle z,\WeightOp z\rangle\le c_n^{-1}e_v^2$ directly, without a coefficient radius. No commutation of $\ProjOp_v$ and $\WeightOp_v$ has been used.

For $C_{\rm tr}={B}_{\psi,n}c_n^{-\vartheta/2}{B}_{Q,n}^{1-\vartheta}>0$ and $\upsilon_v>\varepsilon$, the transfer bound implies
\begin{align*}
 e_v^2\ge(\upsilon_v/C_{\rm tr})^{2/\vartheta},\qquad
 \frac{\upsilon_v^2}{e_v^2}
 \le C_{\rm tr}^{2/\vartheta}\upsilon_v^{-2(1-\vartheta)/\vartheta}
 \le C_{\rm tr}^{2/\vartheta}\varepsilon^{-2(1-\vartheta)/\vartheta}.
\end{align*}
Taking the supremum proves the amplification bound. Zero coefficients or $C_{\rm tr}=0$ give vanishing secants, so there is no positive-numerator case to consider.

\noindent\textit{Accuracy-dependent information.} Set $\Lambda_\varepsilon^{\rm tr}={B}_{\psi,n}c_n(\varepsilon)^{-\vartheta/2}{B}_{Q,n}^{1-\vartheta}$. On the restricted family $\upsilon_v>\varepsilon$, the preceding spectral calculation applies with $c_n(\varepsilon)$ and gives
\begin{align*}
 \upsilon_v\le \Lambda_\varepsilon^{\rm tr}e_v^\vartheta,\qquad
 \upsilon_v^2/e_v^2\le (\Lambda_\varepsilon^{\rm tr})^{2/\vartheta}
                    \varepsilon^{-2(1-\vartheta)/\vartheta}.
\end{align*}
For $e_v\le\delta$, separating $\upsilon_v\le \varepsilon$ from $\upsilon_v>\varepsilon$ gives $\upsilon_v\le \varepsilon+\Lambda_\varepsilon^{\rm tr}\delta^\vartheta$. Taking the infimum proves the transfer-envelope bound.

If $c_n(\varepsilon)\ge c_*\varepsilon^{\beta_{\mathrm{info}}}$, put $q_{\rm bal}=\beta_{\mathrm{info}}\vartheta/2$ and absorb the bounded factors into $C_*$. When $q_{\rm bal}>0$, choose $\varepsilon=\delta^{\vartheta/(1+q_{\rm bal})}\le \varepsilon_0$ for sufficiently small $\delta$. Then
\begin{align*}
 \upsilon_v
 &\le \delta^{\vartheta/(1+q_{\rm bal})}+C_*\delta^\vartheta
                      \delta^{-q_{\rm bal}\vartheta/(1+q_{\rm bal})}
 =(1+C_*)\delta^{2\vartheta/(2+\beta_{\mathrm{info}}\vartheta)}.
\end{align*}
When $q_{\rm bal}=0$, let $r\downarrow0$ to obtain $\upsilon_v\le C_*\delta^\vartheta$; when $\delta=0$, the same limit gives $\upsilon_v=0$.

Finally, on the same accuracy-restricted family, the directional probability event with threshold $b_{\rm dir}\|\WeightOp_v^{1/2}\ProjOp_vd_v\|$ and probability $p$ gives
\begin{align*}
 e_v^2
 &\ge\E\!\left[h_v(X)^2\ind\{|h_v(X)|\ge b_{\rm dir}\|\WeightOp_v^{1/2}\ProjOp_vd_v\|\}\right]
 \ge p b_{\rm dir}^2\langle \ProjOp_vd_v,\WeightOp_v\ProjOp_vd_v\rangle.
\end{align*}
Thus its constants may depend on $r$ without requiring observations near the policy-selected action.
\end{proof}

\begin{remark}[Actual secants versus full feature leverage]\label{prop:restricted-secant-direction}
For $\lambda_{\rm lev}>0$, Proposition~\ref{thm:smooth-spectral-transfer} and
$x^\vartheta y^{1-\vartheta}\le\lambda_{\rm lev}^{\vartheta-1}(x+\lambda_{\rm lev} y)$ give
\begin{equation}
 \Lambda_{\Pi_{\mathcal H}(B_{\mathcal H})}(\lambda_{\rm lev})
 :=\sup_{v:d_v\ne0}\frac{\upsilon_v^2}{e_v^2+\lambda_{\rm lev}\|d_v\|^2}
 \le {B}_{\psi,n}^2c_n^{-\vartheta}\lambda_{\rm lev}^{\vartheta-1}.
 \label{eq:derived-secant-leverage}
\end{equation}
To verify the scalar inequality used here, set $u=x/(x+\lambda_{\rm lev} y)\in[0,1]$ when $x+\lambda_{\rm lev} y>0$. Then
\[
 \frac{x^\vartheta y^{1-\vartheta}}{x+\lambda_{\rm lev} y}
 =\lambda_{\rm lev}^{\vartheta-1}u^\vartheta(1-u)^{1-\vartheta}
 \le\lambda_{\rm lev}^{\vartheta-1};
\]
zero denominators are treated by continuity. Applying this to $x=e_v^2$ and $y=\|d_v\|^2$ proves \eqref{eq:derived-secant-leverage}.
The full-feature coefficient based on the inverse second-moment operator
$\Lambda_v\{\lambda_{\rm lev};\Pi_{\mathcal H}(B_{\mathcal H})\}=\int\sup_{\pi\in\Pi_{\mathcal H}(B_{\mathcal H})}
\langle\psi_v\{s,\pi(s)\},(\Sigma_v+\lambda_{\rm lev}\operatorname{Id})^{-1}\psi_v\{s,\pi(s)\}\rangle\nuplus(ds)$
instead controls all feature directions. Here $\operatorname{Id}$ is the identity, and $\Sigma_{v,{\rm lev}}=\Sigma_v+\lambda_{\rm lev}\operatorname{Id}\succeq\lambda_{\rm lev}\operatorname{Id}$ means $\langle z,\Sigma_{v,{\rm lev}} z\rangle\ge\lambda_{\rm lev}\|z\|^2$ for every $z$, so $\Sigma_{v,{\rm lev}}^{-1}$ is bounded. Weighted Cauchy--Schwarz and its equality case give the variational identity
\[
 |\langle d_v,\psi\rangle|^2
 \le\langle d_v,\Sigma_{v,{\rm lev}} d_v\rangle\langle\psi,\Sigma_{v,{\rm lev}}^{-1}\psi\rangle,
 \qquad
 \langle\psi,\Sigma_{v,{\rm lev}}^{-1}\psi\rangle
 =\sup_{u\ne0}\frac{|\langle u,\psi\rangle|^2}{\langle u,\Sigma_{v,{\rm lev}} u\rangle}.
\]
Indeed, write the inner product as $\langle \Sigma_{v,{\rm lev}}^{1/2}u,\Sigma_{v,{\rm lev}}^{-1/2}\psi\rangle$. Equality holds for $u=\Sigma_{v,{\rm lev}}^{-1}\psi$ when $\psi\ne0$. Taking the policy supremum in the weighted Cauchy--Schwarz bound and integrating gives $\upsilon_v^2\le\Lambda_v\{\lambda_{\rm lev};\Pi_{\mathcal H}(B_{\mathcal H})\}(e_v^2+\lambda_{\rm lev}\|d_v\|^2)$. The variational supremum ranges over every $u$, not only $d_v$. Under the stronger comparison $\Sigma_v\succeq c_n\ProjOp_v\WeightOp_v\ProjOp_v$, the spectral argument in the proof of Proposition~\ref{thm:smooth-spectral-transfer} applies to every $\ProjOp_vu$. Combining it with $x^\vartheta y^{1-\vartheta}\le\lambda_{\rm lev}^{\vartheta-1}(x+\lambda_{\rm lev} y)$ yields pointwise
\[
 \langle\psi,\Sigma_{v,{\rm lev}}^{-1}\psi\rangle
 \le c_n^{-\vartheta}\lambda_{\rm lev}^{\vartheta-1}
             \|\WeightOp_v^{-\vartheta/2}\psi\|^2.
\]
The feature-envelope bound then proves the same numerical bound for $\Lambda_v$. At $\vartheta=1$, the spectral integrands $(t+\lambda_{\rm lev})^{-1}$ increase to $t^{-1}$ as $\lambda_{\rm lev}\downarrow0$; monotone convergence gives the zero-ridge bound ${B}_{\psi,n}^2/c_n$ on the finite-energy domain $\operatorname{Range}(\Sigma_v^{1/2})=\{\Sigma_v^{1/2}z:z\in \mathcal E_v\}$, the range of the square-root operator. None of these full-feature conditions is required by the value theorem. The second-moment operator is that of the integrated feature; averaging Jacobian second-moment operators or comparing sorted eigenvalues does not preserve this directional information.
\end{remark}
Restricting a transfer ratio to a prescribed error family also appears for Bellman residuals in \citet{XieEtAl2021BellmanPessimism} and for transition-model errors in \citet{UeharaSun2022PartialCoverage}. The distinction here is between the coefficient direction associated with each exact critic secant and the variational supremum over the entire feature space, which explains the two information conditions in the display above.

\subsection{A nonlinear information verification}\label{subsec:nonlinear-information}
\begin{proposition}[Accuracy-dependent information in a nonlinear critic family]\label{prop:nonlinear-accuracy}
Consider one state, an $L^2$ action ball, and a fixed smooth RKHS policy ball containing the zero action. Let the behavior action satisfy $|\langle A,\phi_1\rangle|=\varsigma_A>0$, as in the signed-series design of Appendix~\ref{app:verification-examples}. For a fixed $\kappa>1$, set
\[
 z(a)=\min\{1,|\langle a,\phi_1\rangle|/\varsigma_A\},\qquad
 F_\ell(a)=\ell\{1-z(a)\}+\ell^\kappa z(a),\quad 0\le \ell\le1.
\]
For exact pair features $\psi(a)=(1-z(a),z(a))$, coefficients
$d_{\ell,u}=(\ell-u,\ell^\kappa-u^\kappa)$ and smoothness-weight operator $\WeightOp=\operatorname{Id}_2$, the smooth-feature constant is one. There is no positive uniform directional-information constant over all pairs, but pairs with evaluation error $\upsilon_{\ell,u}>r$ satisfy the moment inequality with
\[
 c(r)={r^{2\kappa-2}\over(1+\kappa^2)\kappa^{2\kappa-2}}.
\]
Every pair obeys $\upsilon_{\ell,u}\le \kappa\epsilon_{\ell,u}^{1/\kappa}$. Full-feature zero-ridge leverage at the zero action is infinite.
\end{proposition}
Function-class-relative coefficients in \citet{UeharaSun2022PartialCoverage,XieEtAl2021BellmanPessimism} motivate restricting which errors require information. This construction additionally makes the information threshold depend on the desired evaluation accuracy: actual coefficient pairs obey a nonlinear relation even though their ambient feature second-moment operator is singular. It verifies the $\beta_{\mathrm{info}}>0$ branch of Proposition~\ref{thm:smooth-spectral-transfer}, complementing the infinite-dimensional spectral mechanism in Example~\ref{ex:two-state-control}.
\begin{proof}
Since $z(A)=1$ and $z(0)=0$,
\[
 \Sigma=\operatorname{diag}(0,1),\qquad
 \psi(0)=(1,0),\qquad\sup_a\|\psi(a)\|^2\le1.
\]
For $\ell>u$, put $\Delta=\ell-u$ and $\Delta_\kappa=\ell^\kappa-u^\kappa$. Integration gives
$\Delta^\kappa\le\Delta_\kappa\le \kappa\Delta$. Hence
\[
 \epsilon_{\ell,u}=\Delta_\kappa,\qquad
 \Delta\le\upsilon_{\ell,u}\le \kappa\Delta,\qquad
 \|d_{\ell,u}\|^2\le(1+\kappa^2)\Delta^2.
\]
If $\upsilon_{\ell,u}>r$, then $\Delta>r/\kappa$, so
\[
 {\epsilon_{\ell,u}^2\over\|d_{\ell,u}\|^2}
 \ge{\Delta^{2\kappa-2}\over1+\kappa^2}\ge c(r).
\]
Swapping endpoints gives all other pairs. Also
$\upsilon_{\ell,u}\le \kappa\Delta\le \kappa\Delta_\kappa^{1/\kappa}$.
For $(F_\ell,F_0)$ the errors are exactly $\upsilon=\ell$ and $\epsilon=\ell^\kappa$, so
$\epsilon^2/\|d\|^2=\ell^{2\kappa-2}/(1+\ell^{2\kappa-2})\to0$.
The ridge feature leverage at zero is
$\langle\psi(0),(\Sigma+\lambda_{\rm lev}\operatorname{Id})^{-1}\psi(0)\rangle=\lambda_{\rm lev}^{-1}$.
Finally $\vartheta=1$ and $\beta_{\mathrm{info}}=2\kappa-2$ yield $\vartheta_*=1/\kappa$, agreeing with the direct bound. The family is compact in the uniform norm and uniformly action-Lipschitz, since $z$ is $\varsigma_A^{-1}$-Lipschitz and $\ell\in[0,1]$.
\end{proof}

\subsection{From behavior prediction to policy evaluation}\label{subsec:regression-policy-transfer}

Regression bounds hold uniformly over deterministic classes containing every possible Bellman label and corresponding fit, allowing the realized same-data labels to be substituted \citep{NguyenTangEtAl2022Besov}. Appendix~\ref{sec:samedata-fqi} constructs these classes. Target approximation must hold on the full state--action domain: approximation only on behavior support cannot control evaluation at policy actions outside that support. This subsection justifies the error conversion~\eqref{eq:main-common-error} used in the main examples.

For a deterministic Bellman label $y$, let $Q_y$ be its conditional mean defined by the common reward--transition kernel on the whole admissible domain, including policy-selected actions outside observed support. For each label choose a deterministic exact-score approximant $f_y^\circ$ with
$\sup_y\|f_y^\circ-Q_y\|_\infty\le \eta_{Q,n}$.

\begin{corollary}[A common regression-to-policy error bound]\label{cor:common-regression-transfer}
Suppose one event gives simultaneously for the full deterministic label class
\[
 \|\widehat f_y-Q_y\|_{L^2(\nuX)}\le \delta_{Q,n},\qquad
 \|\widehat f_y-F_y\|_\infty\le e_{Q,n},
\]
where $F_y$ is the exact-score version of the fit. Let a deterministic family $\mathcal V_n$ contain every possible pair $(F_y,f_y^\circ)$ substituted on this event. Then
\begin{equation}
 \sup_y\|\widehat f_y-Q_y\|_{\nuplus,\Pi_{\mathcal H}(B_{\mathcal H}),2}
 \le \eta_{Q,n}+e_{Q,n}
      +\omega_n(\delta_{Q,n}+\eta_{Q,n}+e_{Q,n}).
 \label{eq:general-regression-transfer}
\end{equation}
Every quantitative upper bound for the secant transfer envelope in \eqref{eq:finite-accuracy-transfer} may replace the last term.  In particular, under Assumption~\ref{ass:secant-moment-smoothness} and $\|d_v\|\le {B}_{Q,n}$, the sharper bound is
\begin{gather}
 \sup_y\|\widehat f_y-Q_y\|_{\nuplus,\Pi_{\mathcal H}(B_{\mathcal H}),2}
 \le{E}_n(\delta_{Q,n},\eta_{Q,n},e_{Q,n},{B}_{Q,n}),
 \notag\\
 {E}_n(\delta,\eta,e,{B}_d)
 :=\eta+e+{B}_{\psi,n}c_n^{-\vartheta/2}{B}_d^{1-\vartheta}(\delta+\eta+e)^\vartheta.
 \label{eq:common-error-modulus}
\end{gather}
Under the scale-dependent condition of Proposition~\ref{prop:scale-directional-information}\textup{(ii)}, another valid bound is
\begin{equation*}
 {E}_n^{\mathrm{sc}}(\delta,\eta,e,{B}_d)
 :=\eta+e+\inf_{0<\varepsilon\le \varepsilon_0}
 \big\{\varepsilon+ {B}_{\psi,n}c_n(\varepsilon)^{-\vartheta/2}
                    {B}_d^{1-\vartheta}(\delta+\eta+e)^\vartheta\big\}.
\end{equation*}
At $\vartheta=1$ omit the coefficient factor; if ${B}_d=0$ use $\eta+e$ directly.
\end{corollary}
The bound combines the logged prediction error $\delta_{Q,n}$, full-domain target-approximation error $\eta_{Q,n}$, and numerical input error $e_{Q,n}$. On $\mathcal E_n$, it supplies the critic-update error in Assumption~\ref{ass:main-one-step} by taking
\begin{equation}
 \varepsilon_{Q,n}=\eta_{Q,n}+e_{Q,n}
       +\omega_n(\delta_{Q,n}+\eta_{Q,n}+e_{Q,n}).
 \label{eq:main-common-error}
\end{equation}
In the notation of~\eqref{eq:main-uniform-error-event}, these component bounds give $\delta_n=\delta_{Q,n}+\eta_{Q,n}+e_{Q,n}$ and hence $\varepsilon_{Q,n}=\eta_{Q,n}+e_{Q,n}+\omega_n(\delta_n)$. The KRR specialization instead uses its raw fit with $f_y^\circ=Q_y$, $\delta_n=\delta_{K,n}$, and exact inputs; Proposition~\ref{thm:krr-common-prediction} verifies both pair membership and this prediction bound before clipping.
This separates statistical prediction from its conditional-stability conversion, following the stability-versus-noise distinction in \citet{EggerHofmann2018}. Full-domain approximation and numerical errors enter on both sides of the secant comparison, while the label-uniform event permits their use at every same-data fitted-$Q$ step \citep{NguyenTangEtAl2022Besov}.
\begin{proof}
Work on the simultaneous event and take $v=(F_y,f_y^\circ)$. The behavior norm of this exact-score secant is bounded by the triangle inequality:
\begin{align*}
 e_v
 &\le\|F_y-\widehat f_y\|_{L^2(\nuX)}
     +\|\widehat f_y-Q_y\|_{L^2(\nuX)}
     +\|Q_y-f_y^\circ\|_{L^2(\nuX)}\\
 &\le e_{Q,n}+\delta_{Q,n}+\eta_{Q,n}.
\end{align*}
The full-domain bounds also dominate the policy-evaluation norm, so
\begin{align*}
 \|\widehat f_y-Q_y\|_{\nuplus,\Pi_{\mathcal H}(B_{\mathcal H}),2}
 &\le\|\widehat f_y-F_y\|_{\nuplus,\Pi_{\mathcal H}(B_{\mathcal H}),2}\\
 &\quad+\|F_y-f_y^\circ\|_{\nuplus,\Pi_{\mathcal H}(B_{\mathcal H}),2}\\
 &\quad+\|f_y^\circ-Q_y\|_{\nuplus,\Pi_{\mathcal H}(B_{\mathcal H}),2}\\
 &\le e_{Q,n}+\upsilon_v+\eta_{Q,n}.
\end{align*}
Apply the finite-accuracy bound~\eqref{eq:finite-accuracy-transfer} to the displayed behavior-norm bound for $e_v$ and substitute it into the displayed policy-evaluation bound. Proposition~\ref{thm:smooth-spectral-transfer}\textup{(i)--(ii)} give the two quantitative alternatives with the same inputs. If ${B}_{Q,n}=0$, the secant vanishes; at $\vartheta=1$, its coefficient factor is absent. The events and bounds cover the full deterministic label class, so taking $\sup_y$ is valid.
\end{proof}

\paragraph{Prediction-rate substitution.}
If $\delta_{Q,n}+\eta_{Q,n}+e_{Q,n}=\widetilde O(N^{-b_Q})$, $b_Q>0$, and ${B}_{\psi,n}$, $c_n^{-1}$ and (when $\vartheta<1$) ${B}_{Q,n}$ are bounded, \eqref{eq:common-error-modulus} gives
\begin{equation*}
 \sup_y\|\widehat f_y-Q_y\|_{\nuplus,\Pi_{\mathcal H}(B_{\mathcal H}),2}
 =\widetilde O(N^{-\vartheta b_Q}).
\end{equation*}
Under \eqref{eq:scale-directional-information} with
$c_n(\varepsilon)\ge c_*\varepsilon^{\beta_{\mathrm{info}}}$ and bounded $c_*^{-1}$,
replace $\vartheta$ by $\vartheta_*=2\vartheta/(2+\beta_{\mathrm{info}}\vartheta)$.
More generally, for $\kappa_{\rm tr}\ge0$, $C_n\{\varepsilon;\Pi_{\mathcal H}(B_{\mathcal H})\}
\le C\varepsilon^{-\kappa_{\rm tr}}$ with bounded $C$ gives exponent
$2b_Q/(2+\kappa_{\rm tr})$. These follow by substitution into \eqref{eq:scale-directional-modulus}, or by choosing $\varepsilon$ of order $(\delta_{Q,n}+\eta_{Q,n}+e_{Q,n})^{2/(2+\kappa_{\rm tr})}$ in \eqref{eq:finite-accuracy-modulus}; for $\kappa_{\rm tr}=0$ let $\varepsilon\downarrow0$. All conditions are evaluated along the chosen iteration sequence $M=M_n$, including its full critic and label classes.

If some actual difference satisfies $e_v=0$ but $\upsilon_v>0$, then
$C_n\{\varepsilon;\Pi_{\mathcal H}(B_{\mathcal H})\}=\infty$ for
$0<\varepsilon<\upsilon_v$. Such a fixed identification failure is retained by the theorem, not removed by a smoothing penalty or by taking a proof-only ridge parameter positive. Appendix~\ref{app:verification-examples} verifies the information condition for the full linear critic kernel of Example~\ref{ex:two-state-control}.

Under the deterministic-class convention of Appendix~\ref{sec:samedata-fqi}, each realized label is substituted into the same uniform bound on the fixed smooth class $\Pi_{\mathcal H}(B_{\mathcal H})$.

\section{General finite-sample value guarantee and proof}\label{sec:finite-sample}
Theorem~\ref{thm:main-value} is the stable-propagation specialization of the following finite-iteration statement.
\begin{theorem}[General finite-iteration value bound]\label{thm:finite-value-general}
Under the data and measurability conditions of Section~\ref{sec:main-method} and Assumption~\ref{ass:mdp-compatibility}, suppose the uniform critic and policy error bounds in Assumption~\ref{ass:main-one-step} hold on $\mathcal E_n$ from Section~\ref{subsec:one-step-errors}. This event also places all fitted critics in $\mathcal Q_n$ and all returned policies in $\Pi_{\mathcal H}(B_{\mathcal H})$. Under Assumption~\ref{ass:main-state-transfer}, put
\[
 \chi=\gamma C_B,\qquad G_M(\chi)=\sum_{j=0}^{M-1}\chi^j.
\]
On $\mathcal E_n$, for the prescribed finite $M\ge1$,
\begin{equation}
 J^\star-J(\widehat\pi_M)
 \le a_{\mathcal H}+\frac{C_{\rm occ}(\widehat\pi_M)}{1-\gamma}
 \{2\Vmax\chi^M+2G_M(\chi)\varepsilon_{Q,n}
                      +(1+2\chi G_M(\chi))\varepsilon_{\pi,n}\}.
 \label{eq:main-value}
\end{equation}
The critic bound $\varepsilon_{Q,n}$ can be supplied by~\eqref{eq:main-common-error} under the verified secant conditions. The bound holds for every finite $\chi$, with probability at least $\Pr(\mathcal E_n)$.
\end{theorem}
Beyond $a_{\mathcal H}$, the terms in~\eqref{eq:main-value} are finite-iteration error, accumulated critic-update error and policy-improvement error. Stable propagation bounds their accumulation uniformly in $M$.
\begin{proof}
We first prove the value bound relative to $\Pi_{\mathcal H}(B_{\mathcal H})$ and then add the approximation gap $a_{\mathcal H}=J^\star-J_{\mathcal H}^\star$.
On $\mathcal E_n$, put
\[
 {E}_m^Q=\|Q^\dagger-\widehat Q_m\|_{\nu_+,\Pi_{\mathcal H}(B_{\mathcal H}),2},\qquad
 h_m(s)=\sup_{\pi\in\Pi_{\mathcal H}(B_{\mathcal H})}|(Q^\dagger-\widehat Q_m)\{s,\pi(s)\}|.
\]
Allow nonuniform error bounds $\varepsilon_{Q,m}$ and $\varepsilon_{\pi,m}$ in the argument.
The Bellman decomposition is
\begin{align*}
 Q^\dagger-\widehat Q_m
 &=\{\mathbb BQ^\dagger-\mathbb B\widehat Q_{m-1}\}
   +\{\mathbb B\widehat Q_{m-1}
                    -\mathbb B^{\widehat\pi_{m-1}}\widehat Q_{m-1}\}\\
 &\quad+\{\mathbb B^{\widehat\pi_{m-1}}\widehat Q_{m-1}-\widehat Q_m\}.
\end{align*}
The Bellman-optimality difference is bounded by $\gamma Ph_{m-1}$, and the greedy-policy difference is bounded by $\gamma Pg_{m-1}$. Minkowski and the one-step state condition give, with $\chi=\gamma C_B$,
\begin{equation*}
 {E}_m^Q\le\chi{E}_{m-1}^Q+\varepsilon_{Q,m}+\chi\varepsilon_{\pi,m-1}.
\end{equation*}
Induction yields
\[
 {E}_M^Q\le\chi^ME_0^Q+\sum_{m=1}^M\chi^{M-m}\varepsilon_{Q,m}
                  +\chi\sum_{m=0}^{M-1}\chi^{M-1-m}\varepsilon_{\pi,m}.
\]
Clipping and probability normalization give $E_0^Q\le\Vmax$.

For any clipped $Q$ and $\pi\in\Pi_{\mathcal H}(B_{\mathcal H})$, adding and subtracting $Q$ gives
\[
 g_{Q^\dagger,\pi}(s)\le g_{Q,\pi}(s)
                 +2\sup_{\pi'\in\Pi_{\mathcal H}(B_{\mathcal H})}|(Q^\dagger-Q)\{s,\pi'(s)\}|.
\]
Apply Lemma~\ref{lem:occupancy}, which gives $J_{\mathcal H}^\star\le\E_\rho V^\dagger$ by a supremum over the localized class. The occupancy condition for the two displayed residuals yields
\begin{equation*}
 J_{\mathcal H}^\star-J(\pi)\le
 \frac{C_{\rm occ}(\pi)}{1-\gamma}
 \{\|g_{Q,\pi}\|_{L^2(\nu_+)}+2\|Q^\dagger-Q\|_{\nu_+,\Pi_{\mathcal H}(B_{\mathcal H}),2}\}.
\end{equation*}
Set $(Q,\pi)=(\widehat Q_M,\widehat\pi_M)$ and substitute the recursion:
\begin{align*}
 J_{\mathcal H}^\star-J(\widehat\pi_M)
 \le\frac{C_{\rm occ}(\widehat\pi_M)}{1-\gamma}
 \bigg\{\varepsilon_{\pi,M}+2\chi^ME_0^Q
 &+2\sum_{m=1}^M\chi^{M-m}\varepsilon_{Q,m}\nonumber\\
 &+2\chi\sum_{m=0}^{M-1}\chi^{M-1-m}\varepsilon_{\pi,m}\bigg\}.
\end{align*}
Uniform bounds and the identity
$J^\star-J(\widehat\pi_M)=a_{\mathcal H}+J_{\mathcal H}^\star-J(\widehat\pi_M)$
give precisely~\eqref{eq:main-value}.
No policy attaining $J_{\mathcal H}^\star$ has been selected in the proof, and the only comparison with the base class $\Pi$ is the explicit approximation gap $a_{\mathcal H}$.
All events are intersected once; their dependence is unrestricted and their failure probabilities can be added.
The proof is valid for every finite $M$. If $\chi=\gamma C_B$ is bounded above by a fixed $\chi_0<1$, independent of $N$ and $M$, then
$G_M(\chi)\le(1-\chi_0)^{-1}$ and the initialization term is bounded by $\Vmax\chi^M=\Vmax(\gamma C_B)^M$. With fixed $\gamma$ and uniformly bounded occupancy constants, this gives~\eqref{eq:main-value-stable} and proves Theorem~\ref{thm:main-value}.
\end{proof}

\paragraph{Direct policy errors.}\label{princ:direct-policy-gap-substitution}
Any bound on the actual $\Pi_{\mathcal H}(B_{\mathcal H})$ defect may supply $\varepsilon_{\pi,m}$.
In particular, a product policy class with separated statewise optimization can use Appendix~\ref{app:verification-examples}'s direct bound. For a coupled policy class, the policy-compatibility gap in Appendix~\ref{sec:policy-update} remains in the objective route.

\section{KRR prediction and moment-based evaluation}\label{app:direct-krr}

We first record kernel feature calculations, then establish uniform prediction, fitted-norm and conditional evaluation bounds. Kernel input regularity controls covering and numerical errors; evaluation additionally requires the directional information and smooth-feature conditions. All statements use the full deterministic label and prediction classes.
Let $\mathcal F$ be a separable RKHS on $\mathcal S\times\mathcal A$, with measurable kernel,
$\sup_xK(x,x)\le\kappa_K^2$. The kernel section $K_x=K(\cdot,x)\in\mathcal F$ represents evaluation through $f(x)=\langle f,K_x\rangle_{\mathcal F}$ (the reproducing property). Assume
\begin{equation}
 \|K_{(s,a)}-K_{(s,a')}\|_{\mathcal F}
 \le \mathrm{Lip}_K\|a-a'\|_{L^2}.\label{eq:krr-common-regularity}
\end{equation}
The second-moment operator is $\Sigma_K=\int K_x\otimes K_x\,\nuX(dx)$ and its empirical version is
$\widehat\Sigma_K=n^{-1}\sum_{i,t}K_{X_{i,t}}\otimes K_{X_{i,t}}$.
For each deterministic label $y$, let $\widetilde f_{\varrho,y}$ be an approximate minimizer of
\[
 \widehat\Psi_{\varrho,y}(f)
 =n^{-1}\sum_{i,t}\{Y_{i,t}(y)-f(X_{i,t})\}^2+\varrho\|f\|_{\mathcal F}^2,
\]
with $\varrho>0$ and objective gap at most $\tau_K\ge0$. FQI returns
$\widehat f_y=\clip_{[0,\Vmax]}\widetilde f_{\varrho,y}$.

\begin{proposition}[Functional kernel secants and a smoothness verification]\label{prop:krr-feature-smoothness}
For an RKHS $\mathcal F$, the un-clipped difference $f-g$ has
$\mathcal E=\mathcal F$, $d=f-g$, and $\psi(s,a)=K_{(s,a)}$.
Thus the same moment theorem applies with a deterministic spectral weight operator
$\WeightOp_K$ whenever its quadratic-form and smooth-feature conditions are verified.

A direct infinite-dimensional example uses finitely many states and a kernel length scale $\ell_K>0$:
\[
 K\{(s,a),(s',a')\}=\ind\{s=s'\}
                    (1+\ell_K^{-2}\langle a,a'\rangle_{L^2})^m,
 \qquad m\in\mathbb N.
\]
For each state use the ambient direct sum of symmetric tensor powers $\{L^2([0,1])\}^{\odot k}$, $0\le k\le m$. Here $\{L^2([0,1])\}^{\odot k}$ is the subspace of the $k$-fold Hilbert tensor product unchanged by permutations of its factors, and $a^{\otimes k}=a\otimes\cdots\otimes a$. The conventions $\{L^2([0,1])\}^{\odot0}=\mathbb R$, $a^{\otimes0}=1$ and $\BaseWeightOp_\alpha^{\otimes0}=1$ give the intercept. Direct sums, denoted by $\bigoplus$, stack orthogonal blocks, so their squared norms add. RKHS coefficients are isometrically identified with the closed feature span, while the spectral weight operator acts on the ambient space; invariance of that span is not required. The features are
\[
 \psi_m(a)=\bigoplus_{k=0}^m
       {\binom mk}^{1/2}\ell_K^{-k}a^{\otimes k},\qquad
 \WeightOp_K=\bigoplus_{s}\bigoplus_{k=0}^m \BaseWeightOp_\alpha^{\otimes k}.
\]
Let $\ProjOp_K$ denote the projection onto that closed feature span.
For $\vartheta=\min\{1,\zeta/\alpha\}$ and $\|a\|_{\zeta}\le {B}_{\Pi,\zeta}$,
\begin{equation*}
 \|\WeightOp_K^{-\vartheta/2}K_{(s,a)}\|^2
 \le(1+\ell_K^{-2}{B}_{\Pi,\zeta}^2)^m.
\end{equation*}
For a uniform policy radius ${B}_{\Pi,\zeta}$, this gives ${B}_{\psi,n}^2=(1+\ell_K^{-2}{B}_{\Pi,\zeta}^2)^m$.
For only the integrated radius \eqref{eq:main-policy-radius}, the full functional linear kernel ($m=1$) still has
${B}_{\psi,n}^2=1+\ell_K^{-2}{B}_{\Pi,\zeta}^2$.
For $m>1$, also retain a deterministic uniform bound
$\sup_{s,\pi}\|\pi(s)\|_{\zeta}\le\overline{B}_{\Pi,\zeta}$; then a valid integrated factor is
\[
 {B}_{\psi,n}^2=(1+\ell_K^{-2}\overline{B}_{\Pi,\zeta}^2)^{m-1}(1+\ell_K^{-2}{B}_{\Pi,\zeta}^2).
\]
Alternatively, directly bound the integrated $2m$-th policy-smoothness moment.  These are not finite-dimensional action discretizations: infinitely many action coordinates enter the kernel.
\end{proposition}
Functional-input kernels such as the linear and Gaussian examples in \citet[Remark~2]{WangWongZhangChan2026FFTEE} encode functions directly through Hilbert-space geometry. The tensor calculation here identifies how functional-policy regularity controls the feature energy required by secant transfer, with infinitely many action coordinates retained in the kernel.
\begin{proof}
For $d=f-g\in\mathcal F$, the reproducing property gives
\[
 \begin{aligned}
 (f-g)(s,a)&=\langle d,K_{(s,a)}\rangle_{\mathcal F},\\
 \|f-g\|_{L^2(\nuX)}^2
 &=\E_{\nuX}\langle d,K_X\rangle^2
 =\left\langle d,\E_{\nuX}[K_X\otimes K_X]d\right\rangle
 =\langle d,\Sigma_K d\rangle.
 \end{aligned}
\]
For polynomial features, the tensor inner-product identity and binomial theorem yield
\[
 \begin{aligned}
 \langle a^{\otimes k},(a')^{\otimes k}\rangle
 &=\prod_{q=1}^k\langle a,a'\rangle=\langle a,a'\rangle^k,\\
 \langle\psi_m(a),\psi_m(a')\rangle
 &=\sum_{k=0}^m\binom mk\ell_K^{-2k}\langle a,a'\rangle^k
 =(1+\ell_K^{-2}\langle a,a'\rangle)^m,\\
 \langle e_s\otimes\psi_m(a),e_{s'}\otimes\psi_m(a')\rangle
 &=\ind\{s=s'\}(1+\ell_K^{-2}\langle a,a'\rangle)^m.
 \end{aligned}
\]
The coefficient lies in the closed feature span. Writing $a_j=\langle a,\phi_j\rangle$, spectral calculus gives
\[
 \begin{aligned}
 \|(\BaseWeightOp_\alpha^{\otimes k})^{-\vartheta/2}a^{\otimes k}\|^2
 &=\sum_{j_1,\ldots,j_k\ge1}
       \prod_{q=1}^k j_q^{2\alpha\vartheta}|a_{j_q}|^2\\
 &=\left(\sum_{j\ge1}j^{2\alpha\vartheta}|a_j|^2\right)^k
 =\|a\|_{\alpha\vartheta}^{2k},\\
 \|\WeightOp_K^{-\vartheta/2}\psi_m(a)\|^2
 &=\sum_{k=0}^m\binom mk\ell_K^{-2k}\|a\|_{\alpha\vartheta}^{2k}\\
 &=(1+\ell_K^{-2}\|a\|_{\alpha\vartheta}^2)^m
 \le(1+\ell_K^{-2}\|a\|_{\zeta}^2)^m
 \qquad(\alpha\vartheta\le \zeta).
 \end{aligned}
\]
The infinite sums follow from finite coordinate truncations by monotone convergence.
Thus $\|a\|_{\zeta}\le {B}_{\Pi,\zeta}$ gives the pointwise bound $(1+\ell_K^{-2}{B}_{\Pi,\zeta}^2)^m$.

For an integrated radius, set $x(s)=\ell_K^{-2}\sup_{\pi\in\Pi_{\mathcal H}(B_{\mathcal H})}\|\pi(s)\|_{\zeta}^2$.
Then
\[
 \begin{aligned}
 \int x\,d\nuplus&\le\ell_K^{-2}{B}_{\Pi,\zeta}^2,\\
 \int(1+x)\,d\nuplus&\le1+\ell_K^{-2}{B}_{\Pi,\zeta}^2
 &&(m=1),\\
 \int(1+x)^m\,d\nuplus
 &\le(1+\bar x)^{m-1}\int(1+x)\,d\nuplus
 &&(m>1,\ x\le\bar x),\\
 &\le(1+\ell_K^{-2}\overline{B}_{\Pi,\zeta}^2)^{m-1}(1+\ell_K^{-2}{B}_{\Pi,\zeta}^2)
 &&(\bar x=\ell_K^{-2}\overline{B}_{\Pi,\zeta}^2).
 \end{aligned}
\]
The feature second-moment comparison remains a separate condition; no invariance of the feature span under $\WeightOp_K$ is used.
\end{proof}

For the full functional-distance Gaussian kernel, the following Lipschitz calculation controls entropy and numerical errors, not the spectral feature condition. Its explicit secant rate still requires a compatible $\WeightOp_K$ and a feature second-moment comparison; otherwise \eqref{eq:general-regression-transfer} retains the amplification at the prescribed accuracy. A Gaussian kernel imposes no Gaussian distribution on observed actions.

\begin{proposition}[Kernel input regularity for entropy and numerical errors]\label{prop:krr-smooth-verification}
For $f,g\in\mathcal F$, condition \eqref{eq:krr-common-regularity} implies
\[
 \operatorname{Lip}_a(f-g)\le \mathrm{Lip}_K\|f-g\|_{\mathcal F},\qquad
 \operatorname{Lip}_a(\clip f)\le \mathrm{Lip}_K\|f\|_{\mathcal F}.
\]
The full functional-distance Gaussian kernel
\[
 K\{(s,a),(s',a')\}=k_S(s,s')
       \exp\{-\|a-a'\|_{L^2}^2/(2\ell_K^2)\},\qquad
 \sup_s k_S(s,s)\le\kappa_S^2,
\]
satisfies \eqref{eq:krr-common-regularity} with $\mathrm{Lip}_K=\kappa_S/\ell_K$.
Adding an action-independent intercept leaves this constant unchanged; composing with a bounded linear score map multiplies it by that map's norm.
\end{proposition}
Kernel-section regularity also controls functional-input reconstruction errors in \citet[Assumption~6]{WangWongZhangChan2026FFTEE}. Here the same calculation supplies the action-Lipschitz constants needed for joint critic--policy covers and numerical objective comparisons, while the preceding proposition addresses the different spectral energy used for evaluation.
\begin{proof}
The reproducing property and Cauchy--Schwarz give
\[
 \begin{aligned}
 |(f-g)(s,a)-(f-g)(s,a')|
 &=|\langle f-g,K_{(s,a)}-K_{(s,a')}\rangle|\\
 &\le\|f-g\|_{\mathcal F}
       \|K_{(s,a)}-K_{(s,a')}\|_{\mathcal F}\\
 &\le \mathrm{Lip}_K\|f-g\|_{\mathcal F}\|a-a'\|_{L^2},\\
 |\clip f(s,a)-\clip f(s,a')|
 &\le|f(s,a)-f(s,a')|
 \qquad\text{(projection)}\\
 &\le \mathrm{Lip}_K\|f\|_{\mathcal F}\|a-a'\|_{L^2}.
 \end{aligned}
\]
For the Gaussian kernel, put $x=\|a-a'\|_{L^2}^2/(2\ell_K^2)$. Then
\[
 \begin{aligned}
 \|K_{(s,a)}-K_{(s,a')}\|_{\mathcal F}^2
 &=K((s,a),(s,a))+K((s,a'),(s,a'))\\
 &\quad-2K((s,a),(s,a'))\\
 &=2k_S(s,s)(1-e^{-x})
 =2k_S(s,s)\int_0^x e^{-t}\,dt\\
 &\le2\kappa_S^2x
 =\kappa_S^2\ell_K^{-2}\|a-a'\|_{L^2}^2,\\
 \|K_{(s,a)}-K_{(s,a')}\|_{\mathcal F}
 &\le(\kappa_S/\ell_K)\|a-a'\|_{L^2}.
 \end{aligned}
\]
An action-independent intercept contributes $c(s)+c(s)-2c(s)=0$.
For a bounded linear score map $\operatorname{Score}$, the same calculation gives
\[
 \|K_{(s,a)}-K_{(s,a')}\|_{\mathcal F}
 \le(\kappa_S/\ell_K)\|\operatorname{Score}(a-a')\|
 \le(\kappa_S/\ell_K)\|\operatorname{Score}\|_{\mathrm{op}}\|a-a'\|_{L^2}.
\]
\end{proof}

\subsection{Uniform prediction and a controlled fitted norm}\label{subsec:krr-uniform-prediction}

The cluster bound controls the score and second-moment operator uniformly over the label class. Suppose
\begin{equation}
 Q_y\in\mathcal F,\quad \sup_y\|Q_y\|_{\mathcal F}\le{B}_F,
 \quad 0\le Y(y),Q_y\le\Vmax.\label{eq:krr-target-ball}
\end{equation}
This is a regression approximation condition, not an identification condition of the general FQI theorem.  Let
${H}_Y(\eta)=\log \operatorname{Cov}(\eta,\mathcal Y_n,\|\cdot\|_\infty)$,
with net centers in the label class, and use the independent-subject scale $N$.
For $0<\delta<1$ define
\begin{gather}
 s_{K,n}=2\kappa_K\eta+C\kappa_K\Vmax
       \sqrt{\{H_Y(\eta)+\log(8/\delta)+1\}/N},\label{eq:krr-unwhitened-score}\\
 v_{K,n}=C\kappa_K^2\sqrt{\{\log(8/\delta)+1\}/N},\label{eq:krr-unwhitened-covariance}\\
 {B}_{K,n}={B}_F+s_{K,n}/\varrho+\sqrt{\tau_K/\varrho},\label{eq:krr-fitted-radius}\\
 \delta_{K,n}^2=\varrho {B}_F^2+2({B}_{K,n}+{B}_F)s_{K,n}+\tau_K
                         +v_{K,n}({B}_{K,n}+{B}_F)^2.\label{eq:krr-new-prediction-radius}
\end{gather}
The universal constant $C$ is chosen large enough for the following simultaneous event.

\begin{proposition}[Uniform KRR prediction, fitted norm and conditional evaluation]
\label{thm:krr-common-prediction}\label{cor:krr-common-evaluation}
\textup{(i) Prediction and fitted norm.} Under Assumptions~\ref{ass:mdp-compatibility} and~\ref{ass:bounded-clusters},
\eqref{eq:krr-target-ball}, finite ${H}_Y(\eta)$, and minimization of $\widehat\Psi_{\varrho,y}$ to objective gap $\tau_K$, with probability at least $1-\delta$,
\[
 \sup_y\|\widetilde f_{\varrho,y}\|_{\mathcal F}\le {B}_{K,n},
 \qquad
 \sup_y\qLtwo{\widetilde f_{\varrho,y}-Q_y}\le \delta_{K,n},\qquad\sup_y\qLtwo{\widehat f_y-Q_y}\le \delta_{K,n}.
\]
The probability refers to the full deterministic label class and hence permits all same-data FQI iterations.

\textup{(ii) Evaluation at policy-selected actions.} Fix before fitting a deterministic family containing all raw fitted-critic/Bellman-target pairs covered by the event in part (i). Their clipped error at policy-selected actions is at most $\omega_n(\delta_{K,n})$. If the actual raw kernel secants satisfy Assumption~\ref{ass:secant-moment-smoothness}, then
\[
 \sup_y\|\widehat f_y-Q_y\|_{\nuplus,\Pi_{\mathcal H}(B_{\mathcal H}),2}
 \le {B}_{\psi,n}c_n^{-\vartheta/2}({B}_{K,n}+{B}_F)^{1-\vartheta}\delta_{K,n}^{\vartheta}.
\]
Under Proposition~\ref{prop:scale-directional-information}\textup{(ii)}, replace the right-hand side by the bound \eqref{eq:scale-directional-modulus} at $\delta=\delta_{K,n}$ and ${B}_{Q,n}={B}_{K,n}+{B}_F$.  The raw prediction radius is proved before clipping; neither evaluation accuracy nor a proof-only ridge changes the KRR tuning parameter.
\end{proposition}
The fit-and-target uniformity of \citet{NguyenTangEtAl2022Besov} is combined here with RKHS ridge fitting: a single subject-level score and operator event controls prediction and the raw fitted norm for every Bellman label. The norm bound supplies the coefficient radius needed for the subsequent regression-to-policy-evaluation transfer, linking the regression calculation to the functional-action geometry.
\begin{proof}
First control the score uniformly over labels. Put
$\epsilon_{i,t,y}=Y_{i,t,y}-Q_y(X_{i,t})$ and
$\mathbf z_{i,y}^{\rm sub}=(N/n)\sum_t\epsilon_{i,t,y}K_{X_{i,t}}$, so that
$\mathbf z_y=N^{-1}\sum_i \mathbf z_{i,y}^{\rm sub}$. Conditional compatibility gives
$\E \mathbf z_{i,y}^{\rm sub}=0$, and $\|\mathbf z_{i,y}^{\rm sub}\|\le {B}_{\rm sc}:=\kappa_K\Vmax$.
Independence across subjects, followed by the bounded-differences inequality
\citep{BoucheronLugosiMassart2013}, yields
\begin{align*}
 \E\|\mathbf z_y\|^2&\le {B}_{\rm sc}^2/N,\\
 \Pr\!\left\{\|\mathbf z_y\|>{B}_{\rm sc}/\sqrt N+{B}_{\rm sc}\sqrt{2x/N}\right\}&\le e^{-x}.
\end{align*}
Indeed, replacing one subject changes the norm by at most $2{B}_{\rm sc}/N$.
Use $x={H}_Y(\eta)+\log(8/\delta)$ and a union bound on a deterministic
$\eta$-label net. The common kernel implies
$\|Q_y-Q_{y'}\|_\infty\le\|y-y'\|_\infty$, whence
$\|\mathbf z_y-\mathbf z_{y'}\|\le2\kappa_K\eta$ for a net representative.
Thus $\sup_y\|\mathbf z_y\|\le s_{K,n}$ on the stated score event.

Next write $\widehat\Sigma_K=N^{-1}\sum_i \CovOp_i$, where
$\CovOp_i=(N/n)\sum_tK_{X_{i,t}}\otimes K_{X_{i,t}}$.
The identity $\|K_x\otimes K_x\|_{\mathrm{HS}}=\|K_x\|^2$ gives
$\|\CovOp_i\|_{\mathrm{HS}}\le \kappa_{\rm cov}:=\kappa_K^2$.
The independent centered Hilbert--Schmidt sums satisfy
$\E\|\widehat\Sigma_K-\Sigma_K\|_{\mathrm{HS}}^2\le \kappa_{\rm cov}^2/N$, and replacement
of one subject changes this norm by at most $2\kappa_{\rm cov}/N$.
The preceding bounded-differences calculation with $x=\log(8/\delta)$ gives
\begin{align*}
 \|\widehat\Sigma_K-\Sigma_K\|_{\mathrm{op}}
 \le\|\widehat\Sigma_K-\Sigma_K\|_{\mathrm{HS}}\le v_{K,n}.
\end{align*}
Intersect the score and second-moment events. Their total failure is at most $\delta$.

For the exact ridge minimizer, the normal equation and spectral calculus give
\begin{align*}
 \widehat f_{\varrho,y}
 &=(\widehat\Sigma_K+\varrho\operatorname{Id})^{-1}\widehat\Sigma_K Q_y
       +(\widehat\Sigma_K+\varrho\operatorname{Id})^{-1}\mathbf z_y,\\
 \|\widehat f_{\varrho,y}\|&\le{B}_F+s_{K,n}/\varrho,
 \qquad\|\widetilde f_{\varrho,y}-\widehat f_{\varrho,y}\|^2
       \le\tau_K/\varrho.
\end{align*}
Here $(\widehat\Sigma_K+\varrho\operatorname{Id})^{-1}\widehat\Sigma_K$ has norm at most one;
the bound for $\|\widetilde f_{\varrho,y}-\widehat f_{\varrho,y}\|^2$ follows by completing the square in the ridge objective,
whose excess is at least $\varrho\|\widetilde f_{\varrho,y}-\widehat f_{\varrho,y}\|^2$.
Consequently $\|\widetilde f_{\varrho,y}\|\le {B}_{K,n}$ uniformly in $y$.

Finally, set $e_y=\widetilde f_{\varrho,y}-Q_y$. Comparing the approximate
ridge fit with $Q_y$ and discarding the nonnegative fitted penalty gives
\begin{align*}
 \langle e_y,\widehat\Sigma_K e_y\rangle
 &\le2\langle e_y,\mathbf z_y\rangle+\varrho {B}_F^2+\tau_K,\\
 \|e_y\|_{L^2(\nuX)}^2
 &=\langle e_y,\widehat\Sigma_K e_y\rangle+\langle e_y,(\Sigma_K-\widehat\Sigma_K)e_y\rangle\\
 &\le2({B}_{K,n}+{B}_F)s_{K,n}+\varrho {B}_F^2+\tau_K
       +v_{K,n}({B}_{K,n}+{B}_F)^2=\delta_{K,n}^2.
\end{align*}
We used $\|e_y\|\le {B}_{K,n}+{B}_F$. Since $Q_y\in[0,\Vmax]$,
$|\clip(\widetilde f_{\varrho,y})-Q_y|\le|e_y|$ pointwise.
This proves the clipped prediction bound on the same full-label event.

For evaluation, the same raw difference $e_y=\widetilde f_{\varrho,y}-Q_y$ satisfies
$\|e_y\|_{L^2(\nuX)}\le \delta_{K,n}$ and $\|e_y\|_{\mathcal F}\le {B}_{K,n}+{B}_F$.
The pointwise clipping inequality already proved gives
\begin{align*}
 \sup_y\|\widehat f_y-Q_y\|_{\nuplus,\Pi_{\mathcal H}(B_{\mathcal H}),2}
 &\le\sup_y\|e_y\|_{\nuplus,\Pi_{\mathcal H}(B_{\mathcal H}),2}
 \le\omega_n(\delta_{K,n}).
\end{align*}
The last inequality is \eqref{eq:finite-accuracy-modulus} on the deterministic raw-secant family.
Substitution of these prediction and coefficient radii into the uniform or accuracy-dependent part of
Proposition~\ref{thm:smooth-spectral-transfer} proves the two quantitative alternatives.
No clipped critic is treated as an RKHS element, and no new probability event is introduced.

\end{proof}

\subsection{Ridge calibration and explicit prediction rates}
\begin{corollary}[Boundary and interior ridge calibration]
\label{cor:rev-krr-boundary}\label{cor:krr-common-rate}
Retain the full-domain target ball and the simultaneous score and operator event of
Proposition~\ref{thm:krr-common-prediction}, with bounded target norm and kernel envelope.
Use the radii \eqref{eq:krr-fitted-radius}--\eqref{eq:krr-new-prediction-radius}.

\textup{(i) Bounded ratios, including boundary tuning.}
If $\varrho_n\to0$, $s_{K,n}+v_{K,n}=O(\varrho_n)$ and
$\tau_{K,n}=O(\varrho_n)$, then
${B}_{K,n}=O(1)$ and $\delta_{K,n}=O(\varrho_n^{1/2})$.
Under bounded raw-secant information and smooth-feature constants, evaluation is
$O(\varrho_n^{\bar\vartheta/2})$, where
$\bar\vartheta=\vartheta$ for uniform information and
$\bar\vartheta=2\vartheta/(2+\beta_{\mathrm{info}}\vartheta)$ when
$c_n(\varepsilon)\ge c_*\varepsilon^{\beta_{\mathrm{info}}}$ on a fixed interval, with fixed $c_*>0$.

\textup{(ii) Vanishing ratios under interior tuning.}
Suppose in addition that $\mathrm{Lip}_K$ is bounded, $N\to\infty$, and the full label entropy at a polynomially small $\eta_n$ satisfies
${H}_Y(\eta_n)=\widetilde O(N^{p_Y^{\rm ent}})$ for $0\le p_Y^{\rm ent}<1$.
For polynomially small confidence levels, take
\begin{align*}
 \varrho_n=N^{-p_\varrho},\qquad
 0<p_\varrho<(1-p_Y^{\rm ent})/2,\qquad
 \eta_n=o(\varrho_n),\quad\tau_{K,n}=o(\varrho_n).
\end{align*}
Then ${B}_{K,n}={B}_F+o(1)$, $\delta_{K,n}=\widetilde O(\varrho_n^{1/2})$, and under bounded uniform-information transfer constants the evaluation error is
$\widetilde O_p(\varrho_n^{\vartheta/2})$.
\end{corollary}
Classical spectral KRR rates track effective dimension as ridge decreases \citep{CaponnettoDeVito2007}, with the effective-dimension bound corrected by \citet{Sutherland2017}. The present sufficient calibration instead tracks the full-label score and operator bounds, making explicit the ridge scale that keeps all fitted secants inside one bounded coefficient ball.
\begin{proof}
The radius expressions can be written as
\begin{align}
 {B}_{K,n}&={B}_F+\frac{s_{K,n}}{\varrho_n}
                 +\sqrt{\frac{\tau_{K,n}}{\varrho_n}},\nonumber\\
 \frac{\delta_{K,n}^2}{\varrho_n}
 &={B}_F^2+2({B}_{K,n}+{B}_F)\frac{s_{K,n}}{\varrho_n}
      +\frac{\tau_{K,n}}{\varrho_n}
      +({B}_{K,n}+{B}_F)^2\frac{v_{K,n}}{\varrho_n}.
 \label{eq:rev-krr-radii}
\end{align}
For part (i), every displayed ratio is bounded. Equation~\eqref{eq:rev-krr-radii} directly bounds both ${B}_{K,n}$ and $\delta_{K,n}^2/\varrho_n$.
The conditional evaluation part of Proposition~\ref{thm:krr-common-prediction} then gives
$O(\varrho_n^{\bar\vartheta/2})$; it needs boundedness of ${B}_{K,n}$, not convergence to ${B}_F$.

For part (ii), put $p_\varrho^+=(1-p_Y^{\rm ent})/2>p_\varrho$. The score and second-moment bounds give
\begin{align*}
 s_{K,n}&=O(\eta_n)+\widetilde O(N^{-p_\varrho^+}),&
 v_{K,n}&=\widetilde O(N^{-1/2}).
\end{align*}
Because $p_\varrho<p_\varrho^+\le1/2$, these powers dominate the logarithms.
Thus $s_{K,n}/\varrho_n,v_{K,n}/\varrho_n,\tau_{K,n}/\varrho_n\to0$.
The radius identity in~\eqref{eq:rev-krr-radii} gives ${B}_{K,n}={B}_F+o(1)$, and its prediction identity gives $\delta_{K,n}^2/\varrho_n={B}_F^2+o(1)$.
Apply the same raw-secant and clipping bound on events of probability $1-o(1)$.
\end{proof}

For example, if the entropy of the full label class fixed before fitting at the chosen polynomial net resolution is $\widetilde O(N^{p_Y^{\rm ent}})$, $0\le p_Y^{\rm ent}<1$, choose a deterministic logarithmic majorant $\kappa_{\log,n}$ and
\begin{align*}
 \varrho_n=\kappa_{\log,n}N^{-(1-p_Y^{\rm ent})/2},\qquad
 s_{K,n}+v_{K,n}\le C\varrho_n,\qquad \tau_{K,n}\le C\varrho_n.
\end{align*}
This gives $\delta_{K,n}=\widetilde O(N^{-(1-p_Y^{\rm ent})/4})$. The majorant must dominate the actual entropy and confidence factors. The label cover is evaluated at the unconditional pre-fit radius $\{(V_{\max}^2+\tau_{K,n})/\varrho_n\}^{1/2}$, not at the sharper fitted radius subsequently proved on the event. Thus this example is a compatible deterministic tuning condition, not an implicit equation solved with a realized entropy or prediction error.

The deterministic class used to define the labels cannot be chosen retrospectively as the observed small-norm fitted path.  Comparing the empirical objective with zero gives the unconditional radius
$\|\widetilde f_{\varrho,y}\|_{\mathcal F}\le\{(\Vmax^2+\tau_K)/\varrho\}^{1/2}$.
Its clipped ball contains all iterates. The corresponding policy-indexed label class defines ${H}_Y$ in Section~\ref{subsec:krr-uniform-prediction}. The sharper radius ${B}_{K,n}$ is then derived on the simultaneous event of Proposition~\ref{thm:krr-common-prediction}.

For general policy classes, the full label entropy is supplied by the product cover in Appendix~\ref{sec:samedata-fqi}; it includes the actual RKHS policy cover. A finite-dimensional parameter count applies only when that is the declared implementation class. A full functional-distance Gaussian kernel needs its own full-domain entropy and information verification. Example~\ref{ex:two-state-control} instead provides a two-dimensional continuation-label cover for the full functional linear kernel.

The full-sample mixing extension is Proposition~\ref{prop:mix-krr}; the retained-block alternative retains its omission cost in Appendix~\ref{subsec:mix-retained}.

\subsection{An explicit KRR value rate}\label{subsec:krr-value-rate}
Retain the bounded target RKHS ball, uniformly bounded kernel, and exact functional inputs of Corollary~\ref{cor:rev-krr-boundary}. Let $p_K\in[0,1)$ be an exponent such that the full deterministic Bellman-label and policy-objective log covering numbers are $\widetilde O(N^{p_K})$ at uniform-norm resolutions $O(\varrho_n)$. These covers must contain all possible fits and labels before the data are observed, using the unconditional fitted radius described above. For a sufficiently large deterministic logarithmic factor $\kappa_{\log,n}$, choose
\[
 \varrho_n=\kappa_{\log,n}N^{-(1-p_K)/2}.
\]
The factor must dominate the actual covering and confidence logarithms so that $s_{K,n}+v_{K,n}=O(\varrho_n)$ on an event with probability tending to one. Suppose the KRR regularized empirical squared-loss objective gap satisfies $\tau_{K,n}=O(\varrho_n)$, the policy optimization error satisfies $\tau_{\pi,n}=O(\varrho_n)$, and $\lambda_{\pi,n}=O(\varrho_n)$. Require fixed-radius policy localization separately, as in~\eqref{eq:calibrated-policy-radius}.

Let $\Delta_n=\sup_{Q\in\mathcal Q_n}\Delta(Q)$, using the clipped deterministic fitted class. Assume the raw fit/target differences satisfy the uniform or accuracy-dependent transfer conditions of Proposition~\ref{thm:krr-common-prediction}, with bounded factors and exponent $\vartheta_*$ from Section~\ref{sec:main-spectral}. Under the bounded occupancy and stable state-propagation conditions of Theorem~\ref{thm:main-value}, for $M\ge c\log N$ with sufficiently large fixed $c$,
\[
 J^\star-J(\widehat\pi_M)
 \le a_{\mathcal H}+\widetilde O_p\!\left(
 N^{-\vartheta_*(1-p_K)/4}+\sqrt{\Delta_n}\right).
\]
This is the independent-subject version of the rate in Example~\ref{subsec:main-krr}; its mixing version follows below.
\begin{proof}
The boundary calibration gives $B_{K,n}=O(1)$ and $\delta_{K,n}=O(\varrho_n^{1/2})$. The raw-secant transfer bound and clipping give $\varepsilon_{Q,n}=O(\varrho_n^{\vartheta_*/2})$. Applying~\eqref{eq:equal-objective} to the deterministic policy-objective cover yields $u_n=O(\varrho_n)$, since both its net resolution and its concentration remainder have this order. Equation~\eqref{eq:main-policy-error} then gives
\[
 \varepsilon_{\pi,n}=O(\varrho_n^{1/2}+\sqrt{\Delta_n}).
\]
As $0<\vartheta_*\le1$, the critic term dominates $\varrho_n^{1/2}$. The iteration term in Theorem~\ref{thm:main-value} is no larger for a sufficiently large constant $c$, proving the result. These uniform events, the stated computational guarantees, and $\mathcal R_n$ imply $\mathcal E_n$ before substituting the same-data iterates.
\end{proof}

\paragraph{Mixing version.}
Under Assumption~\ref{ass:main-mixing}, use the block length $b=\min\{T,\lceil c_b\log(eNT)\rceil\}$ for sufficiently large fixed $c_b>0$ and effective sample size $\mathcal N_b=NT/b$, as justified in Appendix~\ref{subsec:mix-rate}. Proposition~\ref{prop:mix-krr} supplies the same score and second-moment bounds at scale $\mathcal N_b$, while~\eqref{eq:mix-objective} supplies the policy-objective bound. Replace $N$ by $\mathcal N_b$ in the preceding deterministic covers, ridge coefficient, and computational tolerances, and take $M\ge c\log(NT)$ with the covers valid for the actual iteration sequence. In particular,
\[
 \varrho_n=\kappa_{\log,n}\mathcal N_b^{-(1-p_K)/2}
 =\kappa_{\log,n}b^{(1-p_K)/2}(NT)^{-(1-p_K)/2},
\]
where $\kappa_{\log,n}$ dominates the actual covering and confidence logarithms. Thus the ridge calibration for Example~\ref{subsec:main-krr} also includes the block factor $b^{(1-p_K)/2}$. The prediction and critic-update bounds are
\[
 \delta_{K,n}=\widetilde O((NT)^{-(1-p_K)/4}),\qquad
 \varepsilon_{Q,n}=\widetilde O((NT)^{-\vartheta_*(1-p_K)/4}).
\]
Combining these with the policy-error bound gives the value rate with $\mathcal N_b^{-\vartheta_*(1-p_K)/4}$. Since $b=O(\log(eNT))$, it gives
\[
 J^\star-J(\widehat\pi_M)
 \le a_{\mathcal H}+\widetilde O_p\!\left(
 (NT)^{-\vartheta_*(1-p_K)/4}+\sqrt{\Delta_n}\right),
\]
with additional dependence-comparison failure at most $(NT)^{-2}$ (zero when $b=T$). The target, raw-secant, policy-localization, and state-transfer conditions remain uniform over the recalibrated classes and pooled laws. This proves the stronger rate in Example~\ref{subsec:main-krr}.

\section{AdaFNN prediction and exact functional secants}\label{app:ada-regression}

We use the learned-basis architecture of \citet{YaoMuellerWang2021AdaFNN}. The finite-rank Bellman condition supplies uniform target approximation, while the prediction bounds cover all critic and policy parameters. Evaluation uses exact secants in outer-parameter and functional-basis coordinates, with a separate directional information condition.

\begin{assumption}[Adaptive-basis clipped ReLU fitted-$Q$ regression]\label{ass:ada-q-oracle}
The functional-action set is bounded in $L^2$:
\[
 {B}_A:=\sup_{a\in\mathcal A}\norm{a}_{L^2}<\infty.
\]
Fix a deterministic state feature map $\varphi_{S,n}:\mathcal S\to\mathbb R^{d_{S,n}}$ with compact range $\mathcal X_{S,n}$. When the state is already finite-dimensional, as in the Pendulum experiment, $\varphi_{S,n}$ is the identity. Assumption~\ref{ass:ada-cylindrical} uses this state representation in its finite-rank factorization.

Fix an integer $\dAda\ge1$, independent of the sample size, as the number of adaptive basis nodes.  For $j=1,\ldots,\dAda$, let
\[
 \beta_{\theta_{\beta,j}}(u)=\operatorname{NN}_{\theta_{\beta,j}}(u),\qquad u\in[0,1],
\]
where the micro-network class $\mathcal {B}_n$ has deterministic depth, width, sparsity, activation, and weight bounds, and satisfies $\sup_{\beta\in\mathcal {B}_n}\norm{\beta}_{L^2}\le {B}_{\beta,n}$.  Its exact basis score is
\[
 c_{\theta_{\beta,j}}(a):=\langle a,\beta_{\theta_{\beta,j}}\rangle_{L^2}
 =\int_0^1 a(u)\beta_{\theta_{\beta,j}}(u)\,du.
\]
The raw adaptive-basis ReLU critic is
\[
 f_{\theta,\boldsymbol\theta_\beta}(s,a)
 =\operatorname{NN}_\theta\!\left[\varphi_{S,n}(s),c_{\theta_{\beta,1}}(a),\ldots,c_{\theta_{\beta,\dAda}}(a)\right],
\]
where $\operatorname{NN}_\theta$ ranges over a deterministic bounded-weight ReLU class $\mathcal M_{\mathrm{out},n}$ on
\[
 \mathcal X_{\mathrm{Ada},n}:=\mathcal X_{S,n}\times[-{B}_A{B}_{\beta,n},{B}_A{B}_{\beta,n}]^{\dAda}.
\]
Write $\mathcal Q_{\mathrm{Ada},n}$ for this composite class and
$\overline{\mathcal Q}_{\mathrm{Ada},n}:=\{\clip_{[0,\Vmax]}f:f\in\mathcal Q_{\mathrm{Ada},n}\}$.
The class is fixed before the data are observed even though the fitted basis functions are learned from the data.
The micro-network class contains the zero function, so unused basis nodes in a lower-rank target can be zero-padded, by setting their output weights to zero.

The fitted-$Q$ routine returns a $\tau_{\mathrm{Ada},n}$-approximate global minimizer of the clipped empirical squared loss: for every deterministic label $y\in\mathcal Y_n$,
\[
\begin{split}
&\frac1\nobs\sum_{i,t}
 \{\clip_{[0,\Vmax]}\widehat f_y^{\mathrm{Ada}}(X_{i,t})-Y_{i,t}(y)\}^2\\
&\quad\le
\inf_{f\in\mathcal Q_{\mathrm{Ada},n}}
\frac1\nobs\sum_{i,t}
 \{\clip_{[0,\Vmax]}f(X_{i,t})-Y_{i,t}(y)\}^2
+\tau_{\mathrm{Ada},n},
\end{split}
\]
where $X_{i,t}=(S_{i,t},A_{i,t})$ and
$Y_{i,t}(y)=y(S_{i,t},A_{i,t},R_{i,t},S_{i,t+1})$.
The fitted-$Q$ update uses $\clip_{[0,\Vmax]}\widehat f_y^{\mathrm{Ada}}$.  Since rewards lie in $[0,1]$ and every preceding fitted iterate is clipped to $[0,\Vmax]$, the generated Bellman labels satisfy
\[
0\le Y(y)\le 1+\gamma\Vmax=\Vmax,
\qquad
0\le Q_y\le\Vmax.
\]
The same uniform bound is imposed on any enlarged deterministic label class.

For numerical implementation, let $\widetilde c_{\beta,J_{\mathrm{quad},n}}$ denote the score computed from the supplied observation record.  Assume it is measurable and that its declared representation supplies the bound
\[
 {A}_{\mathrm{quad},n}
 =\sup_{a,\beta}|\widetilde c_{\beta,J_{\mathrm{quad},n}}(a)
                    -\langle a,\beta\rangle_{L^2}|.
\]
This discrepancy is retained, and no vanishing rate is asserted for arbitrary raw point samples of an $L^2$ function. Exact inner products give zero discrepancy. Section~\ref{subsec:rough-numerical-observations} verifies numerical scores from cell averages using regularity of the learned basis, not of the observed action. Proposition~\ref{prop:ada-smooth-verification} transfers the resulting score discrepancy to critic outputs. Point evaluation on an arbitrary $L^2$ equivalence class is not treated as a bounded functional.

Prediction and covering bounds use the implemented scores; secant identities use the exact-score counterpart, with discrepancy bounded in Proposition~\ref{prop:ada-smooth-verification}. Throughout the approximation and prediction statements, $\mathcal Q_{\mathrm{Ada},n}$ denotes the implemented class, and it coincides with the exact-score class when ${A}_{\mathrm{quad},n}=0$. With numerical scores, the common state--score box is enlarged to
\[
 \mathcal X_{S,n}\times
 [-{B}_A{B}_{\beta,n}-{A}_{\mathrm{quad},n},
   {B}_A{B}_{\beta,n}+{A}_{\mathrm{quad},n}]^{\dAda},
\]
and the stated outer-network and entropy conditions are imposed there.

For $\eps,\eta>0$, define
\[
{H}_{F,n}^{\mathrm{Ada}}(\eps)
:=\log \operatorname{Cov}(\eps,\overline{\mathcal Q}_{\mathrm{Ada},n},\norm{\cdot}_\infty),
\qquad
{H}_{Y,M,n}^{\mathrm{Ada}}(\eta)
:=\log \operatorname{Cov}(\eta,\mathcal Y_n,\norm{\cdot}_\infty),
\]
\[
H_{\mathrm{Ada},M,n}(\eps,\eta,\delta)
:={H}_{F,n}^{\mathrm{Ada}}(\eps)+{H}_{Y,M,n}^{\mathrm{Ada}}(\eta)+\log(8M/\delta).
\]
The reachable-label entropy is obtained from the sup-norm Lipschitz Bellman maps, together with a policy cover when labels are indexed jointly by $(Q,\pi)$.

The primary sampling conditions are independent subject clusters (Assumption~\ref{ass:bounded-clusters}) or the full-sample two-group construction under Assumption~\ref{ass:mix-array}. Their effective sample sizes are $N$ and $\mathcal N_b$, respectively; only the full-sample two-group construction carries the comparison failure $\delta_b^{\rm mix}$. The optional one-group route uses $\Neff(b,q)$, omission $\mathsf{o}_{b,q}$ and failure $\delta_{\mathrm{coup}}(b,q)$ from the coupling bound \eqref{lem:block-coupling-Ti}. It carries the additional prediction radius $\Delta_{\mathrm{dep},Q}^{\mathrm{Ada}}=\Vmax\sqrt{\mathsf{o}_{b,q}}$, as proved in Appendix~\ref{subsec:mix-retained}; equality of retained and full design laws is not assumed.
The conversion to the policy-evaluation norm is supplied by Proposition~\ref{thm:smooth-spectral-transfer} after verifying the exact secant representation, smooth-feature bound, and directional moment comparison.
\end{assumption}
The learned score layer is the AdaFNN construction of \citet[Section~3]{YaoMuellerWang2021AdaFNN}, with both basis micro-networks and the outer network included in the fitted class. Declaring that entire class in advance allows the same-data uniformity argument of \citet{NguyenTangEtAl2022Besov} to cover the learned representation and its numerical scores together.

\begin{assumption}[Uniform finite-rank Bellman-target factorization]\label{ass:ada-cylindrical}
For every $y\in\mathcal Y_n$, there exist an integer $d_y^\circ\le \dAda$, a continuous linear map
\[
 \operatorname{Score}_y:L^2([0,1])\to\mathbb R^{d_y^\circ},\qquad \norm{\operatorname{Score}_y}_{\mathrm{op}}\le {B}_{\rm score},
\]
and a continuous outer map
\[
 h_y:\mathcal X_{S,n}\times[-{B}_A{B}_{\rm score},{B}_A{B}_{\rm score}]^{d_y^\circ}\to[0,\Vmax]
\]
such that
\[
 Q_y(s,a)=h_y\{\varphi_{S,n}(s),\operatorname{Score}_y a\},\qquad (s,a)\in \mathcal S\times\mathcal A.
\]
Assume ${B}_{\beta,n}\ge {B}_{\rm score}$.  After zero-padding $\operatorname{Score}_y a$ to dimension $\dAda$, assume that $h_y$ admits a $[0,\Vmax]$-valued extension to the common score domain $\mathcal X_{\mathrm{Ada},n}$, including its numerical enlargement when used.  All stated regularity and approximation bounds apply to this extension; padding alone does not extend its domain.  The outer maps share a deterministic modulus of continuity $\omega_{h,M,n}$ on their compact domains:
\[
|h_y(z,v)-h_y(z,v')|
\le \omega_{h,M,n}(\norm{v-v'}_2),
\qquad \omega_{h,M,n}(r)\downarrow0\ \text{as }r\downarrow0.
\]
The learned action directions are global; state dependence and interaction enter through $h_y$. The fixed $\dAda$ bounds the functional rank of every reachable target.
We take the common modulus $\omega_{h,M,n}$ to be continuous and nondecreasing, allowing passage to unattained approximation infima in Lemma~\ref{lem:ada-approximation}.
\end{assumption}
This uses the finite-rank composition behind \citet[Theorem~1]{YaoMuellerWang2021AdaFNN}, now with score maps explicitly continuous on $L^2$ and a common approximation modulus over the Bellman-label family. The directions may vary with the target label, so the assumption matches adaptive bases rather than a fixed action projection.

By the Hilbert-space Riesz representation theorem, each coordinate of $\operatorname{Score}_y$ has a representer $b_{y,j}\in L^2([0,1])$ satisfying
\[
\operatorname{Score}_y a=\bigl(\langle a,b_{y,1}\rangle,\ldots,\langle a,b_{y,d_y^\circ}\rangle\bigr),
\qquad
\max_j\norm{b_{y,j}}_{L^2}\le {B}_{\rm score}.
\]
Choose such representers and define the uniform micro-network and outer-network approximation errors
\[
{A}_{\beta,M,n}
:=\sup_{y}\max_{1\le j\le d_y^\circ}
\inf_{\beta\in\mathcal {B}_n}\norm{b_{y,j}-\beta}_{L^2},
\]
\[
{A}_{h,M,n}
:=\sup_y\inf_{\operatorname{NN}\in\mathcal M_{\mathrm{out},n}}
\sup_{(z,v)\in\mathcal X_{\mathrm{Ada},n}}|\operatorname{NN}(z,v)-h_y(z,v)|,
\]
where $\mathcal X_{\mathrm{Ada},n}$ is the corresponding zero-padded compact state-score domain.

\begin{lemma}[Uniform adaptive-basis approximation for reachable Bellman targets]\label{lem:ada-approximation}
Under Assumptions~\ref{ass:ada-q-oracle} and~\ref{ass:ada-cylindrical},
\[
\begin{split}
{A}_{\mathrm{Ada},M,n}^{\infty}
&:=\sup_{y\in\mathcal Y_n}
\inf_{f\in\mathcal Q_{\mathrm{Ada},n}}
\norm{\clip_{[0,\Vmax]}f-Q_y}_{\infty}\\
&\le {A}_{h,M,n}
+\omega_{h,M,n}\!\left[
\sqrt{\dAda}\{B_A{A}_{\beta,M,n}+{A}_{\mathrm{quad},n}\}
\right].
\end{split}
\]
In particular, the behavior-design approximation error
\[
{A}_{\mathrm{Ada},M,n}
:=\sup_y\inf_{f\in\mathcal Q_{\mathrm{Ada},n}}
\norm{\clip_{[0,\Vmax]}f-Q_y}_{L^2(\nuX)}
\]
and the corresponding approximation error in the policy-evaluation norm are both bounded by
${A}_{\mathrm{Ada},M,n}^{\infty}$.
If the outer maps are uniformly $\mathrm{Lip}_{h,M,n}$-Lipschitz in their score arguments, the final term is bounded by
\[
\mathrm{Lip}_{h,M,n}\sqrt{\dAda}\{B_A{A}_{\beta,M,n}+{A}_{\mathrm{quad},n}\}.
\]
\end{lemma}
The composition argument follows the basis-plus-outer approximation in \citet{YaoMuellerWang2021AdaFNN}; the displayed bound makes each contribution uniform over reachable Bellman targets and retains numerical score error. Smoothness assumptions can then be converted into explicit ReLU approximation rates using \citet[Theorem~5]{SchmidtHieber2020}.

\begin{proof}
By Riesz representation, $(\operatorname{Score}_y a)_j=\langle a,b_{y,j}\rangle$ with
$\|b_{y,j}\|_{L^2}\le {B}_{\rm score}$. Fix $\tau>0$ and choose basis approximants with
$\|\beta_{y,j}-b_{y,j}\|_{L^2}\le{A}_{\beta,M,n}+\tau$;
use zero in unused coordinates. If $z_y(a)$ is the zero-padded target score,
Cauchy--Schwarz and the numerical score bound give
\begin{align*}
 |\widetilde c_j(a)-(z_y(a))_j|
 &\le{A}_{\mathrm{quad},n}
      +{B}_A({A}_{\beta,M,n}+\tau),\\
 \|\widetilde c(a)-z_y(a)\|_2
 &\le\sqrt{\dAda}\{A_{\mathrm{quad},n}
      +{B}_A({A}_{\beta,M,n}+\tau)\}=:b_\tau.
\end{align*}
Both vectors lie in the declared enlarged score box. Choose an outer
approximant $\operatorname{NN}_y^\tau$ with error at most ${A}_{h,M,n}+\tau$ there.
For $f_y^\tau=\clip \operatorname{NN}_y^\tau(\varphi_{S,n},\widetilde c)$, nonexpansive clipping and
the modulus of the extended target imply
\begin{align*}
 |f_y^\tau(s,a)-Q_y(s,a)|
 &\le|\operatorname{NN}_y^\tau(\varphi_{S,n},\widetilde c)-h_y(\varphi_{S,n},\widetilde c)|
       +|h_y(\varphi_{S,n},\widetilde c)-h_y(\varphi_{S,n},z_y)|\\
 &\le{A}_{h,M,n}+\tau+\omega_{h,M,n}(b_\tau).
\end{align*}
Take the approximation infimum and then $\sup_y$. Letting $\tau\downarrow0$
uses continuity of the common modulus and gives the stated full-domain bound.
The Lipschitz version follows by $\omega_{h,M,n}(u)\le \mathrm{Lip}_{h,M,n}u$.
Both claimed prediction and policy-evaluation norms are bounded by the same
supremum norm because their underlying measures are probability laws.
\end{proof}

Fix a deterministic sequence $t_n\downarrow0$.  For every $y\in\mathcal Y_n$, choose a deterministic parameter $\Theta_y^\circ$ in the adaptive-basis class such that, with
$F_\Theta:=\clip_{[0,\Vmax]}f_\Theta$,
\begin{equation*}
\norm{F_{\Theta_y^\circ}-Q_y}_\infty
\le {A}_{\mathrm{Ada},M,n}^{\infty}+t_n.
\end{equation*}
These are deterministic near-oracle parameters for the fixed reachable-label class; they are not fitted from the observed sample.

\begin{lemma}[Composite entropy of the adaptive-basis critic]\label{lem:ada-entropy}
Assume every $\operatorname{NN}\in\mathcal M_{\mathrm{out},n}$ is $\mathrm{Lip}_{\mathrm{out},n}$-Lipschitz in its score coordinates.  Let
\[
H_{\mathrm{out},n}(u)
:=\log \operatorname{Cov}(u,\mathcal M_{\mathrm{out},n},\norm{\cdot}_{\infty,\mathcal X_{\mathrm{Ada},n}}),
\qquad
{H}_{\beta,n}(u)
:=\log \operatorname{Cov}(u,\mathcal {B}_n,\norm{\cdot}_{L^2}).
\]
For exact $L^2$ scores,
\[
{H}_{F,n}^{\mathrm{Ada}}(\eps)
\le H_{\mathrm{out},n}(\eps/2)
+\dAda\,{H}_{\beta,n}\!\left(
\frac{\eps}{2\mathrm{Lip}_{\mathrm{out},n}{B}_A\sqrt{\dAda}}
\right).
\]
For a deterministic quadrature implementation, the same statement holds with ${B}_A\norm{\beta-\widetilde\beta}_{L^2}$ replaced by a deterministic uniform score-Lipschitz bound for that rule. Depth and weight factors remain in the network covers, or are absorbed when they are logarithmic under the chosen architecture.
\end{lemma}
The sparse-network cover of \citet[Lemma~5]{SchmidtHieber2020} applies separately to the outer network and every basis micro-network. Composing their covers through the score Lipschitz bound accounts for the entire learned critic, and the Bellman-label cover subsequently adds the entropy of the full RKHS policy class.

\begin{proof}
Set $d=\dAda$. Choose an $\eps/2$-net of the outer class and a $u$-net
of the basis class, with representatives $\operatorname{NN}_k$ and $\widetilde\beta_j$.
Cauchy--Schwarz and the outer Lipschitz bound give
\begin{align*}
 \|c_\beta(a)-c_{\widetilde\beta}(a)\|_2&\le {B}_A\sqrt d\,u,\\
 |\operatorname{NN}(\varphi_{S,n},c_\beta)-\operatorname{NN}_k(\varphi_{S,n},c_{\widetilde\beta})|
 &\le\eps/2+\mathrm{Lip}_{\mathrm{out},n}{B}_A\sqrt d\,u.
\end{align*}
Take $u=\eps/(2\mathrm{Lip}_{\mathrm{out},n}{B}_A\sqrt d)$ when the denominator is positive.
Clipping does not increase the error. The product net consequently has at most
\begin{align*}
 \operatorname{Cov}(\eps/2,\mathcal M_{\mathrm{out},n},\|\cdot\|_\infty)
 \operatorname{Cov}(u,\mathcal {B}_n,\|\cdot\|_{L^2})^d
\end{align*}
elements. Taking logarithms proves the bound. If the denominator vanishes,
one basis choice suffices. For numerical scores use their declared
score-Lipschitz bound in the score-perturbation inequality, with the outer cover and
Lipschitz bound on the enlarged box.
\end{proof}

\begin{proposition}[Uniform adaptive-basis prediction under clustered and mixing observations]\label{thm:ada-retained-block-erm-oracle}\label{thm:ada-cluster-erm-oracle}
Under Assumption~\ref{ass:ada-q-oracle}, fix $\eps,\eta>0$. For independent subjects set $(\mathcal N_{\rm stat},\mathsf{o},\delta_{\mathrm{dep}})=(N,0,0)$; for the full-sample two-group partition set $(\mathcal N_{\rm stat},\mathsf{o},\delta_{\mathrm{dep}})=(\mathcal N_b,0,\delta_b^{\rm mix})$. With probability at least $1-\delta-\delta_{\mathrm{dep}}$, uniformly over the full deterministic label class,
\begin{equation*}
\begin{split}
 \|\clip_{[0,\Vmax]}\widehat f_y^{\mathrm{Ada}}-Q_y\|_{L^2(\nuX)}^2
 \le C\bigg[
 &{A}_{\mathrm{Ada},M,n}^2+\Vmax(\eps+\eta)+\tau_{\mathrm{Ada},n}\\
 &+\frac{\Vmax^2H_{\mathrm{Ada},M,n}(\eps,\eta,\delta)}{\mathcal N_{\rm stat}}
 +\Vmax^2\mathsf{o}\bigg].
\end{split}
\end{equation*}
Define $\delta_{\mathrm{Ada},n}$ by the square root, replacing the approximation radius by its full-domain upper bound ${A}_{\mathrm{Ada},M,n}^{\infty}$. The secondary retained-block version uses $(\mathcal N_{\rm stat},\mathsf{o},\delta_{\mathrm{dep}})=(\Neff(b,q),\mathsf{o}_{b,q},\delta_{\mathrm{coup}}(b,q))$ and requires $\Neff(b,q)\ge1$.
\end{proposition}
The regression step uses the double-uniform principle of \citet{NguyenTangEtAl2022Besov}, with the entropy of both learned network components and the policy-indexed labels. Whole subjects or the stated mixing blocks determine the effective sample size, so the same architecture-specific bound connects the independent-subject and temporally informative sampling regimes.
\begin{proof}
Use the deterministic implemented fit and label classes in Lemma~\ref{lem:ti-concentration} or Proposition~\ref{thm:mix-three-inputs}. Their product nets have the cost $H_{\mathrm{Ada},M,n}$, and labels and clipped fits are bounded by $\Vmax$. Those proofs establish the full-class relative excess-loss event before substituting an empirical minimizer. In their final basic inequality, the approximation term can be left as the population infimum rather than replaced immediately by a sup-norm bound:
\[
 \tfrac34\|\widehat f_y-Q_y\|_{L^2(\nuX)}^2
 \le\tfrac54\inf_{f\in\overline{\mathcal Q}_{\mathrm{Ada},n}}
                    \|f-Q_y\|_{L^2(\nuX)}^2
      +\tau_{\mathrm{Ada},n}
      +C\{\Vmax(\eps+\eta)+\Vmax^2H/\mathcal N_{\rm stat}\}.
\]
This proves the assertion with ${A}_{\mathrm{Ada},M,n}$; its full-domain upper bound then gives the declared radius. The numerical discrepancy is already part of the implemented-class approximation and is not discarded. For one-group blocks use \eqref{eq:mix-retained-prediction}, whose proof adds $C\Vmax^2\mathsf{o}$ without equating the design laws. Every event covers the full class before fitting, so no separate conditioning on the fitted labels or independence across FQI iterations is needed.
\end{proof}

\subsection{Exact functional-coordinate secants with all bases learned}

Let
\[
 F_{\theta,\boldsymbol\beta}(s,a)
 =\clip_{[0,\Vmax]}\operatorname{NN}_\theta\{\varphi_{S,n}(s),
                  \langle a,\beta_1\rangle,\ldots,\langle a,\beta_d\rangle\}.
\]
Each endpoint basis $\beta_j$ is still the output of its learned micro-network.  For the proof only, use the functional basis elements themselves as coordinates in
$\mathcal E=\mathbb R^{p_\theta}\oplus \{L^2([0,1])\}^d$.  A segment between two basis elements need not be another micro-network output; it is an analytical path between the two endpoint critics.  Thus no basis is frozen and no different fitting procedure is introduced.

\begin{proposition}[Functional secant features of AdaFNN]\label{prop:ada-functional-secant}
For two endpoint critics let
$d_v=(\theta-\theta',\beta_1-\beta'_1,\ldots,\beta_d-\beta'_d)$.
Take the straight segment $(\theta_t,\boldsymbol\beta_t)$ in these coordinates.  Suppose the clipped outer map along its deterministic parameter hull is Lipschitz in its score vector with constant $\mathrm{Lip}_{\mathrm{out},n}$, and has a pathwise parameter-derivative bound ${B}_{\partial,n}$ on the bounded input domain.  Then there are measurable integrated derivatives $u_v(s,a)\in\mathbb R^{p_\theta}$ and $w_v(s,a)\in\mathbb R^d$ such that
\begin{gather*}
 \psi_v(s,a)=(u_v(s,a),w_{v,1}(s,a)a,\ldots,w_{v,d}(s,a)a),\\
 F_{\theta,\boldsymbol\beta}-F_{\theta',\boldsymbol\beta'}
 =\langle d_v,\psi_v\rangle,\qquad
 \|u_v(s,a)\|\le {B}_{\partial,n},\quad
 \|w_v(s,a)\|\le \mathrm{Lip}_{\mathrm{out},n}.\nonumber
\end{gather*}
For
\[
 \WeightOp_v=\operatorname{Id}_{p_\theta}\oplus \BaseWeightOp_\alpha\oplus\cdots\oplus \BaseWeightOp_\alpha,
 \qquad\vartheta=\min\{1,\zeta/\alpha\},
\]
the integrated policy radius \eqref{eq:main-policy-radius} verifies \eqref{eq:secant-smooth-feature} with
\begin{equation}
 {B}_{\psi,n}^2={B}_{\partial,n}^2+\mathrm{Lip}_{\mathrm{out},n}^2{B}_{\Pi,\zeta}^2.
 \label{eq:ada-smooth-factor}
\end{equation}
Observed actions require only their $L^2([0,1])$ norm for feature integrability.  If the endpoint parameter and basis bounds are $\|\theta\|\le {B}_{\theta,n}$ and $\|\beta_j\|\le {B}_{\beta,n}$, then
$\|d_v\|\le {B}_{Q,n}:=2({B}_{\theta,n}^2+d{B}_{\beta,n}^2)^{1/2}$.
\end{proposition}
The pathwise chain rule of \citet[Lemma~2 and Theorem~8]{BoltePauwels2021Conservative} gives an exact difference representation even with ReLU and clipping. Keeping the learned basis functions as functional coordinates exposes the factors $w_{v,j}a$, through which the functional-action norm directly controls the policy-evaluation feature energy.
\begin{proof}
Put $\delta\theta=\theta-\theta'$ and $\delta\beta_j=\beta_j-\beta'_j$.
Along the analytical path $\theta_t=\theta'+t\delta\theta$,
$\beta_{t,j}=\beta'_j+t\delta\beta_j$, the scores satisfy
$\dot z_{t,j}(a)=\langle a,\delta\beta_j\rangle$.
The intermediate bases need not be micro-network outputs.

Write $F_t(s,a)=\overline{\operatorname{NN}}_{\theta_t}(s,z_t(a))$, including clipping.
On the bounded path this function is Lipschitz, hence absolutely continuous.
ReLU preactivations are piecewise polynomial in $t$; each nonzero piece
has finitely many roots, and an identically zero piece has zero path derivative.
The chain-rule selections in the declared derivative envelopes
\citep[Lemma~2 and Theorem~8]{BoltePauwels2021Conservative} therefore give
\begin{align*}
 \dot F_t&=\langle u_t,\delta\theta\rangle
       +\sum_{j=1}^d w_{t,j}\langle a,\delta\beta_j\rangle
       \quad\text{for almost every }t,\\
 u_v&=\int_0^1u_t\,dt,\qquad w_v=\int_0^1w_t\,dt,
 \qquad\|u_v\|\le {B}_{\partial,n},\quad\|w_v\|\le \mathrm{Lip}_{\mathrm{out},n}.
\end{align*}
The fundamental theorem of calculus now yields the exact identity
\begin{align*}
 F_1-F_0
 &=\langle u_v,\delta\theta\rangle
       +\sum_{j=1}^d\langle w_{v,j}a,\delta\beta_j\rangle_{L^2}
 =\langle d_v,\psi_v\rangle_{\mathcal E},\\
 \psi_v&=(u_v,w_{v,1}a,\ldots,w_{v,d}a).
\end{align*}
Thus all outer and basis differences are retained. For
$\WeightOp_v=\operatorname{Id}\oplus \BaseWeightOp_\alpha\oplus\cdots\oplus \BaseWeightOp_\alpha$ and $\alpha\vartheta\le \zeta$,
\begin{align*}
 \|\WeightOp_v^{-\vartheta/2}\psi_v(s,a)\|^2
 &=\|u_v\|^2+\|w_v\|^2\sum_{k\ge1}k^{2\alpha\vartheta}
                  |\langle a,\phi_k\rangle|^2\\
 &\le {B}_{\partial,n}^2+\mathrm{Lip}_{\mathrm{out},n}^2\|a\|_{\zeta}^2.
\end{align*}
Taking the policy supremum before integrating proves the feature bound
${B}_{\partial,n}^2+\mathrm{Lip}_{\mathrm{out},n}^2{B}_{\Pi,\zeta}^2$.
Behavior integrability and the coefficient radius use only ordinary $L^2([0,1])$ norms:
\begin{align*}
 \E\|\psi_v(S,A)\|^2
 &\le {B}_{\partial,n}^2+\mathrm{Lip}_{\mathrm{out},n}^2\E\|A\|_{L^2}^2<\infty,\\
 \|d_v\|^2
 &=\|\delta\theta\|^2+\sum_j\|\delta\beta_j\|_{L^2}^2
 \le4{B}_{\theta,n}^2+4d{B}_{\beta,n}^2.
\end{align*}
The directional lower-moment comparison is the existing actual-direction hypothesis on this full block feature second-moment operator.
\end{proof}

The associated second-moment operator is the full block operator
\[
 \Sigma_v=\mathbb E_{\nuX}\!
 \begin{bmatrix}u_v\\w_{v,1}A\\\vdots\\w_{v,d}A\end{bmatrix}
 \begin{bmatrix}u_v\\w_{v,1}A\\\vdots\\w_{v,d}A\end{bmatrix}^{\!*}.
\]
Its cross blocks and its dependence on both endpoint critics must be retained.  Neither a covariance of $A$ alone nor the last hidden layer verifies its lower bound.  The representation is an exact secant in outer-parameter and functional-basis coordinates, not a claim about a frozen basis or about the native micro-network parameter Jacobian.

For any fixed finite bounded-weight architecture and bounded state and score inputs, ${B}_{\partial,n}$ is finite on the parameter hull. Growing width, depth, or layer norms can enlarge it. Bounded individual weights do not prove a uniformly bounded product of layer operator norms, ${B}_{\partial,n}$, or the moment coefficient $c_n^{-1}$.

\begin{proposition}[Numerical score accuracy and action regularity]\label{prop:ada-smooth-verification}
Every endpoint critic is $\mathrm{Lip}_{\mathrm{Ada},n}=\mathrm{Lip}_{\mathrm{out},n}\sqrt d {B}_{\beta,n}$-Lipschitz in its action.  If each numerical integral score has uniform error at most ${A}_{\mathrm{quad},n}$, then
\begin{equation*}
 e_{\mathrm{Ada},n}=\mathrm{Lip}_{\mathrm{out},n}\sqrt d\,{A}_{\mathrm{quad},n}
\end{equation*}
bounds the numerical input error.
\end{proposition}
\citet[Algorithm~1]{YaoMuellerWang2021AdaFNN} evaluates each learned score by numerical integration. This bound quantifies how the score accuracy propagates through the outer network and supplies a common error budget for fitting, policy objectives and the exact-score secant comparison.
\begin{proof}
For a fixed endpoint basis, Cauchy--Schwarz and the outer Lipschitz bound give
\[
 \begin{aligned}
 \|c_{\boldsymbol\beta}(a)-c_{\boldsymbol\beta}(a')\|_2^2
 &=\sum_{j=1}^d|\langle a-a',\beta_j\rangle|^2\\
 &\le\|a-a'\|_{L^2}^2\sum_{j=1}^d\|\beta_j\|_{L^2}^2
 \le d{B}_{\beta,n}^2\|a-a'\|_{L^2}^2,\\
 |F(s,a)-F(s,a')|
 &\le \mathrm{Lip}_{\mathrm{out},n}
       \|c_{\boldsymbol\beta}(a)-c_{\boldsymbol\beta}(a')\|_2\\
 &\le \mathrm{Lip}_{\mathrm{out},n}\sqrt d\,{B}_{\beta,n}\|a-a'\|_{L^2}.
 \end{aligned}
\]
For numerical scores, both vectors lie in the stipulated enlarged box, so
\[
 \begin{aligned}
 \|\widetilde c(a)-c(a)\|_2^2
 &=\sum_{j=1}^d|\widetilde c_j(a)-c_j(a)|^2
 \le d{A}_{\mathrm{quad},n}^2,\\
 |\widehat F^{\mathrm{num}}(s,a)-F^{\mathrm{exact}}(s,a)|
 &\le \mathrm{Lip}_{\mathrm{out},n}\|\widetilde c(a)-c(a)\|_2\\
 &\le \mathrm{Lip}_{\mathrm{out},n}\sqrt d\,{A}_{\mathrm{quad},n}.
 \end{aligned}
\]
Taking the endpoint and full-domain suprema gives the stated uniform bounds.
\end{proof}

\begin{corollary}[AdaFNN evaluation and its joint-complexity rate]
\label{cor:ada-common-evaluation}\label{cor:ada-common-rate}
\textup{(i) Conditional evaluation.} Use the deterministic prediction and label classes of Proposition~\ref{thm:ada-cluster-erm-oracle}. For the deterministic exact-score approximants, take
\[
 \eta_{Q,n}={A}_{h,M,n}
 +\omega_{h,M,n}(\sqrt d {B}_A{A}_{\beta,M,n})+t_n.
\]
If the exact functional-coordinate family satisfies the actual-direction comparison \eqref{eq:secant-moment-lower} with the spectral weight operator in
Proposition~\ref{prop:ada-functional-secant}, then
\[
 \sup_y\|\widehat f_y-Q_y\|_{\nuplus,\Pi_{\mathcal H}(B_{\mathcal H}),2}
 \le{E}_n(\delta_{\mathrm{Ada},n},\eta_{Q,n},e_{\mathrm{Ada},n},{B}_{Q,n}),
\]
where ${B}_{\psi,n}$ is given by \eqref{eq:ada-smooth-factor} and
$\vartheta=\min\{1,\zeta/\alpha\}$.  For the accuracy-dependent directional condition, use ${E}_n^{\mathrm{sc}}$ instead.  Without either verification, the general bound is \eqref{eq:general-regression-transfer}.  These are information requirements on the declared family, not automatic properties of every AdaFNN architecture.

\textup{(ii) Joint-complexity specialization.} Use $\mathcal N_{\rm stat}=N$ for equal-length subject clusters or $\mathcal N_{\rm stat}=\mathcal N_b$ for full-sample two-group mixing. Along $M=M_n$, retain the approximation conditions and bounded $\Vmax$, and let ${D}_n$ be a deterministic bound on the full joint covering cost below, including the actual policy RKHS entropy. At polynomial net resolutions satisfying $\eps_n+\eta_n=O({D}_n/\mathcal N_{\rm stat})$, assume explicitly
\[
 {H}_{F,n}^{\mathrm{Ada}}(\eps_n)+{H}_{Y,M_n,n}^{\mathrm{Ada}}(\eta_n)
       +\log(8M_n/\delta_n)=\widetilde O({D}_n).
\]
The product cover verifies this absorption only when the additional depth, width, weight and action-Lipschitz factors in its logarithmic entropy bound are themselves absorbed; otherwise retain those factors in the displayed prediction bound. Under either primary regime,
\[
 \delta_{\mathrm{Ada},n}=\widetilde O\!\left(
 \eta_{Q,n}+e_{\mathrm{Ada},n}+\sqrt{D_n/\mathcal N_{\rm stat}}
                      +\sqrt{\tau_{\mathrm{Ada},n}}\right).
\]
For the secondary retained-block route use $\mathcal N_{\rm stat}=N_{\mathrm{eff}}(b,q)$ and retain
$\Vmax\sqrt{\mathsf{o}_{b,q}}$ and the coupling probability.  When the full prediction-plus-approximation expression is
$\widetilde O(\mathcal N_{\rm stat}^{-b_Q})$ and the constants in part (i) are bounded, policy evaluation error is
$\widetilde O(\mathcal N_{\rm stat}^{-\vartheta b_Q})$. Under the scale-dependent condition replace $\vartheta$ by $\vartheta_*$; if the feature or moment constants grow, retain them explicitly.
\end{corollary}
The prediction rate combines the target approximation and joint covering inputs used in neural fitted-$Q$ analysis \citep{NguyenTangEtAl2022Besov}. Evaluation at functional policy actions then contributes the separate exponent $\vartheta$ (or $\vartheta_*$) through secant stability, making the roles of critic capacity and functional-action regularity visible in the final rate.
\begin{proof}
Choose deterministic exact-score approximants $f_y^\circ$ from
Lemma~\ref{lem:ada-approximation}, including their near-minimizer tolerance in $t_n$.
For the exact-score fit $F_y$, the supplied bounds are
\begin{align*}
 \|f_y^\circ-Q_y\|_\infty&\le \eta_{Q,n},&
 \|\widehat f_y-Q_y\|_{L^2(\nuX)}&\le \delta_{\mathrm{Ada},n},&
 \|\widehat f_y-F_y\|_\infty&\le e_{\mathrm{Ada},n}.
\end{align*}
Thus the two triangle inequalities in Corollary~\ref{cor:common-regression-transfer}
apply with prediction argument
$r_n^\circ=\delta_{\mathrm{Ada},n}+\eta_{Q,n}+e_{\mathrm{Ada},n}$ and additive
error $\eta_{Q,n}+e_{\mathrm{Ada},n}$.
Proposition~\ref{prop:ada-functional-secant} supplies
${B}_{\psi,n}=\{B_{\partial,n}^2+\mathrm{Lip}_{\mathrm{out},n}^2{B}_{\Pi,\zeta}^2\}^{1/2}$ and radius ${B}_{Q,n}$.
Substitution gives
\begin{align*}
 \sup_y\|\widehat f_y-Q_y\|_{\nuplus,\Pi_{\mathcal H}(B_{\mathcal H}),2}
 &\le \eta_{Q,n}+e_{\mathrm{Ada},n}
       +{B}_{\psi,n}c_n^{-\vartheta/2}{B}_{Q,n}^{1-\vartheta}(r_n^\circ)^\vartheta.
\end{align*}
The finite-accuracy and scale-dependent versions use respectively
$\omega_n(r_n^\circ)$ and the infimum in
Proposition~\ref{prop:scale-directional-information}\textup{(ii)} for the last term.
All substitutions occur on the same full-label event.

\noindent\textit{Joint-complexity specialization.} For exact scores, action Lipschitz continuity gives the joint label comparison
\begin{align*}
 \|y_{Q,\pi}-y_{Q',\pi'}\|_\infty
 \le\gamma\{\|Q-Q'\|_\infty
       +\mathrm{Lip}_{\mathrm{Ada},n}\|\pi-\pi'\|_{\infty,L^2}\}.
\end{align*}
Choose critic and policy net radii $\eta/(2\gamma)$ and
$\eta/(2\gamma\mathrm{Lip}_{\mathrm{Ada},n})$; when $\mathrm{Lip}_{\mathrm{Ada},n}=0$ only
one policy representative is needed. Lemma~\ref{lem:ada-entropy} and the
assumed logarithmic absorption give
${H}_F(\eps_n)+{H}_Y(\eta_n)+\log(8M_n/\delta_n)=\widetilde O({D}_n)$.
For numerical scores, use a verified action-Lipschitz bound for the implemented critic in this product-cover argument, or impose the full implemented-class entropy bound in part (ii) directly; a uniform score discrepancy alone does not imply action Lipschitz continuity.
The numerical version of an exact approximant has full-domain error at most
$\eta_{Q,n}+e_{\mathrm{Ada},n}$. The prediction theorem therefore yields
\begin{align*}
 \delta_{\mathrm{Ada},n}^2
 &\le\widetilde O\{(\eta_{Q,n}+e_{\mathrm{Ada},n})^2+{D}_n/\mathcal N_{\rm stat}
                    +\tau_{\mathrm{Ada},n}+\Vmax^2\mathsf{o}\},\\
 \delta_{\mathrm{Ada},n}
 &\le\widetilde O\{\eta_{Q,n}+e_{\mathrm{Ada},n}+\sqrt{D_n/\mathcal N_{\rm stat}}
                    +\sqrt{\tau_{\mathrm{Ada},n}}+\Vmax\sqrt{\mathsf{o}}\}.
\end{align*}
The displayed bound for $\delta_{\mathrm{Ada},n}$ uses $\sqrt{\sum_jx_j}\le\sum_j\sqrt{x_j}$.
Both primary regimes have $\mathsf{o}=0$; the one-group alternative retains
its omission with $\mathcal N_{\rm stat}=\Neff(b,q)$. Add the relevant dependence failure to $\delta_n$.
Finally, if $r_n=\delta_{\mathrm{Ada},n}+\eta_{Q,n}+e_{\mathrm{Ada},n}
=\widetilde O(\mathcal N_{\rm stat}^{-b_Q})$, substitute this bound into part (i).
Bounded transfer factors give $r_n+Cr_n^\vartheta=\widetilde O(\mathcal N_{\rm stat}^{-\vartheta b_Q})$.
The accuracy-dependent information condition replaces $\vartheta$ by $\vartheta_*$.
\end{proof}

\section{Policy-norm implementation and coverage-compatible empirical total curvature}\label{app:closed-balanced-rates}
The theoretical update uses only $\|\pi\|_{\mathcal H}^2$. For the Sobolev construction in Proposition~\ref{prop:functional-action-rkhs}, this norm already contains the within-action terms that imply the spectral bound in Proposition~\ref{prop:curvature-spectral}; no second smoothness penalty is added to Algorithm~\ref{alg:ffqi-main}. The reported spline computation uses empirical total curvature in place of direct evaluation of this complete norm. The two objectives are not asserted to be identical. The role verified below is the one needed for coverage: empirical total curvature is a computable surrogate whose localized policies obey the same classwise spectral envelope. Curvature alone is a seminorm because it vanishes on affine action functions; the action-feasibility bound controls this null component through the $L^2$ term in Proposition~\ref{prop:curvature-spectral}, while the declared coefficient-map class controls state-side complexity.

For a fixed spline basis vector $\mathbf b^{\rm spl}$ in the action argument, set
\[
 G_0=\int_0^1\mathbf b^{\rm spl}(u)\mathbf b^{\rm spl}(u)^\top\,du,
 \qquad
 G_2=\int_0^1(\mathbf b^{\rm spl})''(u)
                    \{(\mathbf b^{\rm spl})''(u)\}^\top\,du,
 \qquad G_{\mathcal G}=G_0+G_2.
\]
In the two-state running example, writing $\pi(s)(u)=\mathbf b^{\rm spl}(u)^\top c_s$ gives the exact computable penalty
\[
 \|\pi\|_{\mathcal H}^2=\sum_{s=1}^2c_s^\top G_{\mathcal G}c_s.
\]
For the continuous-state construction $\mathcal H=H^1([0,1];\mathcal G)$ and a coefficient map $c\in H^1([0,1];\mathbb R^{p_n})$,
\[
 \|\pi\|_{\mathcal H}^2=\int_0^1
 \{c(s)^\top G_{\mathcal G}c(s)+c'(s)^\top G_{\mathcal G}c'(s)\}\,ds.
\]
Thus the same quadratic penalty is implementable from spline mass and second-derivative matrices, with an additional state-derivative term when the policy class couples continuous states. For bounded differentiable features $z(s)$ and $c(s)=a_{\max}\tanh(\Theta_{\rm pol}^\top z(s))$, the chain rule and $|\tanh'|\le1$ give
$\|c'(s)\|\le a_{\max}\|\Theta_{\rm pol}\|_{\rm op}\|z'(s)\|$.

For a deterministic feasible spline sieve $\mathcal C_n$, define
\begin{equation}
 \widehat\Omega_{\mathcal D}(\pi)
 =\frac1n\sum_{i,t}\|\partial_u^2\pi(S_{i,t+1})\|_{L^2}^2,
 \qquad
 \Omega(\pi)=\int\|\partial_u^2\pi(s)\|_{L^2}^2\nu_+(ds).
 \label{eq:empirical-population-curvature}
\end{equation}
The first quantity is the implemented surrogate and the second is its population counterpart. They measure the within-action component of the Sobolev geometry; state-side complexity remains controlled by the declared spline coefficient-map class.

\begin{proposition}[Curvature sublevels supply a classwise spectral envelope]\label{prop:empirical-curvature-envelope}
Suppose $\sup_{\pi\in\mathcal C_n,s}\|\pi(s)\|_{L^2}\le R_A$. Let
\[
 \mathcal G_n^{(2)}=\{s\mapsto\partial_u^2\pi(s):\pi\in\mathcal C_n\},
 \qquad
 U_n=\sup_{g\in\mathcal G_n^{(2)},s}\|g(s)\|_{L^2}^2<\infty,
\]
and write
\[
 \mathcal N_n(v)
 =\operatorname{Cov}\{v,\mathcal G_n^{(2)},\|\cdot\|_{\infty,L^2}\}.
\]
For $B\ge0$, put $\mathcal C_n(B)=\{\pi\in\mathcal C_n:\Omega(\pi)\le B\}$ and
\begin{equation}
 W_n(B)=\min\left[U_n,
   \inf_{v>0}\left\{\sqrt{\mathcal N_n(v)B}+2v\right\}^2\right].
 \label{eq:curvature-state-envelope}
\end{equation}
If the displayed covering numbers are finite, then, for the shifted Legendre basis in Proposition~\ref{prop:curvature-spectral} and $0<\zeta\le2$,
\begin{equation}
 \int\sup_{\pi\in\mathcal C_n(B)}\|\pi(s)\|_\zeta^2\nu_+(ds)
 \le C_S\{R_A^2+W_n(B)\}.
 \label{eq:curvature-sublevel-spectral-envelope}
\end{equation}
The bound $C_S(R_A^2+U_n)$ remains valid without the covering-number refinement.
\end{proposition}
\begin{proof}
Choose a $v$-cover of $\mathcal G_n^{(2)}$. From every cover ball meeting the population sublevel class, choose one member of that intersection. At most $\mathcal N_n(v)$ representatives $g_j$ are needed, and they form a $2v$-cover of the derivative maps in $\mathcal C_n(B)$. Because each representative belongs to the sublevel class,
\[
 \left\|\sup_{\pi\in\mathcal C_n(B)}
          \|\partial_u^2\pi(\cdot)\|_{L^2}\right\|_{L^2(\nu_+)}
 \le\left\|\max_j\|g_j(\cdot)\|_{L^2}\right\|_{L^2(\nu_+)}+2v
 \le\sqrt{\mathcal N_n(v)B}+2v.
\]
The pointwise envelope gives the alternative bound $U_n^{1/2}$. Apply Proposition~\ref{prop:curvature-spectral} to every selected action, take the policy supremum at each state and then integrate. The uniform $L^2$ action bound supplies $R_A^2$, and optimizing over $v$ proves~\eqref{eq:curvature-sublevel-spectral-envelope}.
\end{proof}

\begin{proposition}[Empirical total curvature localizes the population envelope]\label{prop:empirical-curvature-comparison}
Suppose that, on an event $\mathcal E_{\Omega,n}$, uniformly over $\mathcal C_n$,
\begin{equation}
 \Omega(\pi)\le c_{\Omega,n}\widehat\Omega_{\mathcal D}(\pi)+v_{\Omega,n}.
 \label{eq:one-sided-curvature-comparison}
\end{equation}
Then the empirical sublevel set
$\widehat{\mathcal C}_n(b)=\{\pi\in\mathcal C_n:\widehat\Omega_{\mathcal D}(\pi)\le b\}$
is contained in $\mathcal C_n(c_{\Omega,n}b+v_{\Omega,n})$ on that event. Consequently, its classwise spectral radius is bounded by the right side of~\eqref{eq:curvature-sublevel-spectral-envelope} with $B=c_{\Omega,n}b+v_{\Omega,n}$.

One sufficient additive comparison follows directly from the independent-subject empirical process. Let
$r_\pi(s)=\|\partial_u^2\pi(s)\|_{L^2}^2\in[0,U_n]$ and
$H_{\Omega,n}(\epsilon)=\log\operatorname{Cov}\{\epsilon,\{r_\pi:\pi\in\mathcal C_n\},\|\cdot\|_\infty\}$.
For every $\epsilon>0$, with probability at least $1-\delta$,
\begin{equation}
 \sup_{\pi\in\mathcal C_n}|\widehat\Omega_{\mathcal D}(\pi)-\Omega(\pi)|
 \le 2\epsilon+U_n
 \sqrt{\frac{H_{\Omega,n}(\epsilon)+\log(2/\delta)}{2N}}.
 \label{eq:curvature-uniform-deviation}
\end{equation}
Thus~\eqref{eq:one-sided-curvature-comparison} holds with $c_{\Omega,n}=1$ and the displayed right side as $v_{\Omega,n}$.
\end{proposition}
\begin{proof}
The set inclusion is immediate from~\eqref{eq:one-sided-curvature-comparison}, after which Proposition~\ref{prop:empirical-curvature-envelope} applies. For~\eqref{eq:curvature-uniform-deviation}, regard the average of $r_\pi(S_{i,t+1})$ over $t$ as one $[0,U_n]$-valued observation from subject $i$. These $N$ subject averages are independent and their mean expectation is $\Omega(\pi)$. Apply Hoeffding's inequality on a fixed $\epsilon$-net, take a union bound, and pay $2\epsilon$ to pass from the net to the full class. No within-subject independence is used.
\end{proof}

The next relative comparison can sharpen the additive bound when the derivative maps share a fixed finite state representation.
\begin{proposition}[Relative comparison for an empirical derivative criterion]\label{prop:rev-relative-penalty}
Suppose $\partial_u^q\pi(s)(u)=w(s)^\top v_\pi(u)$ in a declared numerical family, and put
$\Sigma_w=\E_{\nu_+}w(S)w(S)^\top$ and
$\widehat\Sigma_w=n^{-1}\sum_{i,t}w(S_{i,t+1})w(S_{i,t+1})^\top$.
If $c_-\Sigma_w\preceq\widehat\Sigma_w\preceq c_+\Sigma_w$, including null directions, then the empirical and population averages of $\|\partial_u^q\pi(S)\|_{L^2}^2$ lie within the same factors, uniformly over that family.
\end{proposition}
\begin{proof}
Tonelli writes the population quantity as
$\int v_\pi(u)^\top\Sigma_wv_\pi(u)\,du$ and the empirical quantity with $\widehat\Sigma_w$. Integrating the assumed quadratic-form comparison proves the claim. The Gram concentration calculation in Appendix~\ref{sec:ti-cluster} gives $c_-=1/2,c_+=3/2$ under its stated leverage and sample-size conditions.
\end{proof}
In particular, the lower comparison gives~\eqref{eq:one-sided-curvature-comparison} with
$c_{\Omega,n}=c_-^{-1}$ and $v_{\Omega,n}=0$.
A numerical derivative matrix must likewise be compared with $G_2$ on the coefficient space. The ordering $G_2\preceq c_{\rm num}\widehat G_2$ on the coefficient space gives the one-sided comparison required above, including the null directions; alternatively, a uniform numerical discrepancy can be absorbed into $v_{\Omega,n}$.

\begin{corollary}[The Pendulum empirical total curvature is coverage-compatible]\label{cor:pendulum-curvature-coverage}
For the experimental cubic B-spline policies in Appendix~\ref{app:implemented-penalty}, suppose the B-spline basis is nonnegative and forms a partition of unity and $c_C(s)\in[-2,2]^{p_u}$, as in the implemented bounded output map. Then, for every state and policy parameter,
\[
 \|\pi_C(s)\|_{L^2}^2\le4,
 \qquad
 \|\partial_u^2\pi_C(s)\|_{L^2}^2
 =c_C(s)^\top W_{p_u}c_C(s)
 \le4p_u\lambda_{\max}(W_{p_u})=\Omega_{\rm scale}(p_u).
\]
Consequently, for $0<\zeta\le2$,
\begin{equation}
 \int\sup_C\|\pi_C(s)\|_\zeta^2\nu_+(ds)
 \le C_S\{4+\Omega_{\rm scale}(p_u)\}<\infty.
 \label{eq:pendulum-spectral-envelope}
\end{equation}
Empirical total-curvature sublevels admit the sharper radius in Propositions~\ref{prop:empirical-curvature-envelope}--\ref{prop:empirical-curvature-comparison} whenever their stated comparison and covering conditions are used. Hence the implemented criterion verifies the policy-smoothness input to Assumption~\ref{ass:main-secant}\textup{(iii)}. Together with the directional-information comparison already imposed on the relevant critic secants in Assumption~\ref{ass:main-secant}\textup{(ii)}, it yields the spectral transfer bound of Proposition~\ref{thm:main-spectral}. Example~\ref{ex:two-state-control} proves this comparison from its positive logged coordinate weights; the KRR and AdaFNN results use the corresponding moment comparisons stated in their specializations.
\end{corollary}
\begin{proof}
Partition of unity and the coefficient box give $|\pi_C(s)(u)|\le2$. The quadratic-form bound follows from $\|c_C(s)\|_2^2\le4p_u$. Proposition~\ref{prop:curvature-spectral} then gives~\eqref{eq:pendulum-spectral-envelope} pointwise before integration.
\end{proof}
The reported experiment fixes $p_u=12$, so~\eqref{eq:pendulum-spectral-envelope} is a deterministic finite radius. If the spline dimension grows with sample size, the displayed $p_u$- and matrix-dependent radius must be retained in the transfer factor; this corollary does not assert a dimension-uniform bound.

\section{A quantitative RKHS-compatible approximation regime}\label{sec:structural-rates}
This section records how explicit approximation and entropy conditions enter the value theorem. The RKHS policy class remains fixed.
The full-kernel KRR rate for Example~\ref{ex:two-state-control} is proved separately in Appendix~\ref{app:verification-examples}.

For a finite-state model whose policy RKHS is a fixed Sobolev product space, let the label-uniform Bellman targets have at most $d$ functional scores as in Assumption~\ref{ass:ada-cylindrical}.
Suppose their score representers belong to a fixed one-dimensional $C^{s_b}$ ball and their outer functions, separately for each state, belong to a fixed $C^{s_h}$ ball on a compact $d$-dimensional score box.
Require the extended outer functions to be uniformly Lipschitz in scores.
Sparse ReLU approximations \citep{Yarotsky2017,SchmidtHieber2020} at accuracy radius $\varepsilon_{\mathrm{app}}$ can be chosen with outer and basis parameter counts
\[
 p_h(\varepsilon_{\mathrm{app}})=\widetilde O(\varepsilon_{\mathrm{app}}^{-d/s_h}),\qquad
 p_\beta(\varepsilon_{\mathrm{app}})=\widetilde O(\varepsilon_{\mathrm{app}}^{-1/s_b}),
\]
under the stated bounded-domain smoothness conditions.
The full-domain approximation bound in Lemma~\ref{lem:ada-approximation} is then $O(\varepsilon_{\mathrm{app}})$.
Take logarithmic depths and polynomial weight bounds so their covering factors at polynomial resolutions are logarithmic.
For the declared whole critic class, retain a verified action-Lipschitz bound
$\mathrm{Lip}(\varepsilon_{\mathrm{app}})\le C\varepsilon_{\mathrm{app}}^{-p_{\rm Lip}}$ with $p_{\rm Lip}\ge0$; smooth targets alone do not bound every fitted network's Lipschitz constant.

The entropy of the localized policy class from Appendix~\ref{sec:samedata-fqi}, at resolution $\varepsilon_{\mathrm{app}}^2/\mathrm{Lip}(\varepsilon_{\mathrm{app}})$, is
$\widetilde O(\varepsilon_{\mathrm{app}}^{-1-p_{\rm Lip}/2})$.
The composite critic and label covers consequently have order $\widetilde O(\varepsilon_{\mathrm{app}}^{-d_{\rm app}})$, where
\begin{equation*}
 D_{\rm Ada}=d_{\rm app}=\max\{d/s_h,\ 1/s_b,\ 1+p_{\rm Lip}/2\}.
\end{equation*}
This exponent includes full RKHS policy complexity. It is not the parameter count of the numerical spline optimizer.

\begin{corollary}[A conditional structural value rate]\label{cor:rkhs-structural-rate}
Under the preceding smoothness, architecture, entropy and numerical conditions, suppose the state constants are bounded with $\chi\le\chi_0<1$.
Require the membership and localization event and the secant conditions to hold with probability tending to one, with bounded feature/coefficient/information factors. Write $\vartheta_*$ for the verified transfer exponent, allowing accuracy-dependent information.
Assume the fixed-radius localization condition in~\eqref{eq:calibrated-policy-radius} holds for $B_{\mathcal H}$, and let $\Delta_{\varepsilon_{\mathrm{app}}}$ bound the policy-compatibility gap $\sup_{Q\in\mathcal Q_{\varepsilon_{\mathrm{app}}}}\Delta(Q)$.
Choose a deterministic logarithmic majorant $\kappa_{\log,N}$ and
\[
 \varepsilon_{\mathrm{app},N}=(\kappa_{\log,N}/N)^{1/(2+d_{\rm app})},\quad
 \lambda_{\pi,n}=O(\varepsilon_{\mathrm{app},N}),\quad
 \tau_{\rm Ada,n}=O(\varepsilon_{\mathrm{app},N}^2).
\]
Use exactly evaluated scores or a discrepancy of order $O(\varepsilon_{\mathrm{app},N})$ with the same implemented-class cover.
For conforming product splines take $d_u\gtrsim \varepsilon_{\mathrm{app},N}^{-(1+p_{\rm Lip})/2}$ and a policy optimization error of $O(\varepsilon_{\mathrm{app},N})$.
Then, with $M\ge c\log N$ for sufficiently large $c$,
\[
 J^\star-J(\widehat\pi_M)
 =a_{\mathcal H}+\widetilde O_p\!\left\{N^{-\min\{\vartheta_*,1/2\}/(2+d_{\rm app})}
                                  +\sqrt{\Delta_{\varepsilon_{\mathrm{app},N}}}\right\}.
\]
\end{corollary}
The approximation and network covers use \citet[Theorem~5 and Lemma~5]{SchmidtHieber2020}, while the policy cover uses the full Sobolev product envelope containing the localized class. The minimum in the exponent separates spectral transfer from the objective-to-greedy conversion: the objective-to-greedy conversion contributes $\varepsilon_{\mathrm{app},N}^{1/2}$ even when critic evaluation is faster.
\begin{proof}
Use fit and label resolutions of order $\varepsilon_{\mathrm{app},N}^2$. The entropy bound and choice of $\kappa_{\log,N}$ give
$\{H_F+{H}_Y+\log(1/\delta_N)\}/N=O(\varepsilon_{\mathrm{app},N}^2)$ for polynomially small $\delta_N$.
The clipped squared-loss proposition therefore gives prediction radius $O(\varepsilon_{\mathrm{app},N})$.
The same joint unpenalized objective cover yields $u_n=O(\varepsilon_{\mathrm{app},N})$.
The fixed ball bounds the deterministic penalty term by $\lambda_{\pi,n}B_{\mathcal H}^2=O(\varepsilon_{\mathrm{app},N})$.
Proposition~\ref{prop:spline-objective} gives representation gap
$O(\varepsilon_{\mathrm{app},N}^{-p_{\rm Lip}}d_u^{-2})=O(\varepsilon_{\mathrm{app},N})$.
Together with the stipulated solver and numerical objective tolerances, Lemma~\ref{lem:functional-objective-transfer} gives
$\varepsilon_{\pi,n}=O\{\varepsilon_{\mathrm{app},N}^{1/2}+\sqrt{\Delta_{\varepsilon_{\mathrm{app},N}}}\}$.
The verified transfer proposition gives $\varepsilon_{Q,n}=O(\varepsilon_{\mathrm{app},N}^{\vartheta_*})$.
Theorem~\ref{thm:main-value} and $M\asymp\log N$ finish the proof.
\end{proof}
The information and derivative envelopes remain explicit conditions on the whole growing critic family.
If their product grows, retain that product in~\eqref{eq:main-common-error} before propagating errors; the displayed exponent requires the bounded-factor premise.
For a coupled continuous-state RKHS class, the verified implementation in Appendix~\ref{subsec:continuous-rkhs-implementation} supplies representation and membership bounds, and the same value theorem additionally retains $\sup_Q\Delta(Q)$.

\subsection{Continuous states with the complete fixed RKHS ball}\label{subsec:continuous-full-rate}
Use the complete policy class of Proposition~\ref{prop:continuous-full-entropy} as an entropy envelope, without restricting it to its numerical projections. The value and coverage arguments use the same fixed comparison ball $\Pi_{\mathcal H}(B_{\mathcal H})$. At accuracy radius $\varepsilon_{\mathrm{app}}$, let $\mathcal Q_{\varepsilon_{\mathrm{app}}}$ be a clipped fitting class fixed before fitting. Suppose its full-domain Bellman targets, indexed by every $Q\in\mathcal Q_{\varepsilon_{\mathrm{app}}}$ and $\pi\in\Pi_{\mathcal H}(B_{\mathcal H})$, have approximants in the fitting class with uniform error $O(\varepsilon_{\mathrm{app}})$. Assume
\begin{equation}
 \log \operatorname{Cov}(\varepsilon_{\mathrm{app}}^2,\mathcal Q_{\varepsilon_{\mathrm{app}}},\|\cdot\|_\infty)
       \le C \varepsilon_{\mathrm{app}}^{-p_Q^{\rm ent}}\log^c(e/\varepsilon_{\mathrm{app}}),\qquad
 \mathrm{Lip}(\varepsilon_{\mathrm{app}})\le C\varepsilon_{\mathrm{app}}^{-p_{\rm Lip}},\quad p_{\rm Lip}\ge0,
 \label{eq:continuous-critic-inputs}
\end{equation}
where $\mathrm{Lip}(\varepsilon_{\mathrm{app}})$ is an action-Lipschitz bound on the entire class and $c$ is fixed. Define
\[
 {D}_{\rm ct}=\max\{p_Q^{\rm ent},5+5p_{\rm Lip}/2\},\qquad
 \Delta_{\varepsilon_{\mathrm{app}}}=\sup_{Q\in\mathcal Q_{\varepsilon_{\mathrm{app}}}}\Delta(Q).
\]
The policy part of this exponent follows by evaluating~\eqref{eq:continuous-full-entropy} at resolution $\varepsilon_{\mathrm{app}}^2/\{1+\mathrm{Lip}(\varepsilon_{\mathrm{app}})\}$. Thus the joint fit, objective and Bellman-label log-covering cost is $\widetilde O(\varepsilon_{\mathrm{app}}^{-{D}_{\rm ct}})$.

\begin{corollary}[Continuous-state value rate under a full-ball entropy envelope]\label{cor:continuous-full-rate}
Under~\eqref{eq:continuous-critic-inputs}, use independent subjects, clipped least-squares fitting with objective tolerance $O(\varepsilon_{\mathrm{app},N}^2)$, and exact scores. Suppose the uniform membership and localization, secant and state-transfer conditions hold with probability tending to one, with bounded transfer factors, bounded occupancy constants and $\chi\le\chi_0<1$. Let
\[
 \varepsilon_{\mathrm{app},N}=(\kappa_{\log,N}/N)^{1/(2+{D}_{\rm ct})},\qquad
 \lambda_{\pi,n}=O(\varepsilon_{\mathrm{app},N}),\qquad M\ge c_M\log N,
\]
where $\kappa_{\log,N}$ is a sufficiently large deterministic logarithmic majorant and $c_M$ is sufficiently large. Require the localized-class empirical policy-objective gap $\tau_{\pi,n}=O(\varepsilon_{\mathrm{app},N})$ and the fixed-radius norm-localization condition~\eqref{eq:calibrated-policy-radius}. Then
\begin{equation*}
 J^\star-J(\widehat\pi_M)
 =a_{\mathcal H}+\widetilde O_p\left\{
 N^{-\min\{\vartheta_*,1/2\}/(2+{D}_{\rm ct})}
       +\sqrt{\Delta_{\varepsilon_{\mathrm{app},N}}}\right\},\qquad
 \vartheta_*={2\vartheta\over2+\beta_{\mathrm{info}}\vartheta}.
\end{equation*}
Here $\beta_{\mathrm{info}}=0$ for uniform directional information; otherwise require the accuracy-dependent condition with fixed prefactor and accuracy interval. If $\Delta_{\varepsilon_{\mathrm{app},N}}=O(\varepsilon_{\mathrm{app},N})$, this is $a_{\mathcal H}$ plus the displayed polynomial rate; the approximation term vanishes when a globally optimal policy belongs to $\Pi_{\mathcal H}(B_{\mathcal H})$.
\end{corollary}
This extends the restricted-policy improvement issue in \citet{AntosEtAl2007} to a complete continuous-state, function-valued RKHS entropy envelope. The localized class inherits the entropy bound of that ball, while the residual propagation uses the same state-side roles as \citet{MunosSzepesvari2008,FarahmandMunosSzepesvari2010}. The additional $\Delta_{\varepsilon_{\mathrm{app}}}$ term bounds the policy-compatibility gap for the single coupled policy class, and $a_{\mathcal H}$ records approximation by the fixed comparison ball.
\begin{proof}
Use policy resolution $\varepsilon_{\mathrm{app},N}^2/\{1+\mathrm{Lip}(\varepsilon_{\mathrm{app},N})\}$ and critic resolution $\varepsilon_{\mathrm{app},N}^2$ in the product label cover. Equations~\eqref{eq:equal-objective}--\eqref{eq:equal-prediction}, with a polynomially small failure probability, give
\[
 u_n=\widetilde O(\varepsilon_{\mathrm{app},N}),\qquad
 \delta_{Q,n}=\widetilde O(\varepsilon_{\mathrm{app},N}),
 \qquad {\widetilde O(\varepsilon_{\mathrm{app},N}^{-{D}_{\rm ct}})\over N}=O(\varepsilon_{\mathrm{app},N}^2).
\]
The target approximation and numerical input bounds are included in the prediction radius. The deterministic penalty contribution is at most $\lambda_{\pi,n}B_{\mathcal H}^2=O(\varepsilon_{\mathrm{app},N})$. Lemma~\ref{lem:functional-objective-transfer} therefore gives
$\varepsilon_{\pi,n}=\widetilde O(\sqrt{\Delta_{\varepsilon_{\mathrm{app},N}}}+\varepsilon_{\mathrm{app},N}^{1/2})$.
Proposition~\ref{thm:smooth-spectral-transfer} and Corollary~\ref{cor:common-regression-transfer} give
$\varepsilon_{Q,n}=\widetilde O(\varepsilon_{\mathrm{app},N}^{\vartheta_*})$.
Substitute these bounds into Theorem~\ref{thm:main-value}; geometric decay controls the iteration remainder.

For a concrete computation, the energy projections and global shrinkage of Appendix~\ref{subsec:continuous-rkhs-implementation} preserve this same full ball. Choose
$h_s\lesssim \varepsilon_{\mathrm{app},N}^{2(1+p_{\rm Lip})}$ and $d_u\gtrsim \varepsilon_{\mathrm{app},N}^{-(1+p_{\rm Lip})/2}$.
Their empirical representation loss is
\[
 CL(\varepsilon_{\mathrm{app},N})B_{\mathcal H}(h_s^{1/2}+d_u^{-2})=O(\varepsilon_{\mathrm{app},N}).
\]
Solver and numerical objective tolerances $O(\varepsilon_{\mathrm{app},N})$ then supply the stated localized-class $\tau_{\pi,n}$ through~\eqref{eq:full-policy-objective-decomposition}. These projections change the implementation, not the localized policy class in any statistical supremum.
\end{proof}

The selector condition can be checked directly. If each $Q\in\mathcal Q_{\varepsilon_{\mathrm{app}}}$ has a measurable pointwise maximizing selector in the same $\Pi_{\mathcal H}(B_{\mathcal H})$, then $\Delta_{\varepsilon_{\mathrm{app}}}=0$. For example, fix $v_\circ\in\mathcal G$ with $\|v_\circ\|_{L^2}=1$ and take
\[
 Q(s,a)=\clip_{[0,\Vmax]}\{q(s)-\eta(s)(\langle v_\circ,a\rangle-g(s))^2\},
\]
where $0\le q\le\Vmax$, $\eta\ge0$, $\|g\|_{H^1}\|v_\circ\|_{\mathcal G}\le B_{\mathcal H}$ and $\|g\|_\infty\le{B}_A$. The policy $\pi_Q(s)=g(s)v_\circ$ lies in $\Pi_{\mathcal H}(B_{\mathcal H})$ and attains $q(s)$ at every state. This verifies the selector requirement for this critic family; the policy domain itself remains infinite dimensional.

Numerical scores give the same conclusion when implemented critics differ from their exact-score counterparts by $O(\varepsilon_{\mathrm{app},N})$ uniformly on the fitting and policy-evaluation domains and the implemented critics and labels, together with those counterparts, admit the same deterministic covering bounds. The proof then uses these implemented-class covers in concentration and charges the uniform discrepancy in Corollary~\ref{cor:common-regression-transfer}.

For learned-basis critics with $d$ scores, outer $C^{s_h}$ targets ($s_h\ge1$) on a fixed compact continuous state--score domain, uniformly Lipschitz in the scores, and $C^{s_b}$ score representers give the approximation budgets of \citet[Theorem~5 and Lemma~5]{SchmidtHieber2020} with
\[
 p_Q^{\rm ent}=\max\{(1+d)/s_h,1/s_b\},\qquad
 {D}_{\rm ct}=\max\{(1+d)/s_h,1/s_b,5+5p_{\rm Lip}/2\}.
\]
The budgets use the finite-rank factorization of Assumption~\ref{ass:ada-cylindrical}, uniformly over the full policy-indexed labels. Any restriction enforcing the declared Lipschitz or selector bounds must also contain the target approximants required in~\eqref{eq:continuous-critic-inputs}; this containment is a premise of the specialization. The same argument permits other critic covers by retaining their $p_Q^{\rm ent}$. These are sufficient rates from an explicit full-ball cover.

\subsection{Retaining growing transfer factors}\label{subsec:growing-transfer}
\begin{corollary}[Value rates with growing critic transfer factors]\label{cor:growing-transfer}
Use either structural setting above, with $\varepsilon_{\mathrm{app},N}=\widetilde O(N^{-1/(2+D)})$, prediction, approximation and numerical errors $\widetilde O(\varepsilon_{\mathrm{app},N})$, and policy error
$\widetilde O(\varepsilon_{\mathrm{app},N}^{1/2}+\sqrt{\Delta_{\varepsilon_{\mathrm{app},N}}})$.
Retain bounded state constants, $\chi\le\chi_0<1$ and $M\ge c_M\log N$ for sufficiently large $c_M$.
Replace bounded transfer factors by the accuracy-dependent information bound
$c_n(\varepsilon)\ge\underline c_n\varepsilon^{\beta_{\mathrm{info}}}$ on a fixed accuracy interval, where $\beta_{\mathrm{info}}\ge0$ is the information-decay exponent, and assume
\[
 {B}_{\psi,n}\underline c_n^{-\vartheta/2}
       {B}_{Q,n}^{1-\vartheta}
       =\widetilde O(\varepsilon_{\mathrm{app},N}^{-p_{\rm grow}}),\qquad 0\le p_{\rm grow}<\vartheta.
\]
At $\vartheta=1$ the coefficient-radius factor is omitted. Then
\begin{equation*}
 J^\star-J(\widehat\pi_M)
 =a_{\mathcal H}+\widetilde O_p\left\{
 N^{-\frac1{2+D}
       \min\{\frac{2(\vartheta-p_{\rm grow})}{2+\beta_{\mathrm{info}}\vartheta},\,1/2\}}
       +\sqrt{\Delta_{\varepsilon_{\mathrm{app},N}}}\right\}.
\end{equation*}
For a finite-state setting with a product localized class, or with the direct statewise defect bound used in Example~\ref{ex:two-state-control}, take $\Delta_{\varepsilon_{\mathrm{app},N}}=0$.
\end{corollary}
Hilbert-scale interpolation relates weak and strong errors through smoothness \citep[equation~(1.4)]{EggerHofmann2018}. This corollary additionally retains the sample-dependent multiplier of the critic secants before balancing accuracy. The fixed comparison-ball radius keeps the policy-action evaluation radius bounded; only verified growth of critic features, coefficients or inverse information is charged here.
\begin{proof}
The finite-accuracy version of Proposition~\ref{thm:smooth-spectral-transfer} bounds the critic evaluation error by an additive $\widetilde O(\varepsilon_{\mathrm{app},N})$ plus
\[
 \inf_{0<\varepsilon\le \varepsilon_0}
 \{\varepsilon+\widetilde O(\varepsilon_{\mathrm{app},N}^{\vartheta-p_{\rm grow}})
                    \varepsilon^{-\beta_{\mathrm{info}}\vartheta/2}\}.
\]
For $\beta_{\mathrm{info}}>0$, balancing gives exponent
$b_{\rm tr}=2(\vartheta-p_{\rm grow})/(2+\beta_{\mathrm{info}}\vartheta)>0$; for $\beta_{\mathrm{info}}=0$, let $\varepsilon$ decrease to zero and take $b_{\rm tr}=\vartheta-p_{\rm grow}$.
Since $b_{\rm tr}\le1$, the additive error is no slower. The main value theorem with the stated policy error yields
$a_{\mathcal H}+\widetilde O_p(\varepsilon_{\mathrm{app},N}^{\min\{b_{\rm tr},1/2\}}+\sqrt{\Delta_{\varepsilon_{\mathrm{app},N}}})$,
which is the displayed bound. Setting $p_{\rm grow}=0$ recovers the bounded-factor rate.
\end{proof}

\subsection{Unequal-length interpretation}\label{subsec:cluster-length-rates}
For the weighted extension of Appendix~\ref{sec:ti-cluster}, replace the structural proof's concentration scale by
$\mathcal N_{\infty}=n/T_{\max}$, retaining its actual objective scale $\mathcal N_2=n^2/\sum_iT_i^2$ where useful.
The conservative KRR score/operator argument instead uses $N_\star=N/C_{T,n}^2$.
If $\overline T\asymp N^{p_{\rm avg}}$ and $C_{T,n}\asymp N^{p_C}$, these scales obey
\[
 n\asymp N^{1+p_{\rm avg}},\quad
 \mathcal N_{\infty}\asymp N^{1-p_C},\quad
 N_\star\asymp N^{1-2p_C},\qquad
 \mathcal N_{\infty}\le\mathcal N_2\le N.
\]
Thus a verified neural structural power $p$ has order $\widetilde O_p(N^{-(1-p_C)p})$ for $p_C<1$, while the conservative KRR bound must use its own scale and constants.
For $N-N^{p_{\rm long}}$ subjects with $N^{p_{\rm short}}$ transitions each and $N^{p_{\rm long}}$ subjects with $N^{p_{\rm longlen}}$ transitions each, $p_{\rm longlen}\ge p_{\rm short}$, $0<p_{\rm long}<1$, put
$p_T=\max\{1+p_{\rm short},p_{\rm long}+p_{\rm longlen}\}$.
Direct summation gives
\[
 n\asymp N^{p_T},\qquad
 \mathcal N_{\infty}\asymp N^{p_T-p_{\rm longlen}},\qquad
 \mathcal N_2\asymp N^{2p_T-\max\{1+2p_{\rm short},p_{\rm long}+2p_{\rm longlen}\}}.
\]
These are length-array calculations. They do not verify policy entropy or information for the changing pooled design.
For conditional random-length statements the constants must hold uniformly on the stated length event, and its failure probability is added.

\section{Quantitative beta-mixing rates for the full-sample algorithm}\label{sec:mix-complete}
We retain the approximation, information and Bellman arguments, and prove
one full-sample event for prediction and policy improvement, together with an optional finite-feature Gram comparison.
Two alternating block groups include all terminal fragments, so their sum
is the unchanged transition-weighted objective and has no omission error.
Section~\ref{subsec:mix-retained} treats the secondary one-group construction,
where omission and coupling costs remain explicit.

\subsection{Dependence and a partition that uses every transition}\label{subsec:mix-definition}
Use the common-length data with $n=NT$ and independent subjects. The same proof extends to the length array of Appendix~\ref{sec:ti-cluster} by retaining each block's actual length. Write
\begin{align*}
 Z_{i,t}=(S_{i,t},A_{i,t},R_{i,t},S_{i,t+1}),\qquad
 \widehat{\mathbb P} f=n^{-1}\sum_{i,t}f(Z_{i,t}),\qquad \mathbb P f=\mathbb E\widehat{\mathbb P} f.
\end{align*}
The dependence condition is on these complete transition records, not only on the state sequence. This matters, for example, when a persistent subject-level action component remains after state dependence has weakened.

\begin{assumption}[Absolute regularity of complete records]\label{ass:mix-array}\label{ass:subject-mixing}
For $Z_{i,t}=(S_{i,t},A_{i,t},R_{i,t},S_{i,t+1})$, require
\[
 \sup_{i,N}\sup_{1\le t\le T-k}
 \beta_{\rm mix}\{\sigma(Z_{i,1},\ldots,Z_{i,t}),
                 \sigma(Z_{i,t+k},\ldots,Z_{i,T})\}
 \le\beta_{\rm mix}(k)\longrightarrow0.
\]
Here $\beta_{\rm mix}(\mathcal U,\mathcal V)$ is one half the supremum of
$\sum_{j,\ell}|\Pr(U_j\cap V_\ell)-\Pr(U_j)\Pr(V_\ell)|$
over finite measurable partitions from the two sigma fields.
Its nonincreasing envelope and constants are uniform in subjects and sample size.
\end{assumption}
We use this uniform array condition. Its bound concerns the complete records and is uniform in subjects and sample size; terminal fragments and nonstationary marginals remain allowed.
This uses the absolute-regularity convention of \citet[equations~(1.5), (2.2)]{Bradley2005}. Marginal stationarity, a common subject marginal distribution, and independence between the two groups below are not assumed. The common conditional model still requires $Q_y(X_{i,t})=\mathbb E\{Y_y(Z_{i,t})\mid X_{i,t}\}$ at every observed pair, for each deterministic covered label.

For an integer $1\le b\le T$, partition each subject into consecutive blocks
\begin{align*}
 m_i&=\lceil T/b\rceil,\qquad
 I_{i,j}=\{(j-1)b+1,\ldots,\min(jb,T)\},\qquad
 \ell_{i,j}=|I_{i,j}|,\\
 \mathcal I_c&=\{(i,j):j\equiv c\pmod 2\},\quad c=0,1,\qquad
 \widehat{\mathbb P}_c f=\frac1n\sum_{(i,j)\in\mathcal I_c}\sum_{t\in I_{i,j}}f(Z_{i,t}).
\end{align*}
Thus $\widehat{\mathbb P}=\widehat{\mathbb P}_0+\widehat{\mathbb P}_1$, $\mathbb P=\mathbb P_0+\mathbb P_1$ for $\mathbb P_c=\mathbb E\widehat{\mathbb P}_c$, and $\sum_{i,j}\ell_{i,j}=n$. The normalization in each group is $n$, not that group's sample size. A last block may be short; it is never deleted. Define
\begin{align*}
 {L}_b^{\rm link}&=\sum_i(m_i-2)_+,\qquad
 \delta_b^{\rm mix}=\min\{1,{L}_b^{\rm link}\beta_{\rm mix}(b+1)\},\qquad
 \mathcal N_b=\frac{n}{b}.
\end{align*}
Consecutive blocks within the same group are separated by one complete length-$b$ block. There are exactly ${L}_b^{\rm link}$ within-subject links across the two groups, and
\begin{align*}
 {L}_b^{\rm link}\le n/b,\qquad
 \delta_b^{\rm mix}\le\min\{1,(n/b)\beta_{\rm mix}(b)\},\qquad
 \mathcal N_b\ge N.
\end{align*}
At $b=T$ there is one block per subject and $\delta_b^{\rm mix}=0$, whether or not the mixing envelope is small. The effective sample size $\mathcal N_b$ is justified by the concentration bounds, not a claim that all transitions are independent.

\begin{lemma}[A shared independent-block comparison]\label{lem:mix-comparison}
For each $c$, compare the joint law of the blocks in $\mathcal I_c$ with the product of their actual blockwise marginal laws. Their total-variation distance is at most the sum of the within-subject link coefficients in that group. Consequently, if events $\mathsf{Good}_c^{\rm mix}$, each measurable with respect to its own block group, have failure probabilities at most $\delta_c$ under their respective product laws, then
\begin{align}
 \Pr(\mathsf{Good}_0^{\rm mix}\cap \mathsf{Good}_1^{\rm mix})\ge1-\delta_0-\delta_1-\delta_b^{\rm mix}.
 \label{eq:mix-joint-event}
\end{align}
An event $\mathsf{Good}_c^{\rm mix}$ may simultaneously cover every deterministic fit, Bellman label, policy objective, second-moment operator and optional derivative matrix considered below. The comparison cost is not multiplied by a covering number or by the number of FQI iterations.
\end{lemma}
The independent-block comparison of \citet[Corollary~2.7 and Lemma~4.1]{Yu1994} transfers one uniform event at a time. Keeping the actual nonstationary block marginals and subject boundaries gives the displayed array bound; the two groups are combined by a probability union bound rather than an independence assertion.
\begin{proof}
For two standard-Borel random elements, the partition definition of $\beta_{\rm mix}$ equals the total-variation distance between their joint law and the product of their marginals; this follows first for finite partitions and then by measurable approximation. For ordered same-group block records $\mathbf Z_1,\ldots,\mathbf Z_m$ of one subject, telescoping the joint law against the product law gives
\begin{align}
 \left\|\law(\mathbf Z_1,\ldots,\mathbf Z_m)-\bigotimes_{j=1}^m\law(\mathbf Z_j)\right\|_{\mathrm{TV}}
 \le\sum_{j=2}^m
 \beta_{\rm mix}\{\sigma(\mathbf Z_1,\ldots,\mathbf Z_{j-1}),\sigma(\mathbf Z_j)\}
 \le(m-1)\beta_{\rm mix}(b+1).
 \label{eq:mix-tv-telescope}
\end{align}
The product-law telescoping bound in~\eqref{eq:mix-tv-telescope} follows inductively by separating the last block and using that tensoring by a fixed probability law does not increase total variation. Its $\beta_{\rm mix}(b+1)$ bound uses monotonicity under restriction of sigma fields and the actual time separation. Independence across subjects permits another product-law telescoping, adding their bounds. The sum of the two group link counts is $\sum_i(m_i-2)_+$, including cases $m_i=1,2$. Transfer each group's entire bad event once from its product law and apply a union bound. If the uncapped comparison sum exceeds one, \eqref{eq:mix-joint-event} with $\delta_b^{\rm mix}=1$ is vacuous but valid. The two product laws are used separately; their groups are never asserted to be independent of one another. This is the independent-block comparison principle used by \citet[Lemma~4.1]{Yu1994}, written here for nonstationary subject arrays and unequal terminal blocks.
\end{proof}

\subsection{All three statistical inputs on one event}\label{subsec:mix-statistics}
The next result replaces the cluster concentration steps, not just the sample-size notation. The Gram part is included for optional finite-feature diagnostics and is not needed by the deterministic policy-norm penalty.

\begin{proposition}[Full-sample objective, prediction and Gram control]\label{thm:mix-three-inputs}
Under Assumption~\ref{ass:mix-array}, fix $b$ and deterministic classes before fitting. Let $0<\delta<1$.
For a class $\mathcal O$ with $0\le f\le B$ and sup-norm log covering number ${H}_{\mathcal O}(\epsilon)$,
\begin{align}
 \sup_{f\in\mathcal O}|(\widehat{\mathbb P}-\mathbb P)f|
 \le2\epsilon+CB\sqrt{\frac{H_{\mathcal O}(\epsilon)+\log(16/\delta)}{\mathcal N_b}}.
 \label{eq:mix-objective}
\end{align}
For clipped least squares, let fits, labels and their targets take values in $[0,V_{\max}]$. Suppose the full-domain approximation error is at most $\eta_Q$ and the fitted objective gap is at most $\tau_Q$. Let $H(\epsilon,\eta)$ be the sum of the fit and label log covers. Then
\begin{align*}
 \sup_y\|\widehat f_y-Q_y\|_{L^2(\nu_X)}^2
 \le C\left\{\eta_Q^2+\tau_Q+V_{\max}(\epsilon+\eta)
       +\frac{V_{\max}^2[H(\epsilon,\eta)+\log(16/\delta)]}{\mathcal N_b}\right\}.
\end{align*}
For the linear state-feature derivative representation, put
\begin{align*}
 \Sigma_{\rm st}=\mathbb E_{\nuplus}w_{\rm st}w_{\rm st}^\top,\qquad
 \widehat\Sigma_{\rm st}=n^{-1}\sum_{i,t}w_{\rm st}(S_{i,t+1})w_{\rm st}(S_{i,t+1})^\top,
\end{align*}
and suppose $\operatorname{rank}(\Sigma_{\rm st})=r_{\rm st}\le d_S$ and
$\operatorname*{ess\,sup}_{\nuplus}w_{\rm st}^\top \Sigma_{\rm st}^\dagger w_{\rm st}\le C_{\rm lev}d_S$.
If $\mathcal N_b\ge C d_S\log(16(d_S\vee1)/\delta)$, then
\begin{align}
 \tfrac12\Sigma_{\rm st}\preceq\widehat\Sigma_{\rm st}\preceq\tfrac32\Sigma_{\rm st}.
 \label{eq:mix-gram}
\end{align}
The three conclusions hold jointly with probability at least $1-\delta-\delta_b^{\rm mix}$, after choosing the universal constants consistently. The zero-rank case of \eqref{eq:mix-gram} holds exactly. If the optional empirical derivative criterion of Appendix~\ref{app:closed-balanced-rates} is used, the state Gram event gives its relative comparison with $c_-=1/2$ and $c_+=3/2$.
\end{proposition}
The block principle of \citet{Yu1994} and unequal-range concentration of \citet[Theorem~2]{Hoeffding1963} are applied to all transitions through the two groups. This simultaneously changes the objective, regression and optional finite-feature curvature inputs of FQI, so the temporal gain is established for the fitted full-sample procedure.
\begin{proof}
All concentration arguments in this proof first take place under a group's product-block law; its block marginals are unchanged.
For a fixed bounded $f$, a block contribution to $\widehat{\mathbb P}_c f$ has range length at most $B\ell_{i,j}/n$. Unequal-range Hoeffding concentration \citep[Theorem~2]{Hoeffding1963} and
\begin{align*}
 \sum_{i,j}\frac{\ell_{i,j}^2}{n^2}\le\frac bn=\mathcal N_b^{-1}
\end{align*}
give the two group deviations. Their sum is at most the right side of \eqref{eq:mix-objective} on a deterministic net, by Cauchy--Schwarz over the two groups. The centered off-net extension costs $2\epsilon$ for the full mean. This applies to unpenalized policy objectives with $f(Z)=Q\{S',\pi(S')\}$ and the full joint cover.

For prediction, fix $f,y$ and define
\begin{align*}
 d_{f,y}(Z)&=(f(X)-Y_y(Z))^2-(Q_y(X)-Y_y(Z))^2,\qquad
 \epsilon_{f,y}^2=\mathbb P d_{f,y}=\|f-Q_y\|_{L^2(\nu_X)}^2,\\
 \Delta_{i,j}^{f,y}&=n^{-1}\sum_{t\in I_{i,j}}d_{f,y}(Z_{i,t}),\qquad
 \mu_{i,j}=\mathbb E \Delta_{i,j}^{f,y}
 =n^{-1}\sum_{t\in I_{i,j}}\mathbb E(f-Q_y)^2(X_{i,t})\ge0.
\end{align*}
Compatibility and boundedness give $|d_{f,y}|\le V_{\max}^2$ and
$d_{f,y}^2\le4V_{\max}^2(f-Q_y)^2$. Jensen's inequality within each block therefore gives
\begin{align*}
 |\Delta_{i,j}^{f,y}-\mu_{i,j}|&\le2V_{\max}^2b/n,\qquad
 \mathbb E (\Delta_{i,j}^{f,y})^2
 \le4V_{\max}^2(\ell_{i,j}/n)\mu_{i,j}.
\end{align*}
Within group $c$ the centered contributions are independent, and their variance sum is at most
$4V_{\max}^2\mathcal N_b^{-1}\epsilon_c^2$, where $\epsilon_c^2=\sum_{\mathcal I_c}\mu_{i,j}\ge0$ and $\epsilon_0^2+\epsilon_1^2=\epsilon_{f,y}^2$.
Scalar Bernstein followed by Young's inequality yields
\begin{align*}
 |(\widehat{\mathbb P}_c-\mathbb P_c)d_{f,y}|
 \le\tfrac14\epsilon_c^2+CV_{\max}^2x/\mathcal N_b
\end{align*}
with failure at most $2e^{-x}$. Sum the two inequalities and take product fit/label nets with
$x=H(\epsilon,\eta)+\log(16/\delta)$. The common full-domain kernel gives
$\|Q_y-Q_{y'}\|_\infty\le\|y-y'\|_\infty$; the loss and risk extensions cost $CV_{\max}(\epsilon+\eta)$. Thus, on a full-class event,
\begin{align}
 |(\widehat{\mathbb P}-\mathbb P)d_{f,y}|
 \le\tfrac14\epsilon_{f,y}^2+C\{V_{\max}(\epsilon+\eta)+V_{\max}^2x/\mathcal N_b\}.
 \label{eq:mix-relative-full}
\end{align}
Apply the full-sample approximate minimization inequality to the fit and a deterministic approximation oracle:
\begin{align}
 \tfrac34\epsilon_{\widehat f_y,y}^2
 \le\tfrac54\epsilon_{f_y^\circ,y}^2+\tau_Q
       +C\{V_{\max}(\epsilon+\eta)+V_{\max}^2x/\mathcal N_b\}.
 \label{eq:mix-erm}
\end{align}
Let the oracle tolerance decrease to zero. No equality of group design laws was used, and no empirical loss was fitted to only one group.

For the Gram bound, work on $\operatorname{Range}(\Sigma_{\rm st})$ and whiten by
$\widetilde w=\Sigma_{\rm st}^{\dagger/2}w_{\rm st}$. The pooled law gives $\mathbb E_{\nuplus}\widetilde w\widetilde w^\top=\operatorname{Id}$ and $\|\widetilde w\|^2\le\kappa_{\rm st}:=C_{\rm lev}d_S$ almost surely. Every individual observed marginal shares these properties of the support and envelope because it is a positive-weight component of $\nuplus$. For each block set
\begin{align}
 \Xi_{i,j}=\frac1n\sum_{t\in I_{i,j}}\widetilde w(S_{i,t+1})\widetilde w(S_{i,t+1})^\top,
 \qquad \overline\Xi_c=\sum_{\mathcal I_c}\mathbb E \Xi_{i,j}.
 \label{eq:mix-psd-blocks}
\end{align}
Then $0\preceq \Xi_{i,j}\preceq \kappa_{\rm st}/\mathcal N_b\,\operatorname{Id}$, $\overline\Xi_0+\overline\Xi_1=\operatorname{Id}$, and $0\preceq \overline\Xi_c\preceq \operatorname{Id}$. A group mean can be singular and need not be a scalar multiple of $\operatorname{Id}$. With $\Xi_{i,j}^{\circ}=\Xi_{i,j}-\mathbb E \Xi_{i,j}$, the group variance obeys
\begin{align*}
 \|\Xi_{i,j}^{\circ}\|&\le2\kappa_{\rm st}/\mathcal N_b,\qquad
 \sum_{\mathcal I_c}\mathbb E (\Xi_{i,j}^{\circ})^2
 \preceq\sum_{\mathcal I_c}\mathbb E \Xi_{i,j}^2
 \preceq(\kappa_{\rm st}/\mathcal N_b)\overline\Xi_c\preceq(\kappa_{\rm st}/\mathcal N_b)\operatorname{Id}.
\end{align*}
Matrix Bernstein \citep[Theorem~6.1.1]{Tropp2015}, at group deviation $1/4$, gives failure at most
$2r_{\rm st}\exp(-c\mathcal N_b/d_S)$. Both group deviations imply a full deviation at most $1/2$. Transform back on the nonzero eigenspace and retain the common null space. Integrating the same action-derivative quadratic form against the coefficient functions gives the optional empirical derivative comparison. The numerical derivative matrix, if used, still needs its own stated two-sided comparison.

Allocate the independent-product failure budget $\delta$ among these groupwise events, take their unions separately within each group, and apply Lemma~\ref{lem:mix-comparison} once. This proves the joint probability and avoids paying the dependence comparison separately for each net element, statistic or fitted iteration.
\end{proof}

\subsection{KRR uses the same dependence event, with its own ridge calibration}\label{subsec:mix-krr}
The score and operator bounds have their own ridge calibration.

\begin{proposition}[Uniform KRR prediction under the full-sample block comparison]\label{prop:mix-krr}
Let the measurable kernel have separable RKHS $\mathcal F$, diagonal bound $\kappa_K^2$, and full-domain targets satisfying $\sup_y\|Q_y\|_{\mathcal F}\le {B}_F$. Assume $Y_y,Q_y\in[0,V_{\max}]$, common conditional compatibility, and a full deterministic label log cover ${H}_Y(\eta)$. Write $\Sigma_K,\widehat\Sigma_K$ for the population and empirical RKHS second-moment operators and
\begin{align*}
 \mathbf z_y=\frac1n\sum_{i,t}\{Y_y(Z_{i,t})-Q_y(X_{i,t})\}K_{X_{i,t}}.
\end{align*}
Under Assumption~\ref{ass:mix-array}, one event of probability at least $1-\delta-\delta_b^{\rm mix}$ gives
\begin{align*}
 \sup_y\|\mathbf z_y\|_{\mathcal F}&\le s_{K,n},&
 \|\widehat\Sigma_K-\Sigma_K\|_{\mathrm{op}}&\le v_{K,n},\\
 s_{K,n}&=2\kappa_K\eta+C\kappa_KV_{\max}
 \sqrt{\frac{H_Y(\eta)+\log(16/\delta)+1}{\mathcal N_b}},\\
 v_{K,n}&=C\kappa_K^2\sqrt{\frac{\log(16/\delta)+1}{\mathcal N_b}}.
\end{align*}
The full-sample approximate ridge fit with objective gap $\tau_K$ has the simultaneous raw norm and prediction bounds
\begin{align}
 {B}_{K,n}&={B}_F+s_{K,n}/\varrho_n+\sqrt{\tau_K/\varrho_n},\label{eq:mix-krr-fitted}\\
 \delta_{K,n}^2&=\varrho_n{B}_F^2+2({B}_{K,n}+{B}_F)s_{K,n}+\tau_K
                 +v_{K,n}({B}_{K,n}+{B}_F)^2.
 \label{eq:mix-krr-risk}
\end{align}
Clipping does not increase error toward the bounded target. If $\varrho_n\to0$, $s_{K,n}+v_{K,n}=O(\varrho_n)$ and $\tau_K=O(\varrho_n)$, then ${B}_{K,n}=O(1)$ and $\delta_{K,n}=O(\varrho_n^{1/2})$. The raw-kernel secant conditions, when verified with controlled constants, give evaluation error $\widetilde O(\varrho_n^{\vartheta_*/2})$.
\end{proposition}
This uses the same block comparison as \citet{Yu1994}, now on the ridge score and second-moment terms. The fit-and-target uniformity emphasized by \citet{NguyenTangEtAl2022Besov} is retained, and the fitted-norm bound supplies the coefficient radius required by the functional secant calculation.
\begin{proof}
Fix a group $c$ and first work under its product-block law. Write
\[
 \mathbf z_{i,j,y}^{\rm blk}=\frac1n\sum_{t\in I_{i,j}}
       \{Y_y(Z_{i,t})-Q_y(X_{i,t})\}K_{X_{i,t}},\qquad
 \mathbf z_{c,y}^{\rm blk}=\sum_{\mathcal I_c}\mathbf z_{i,j,y}^{\rm blk},\qquad
 \ell_c^{(2)}=\sum_{\mathcal I_c}\ell_{i,j}^2/n^2 .
\]
Conditional compatibility gives $\E \mathbf z_{i,j,y}^{\rm blk}=0$, without independence within a block. Thus
\[
 \|\mathbf z_{i,j,y}^{\rm blk}\|\le\kappa_KV_{\max}\ell_{i,j}/n,\quad
 \E\|\mathbf z_{c,y}^{\rm blk}\|^2=\sum_{\mathcal I_c}\E\|\mathbf z_{i,j,y}^{\rm blk}\|^2
       \le\kappa_K^2V_{\max}^2\ell_c^{(2)},\quad
 \ell_0^{(2)}+\ell_1^{(2)}\le\mathcal N_b^{-1}.
\]
The cross terms vanish by product-block independence. Replacing one block changes $\|\mathbf z_{c,y}^{\rm blk}\|$ by at most $2\kappa_KV_{\max}\ell_{i,j}/n$. Bounded differences \citep{BoucheronLugosiMassart2013} and Jensen therefore imply
\[
 \Pr_{\otimes,c}\!\left\{\|\mathbf z_{c,y}^{\rm blk}\|>
       \kappa_KV_{\max}\sqrt{\ell_c^{(2)}}(1+\sqrt{2x})\right\}
       \le e^{-x}.
\]
An empty group contributes zero. Use $x={H}_Y(\eta)+\log(16/\delta)$ and a deterministic label net in each group. Since
$\|Q_y-Q_{y'}\|_\infty\le\|y-y'\|_\infty$, the two groups' off-net costs sum to $2\kappa_K\eta$, not a term proportional to the number of blocks. Cauchy--Schwarz gives
$\sqrt{\ell_0^{(2)}}+\sqrt{\ell_1^{(2)}}\le\sqrt{2/\mathcal N_b}$, proving the score radius.

For the second-moment operator, let $\BlockCovOp_{i,j}^{\rm blk}=n^{-1}\sum_{t\in I_{i,j}}K_{X_{i,t}}\otimes K_{X_{i,t}}$ and
$\Delta_c^\Sigma=\sum_{\mathcal I_c}(\BlockCovOp_{i,j}^{\rm blk}-\E \BlockCovOp_{i,j}^{\rm blk})$.
The Hilbert--Schmidt identity $\|K_X\otimes K_X\|_{\rm HS}=\|K_X\|^2\le\kappa_K^2$ gives
\[
 \E\|\Delta_c^\Sigma\|_{\rm HS}^2
 \le\sum_{\mathcal I_c}\E\|\BlockCovOp_{i,j}^{\rm blk}\|_{\rm HS}^2
 \le\kappa_K^4\ell_c^{(2)}.
\]
The replacement bound is $2\kappa_K^2\ell_{i,j}/n$. The same norm argument, now with $x=\log(16/\delta)$ and no label net, controls $\|\Delta_c^\Sigma\|_{\rm HS}$. Summing groups and using operator norm at most Hilbert--Schmidt norm yields $v_{K,n}$. Within each group combine the score and second-moment bad events, then apply Lemma~\ref{lem:mix-comparison} once.
These score and operator bounds are the only sampling-specific inputs to the ridge argument. Substitute them in the normal-equation and approximate-minimizer bounds of Proposition~\ref{thm:krr-common-prediction} to obtain \eqref{eq:mix-krr-fitted}--\eqref{eq:mix-krr-risk}; its algebra uses the original full-sample objective, not a fitted block subsample. Corollary~\ref{cor:rev-krr-boundary} then supplies the boundary calibration. The label cover remains fixed before fitting, and the same raw-secant and clipping argument gives evaluation.
\end{proof}
The common label entropy is fixed before fitting. Combining these product-law events with the objective events before the group comparison pays the same dependence failure only once.
\subsection{Retained-block comparison and omission costs}\label{subsec:mix-retained}
Under Assumption~\ref{ass:mix-array}, consider the following alternative to the full-sample two-group partition.
Retain the first $b_{\rm ret}$ transitions in each complete cycle of length $b_{\rm ret}+q_{\rm ret}$, with $b_{\rm ret},q_{\rm ret}\ge1$:
\begin{align*}
 m_{i,{\rm ret}}&=\lfloor T/(b_{\rm ret}+q_{\rm ret})\rfloor,&
 I^{\rm ret}_{i,j}&=\{(j-1)(b_{\rm ret}+q_{\rm ret})+1,\ldots,(j-1)(b_{\rm ret}+q_{\rm ret})+b_{\rm ret}\},\\
 \mathcal N_{\rm ret}&=\sum_i m_{i,{\rm ret}},&\mathsf{o}_{\rm ret}&=1-b_{\rm ret}\mathcal N_{\rm ret}/n.
\end{align*}
Equivalently, $\mathcal N_{\rm ret}=\Neff(b_{\rm ret},q_{\rm ret})$ and
$\mathsf{o}_{\rm ret}=\mathsf{o}_{b_{\rm ret},q_{\rm ret}}$; all incomplete cycles are included in $\mathsf{o}_{\rm ret}$.
\phantomsection\label{not:retained-blocks}

Consecutive retained blocks are separated by $q_{\rm ret}+1$ indices.
Apply the product-law argument in Lemma~\ref{lem:mix-comparison} with this separation and use monotonicity of the mixing envelope.
Successive maximal coupling on standard Borel spaces \citep[Corollary~4.2.5]{Berbee1979}
then gives independent copies with the same block marginals and joint failure at most
\begin{equation}
 \delta_{\mathrm{coup}}(b_{\rm ret},q_{\rm ret})
 \le\sum_i(m_{i,{\rm ret}}-1)_+\beta_{\rm mix}(q_{\rm ret}).
 \label{lem:block-coupling-Ti}
\end{equation}
On this one event all deterministic blockwise function evaluations agree.
For $|f|\le B$, retained and omitted averages satisfy
\begin{align}
 \widehat{\mathbb P}&=(1-\mathsf{o}_{\rm ret})\widehat{\mathbb P}_{\rm ret}+\mathsf{o}_{\rm ret}\widehat{\mathbb P}_{\rm omit},\nonumber\\
 \sup_f|(\widehat{\mathbb P}-\mathbb P)f|
 &\le4B\mathsf{o}_{\rm ret}+\sup_f|(\widehat{\mathbb P}_{\rm ret}-\mathbb P_{\rm ret})f|.
 \label{eq:retained-objective-comparison}
\end{align}
The copied block means have $\mathcal N_{\rm ret}$ independent units. Hoeffding concentration and a deterministic cover
therefore give the objective bound with scale $\mathcal N_{\rm ret}$ and the additional $4B\mathsf{o}_{\rm ret}$,
at failure $\delta+\delta_{\mathrm{coup}}(b_{\rm ret},q_{\rm ret})$. The fitted objective continues to use all transitions and $\mathbb P_{\rm ret}$ is not identified with $\mathbb P$.

With $\mathcal N_{\rm ret}\ge1$, the full-sample prediction bound is
\begin{align}
 \delta_{Q,n}^2\le C\left\{\eta_Q^2+\tau_Q+V_{\max}(\epsilon+\eta)
 +\frac{V_{\max}^2[H(\epsilon,\eta)+\log(C/\delta)]}{\mathcal N_{\rm ret}}
 +V_{\max}^2\mathsf{o}_{\rm ret}\right\}.
 \label{eq:mix-retained-prediction}
\end{align}
The objective comparison \eqref{eq:retained-objective-comparison} carries an $O(\mathsf{o}_{\rm ret})$ term. The full-sample relative Gram event additionally requires
\begin{align}
 \mathcal N_{\rm ret}\ge Cd_S\log(C(d_S\vee1)/\delta),\qquad Cd_S\mathsf{o}_{\rm ret}\le\tfrac18.
 \label{eq:mix-retained-gram-condition}
\end{align}
All three events have joint failure at most $\delta+\delta_{\mathrm{coup}}(b_{\rm ret},q_{\rm ret})$.

For completeness, normalize retained and omitted sums by the full $n$, writing $\widehat{\mathbb P}=\widehat{\mathbb P}_{\rm ret}^{\,0}+\widehat{\mathbb P}_{\rm omit}^{\,0}$ and likewise for their expectations. For the excess loss $d=d_{f,y}$, compatibility gives $\mathbb P_{\rm ret}^{\,0}d,\mathbb P_{\rm omit}^{\,0}d\ge0$. The same blockwise Bernstein and product-net argument as \eqref{eq:mix-relative-full} yields
\[
 |(\widehat{\mathbb P}_{\rm ret}^{\,0}-\mathbb P_{\rm ret}^{\,0})d|
 \le\tfrac14\mathbb P_{\rm ret}^{\,0}d+
 C\{V_{\max}(\epsilon+\eta)+V_{\max}^2x/\mathcal N_{\rm ret}\}.
\]
Since $|d|\le V_{\max}^2$, the omitted centered term is at most $2V_{\max}^2\mathsf{o}_{\rm ret}$. Use $\mathbb P_{\rm ret}^{\,0}d\le \mathbb P d$ and the full empirical minimization inequality \eqref{eq:mix-erm} to obtain \eqref{eq:mix-retained-prediction}; no equality of the design laws is used. For Gram control, globally whiten as in \eqref{eq:mix-psd-blocks}. Retained block matrices are bounded by $Cd_Sb_{\rm ret}/n\le Cd_S/\mathcal N_{\rm ret}$ and their total expectation is at most $\operatorname{Id}$. Matrix Bernstein controls their deviation by $1/4$ under the sample-size condition in~\eqref{eq:mix-retained-gram-condition}. The omitted empirical and population matrices each have norm at most $Cd_S\mathsf{o}_{\rm ret}$, establishing the sufficiency of the omitted-mass condition $Cd_S\mathsf{o}_{\rm ret}\le1/8$. Transfer this one joint event from the product law with cost $\delta_{\mathrm{coup}}(b_{\rm ret},q_{\rm ret})$.

\noindent\textit{KRR score and second-moment bounds.} Retain the full-domain target-ball and kernel conditions of Proposition~\ref{thm:krr-common-prediction}. Replace \eqref{eq:krr-unwhitened-score}--\eqref{eq:krr-unwhitened-covariance} by
\begin{gather*}
 s_n^{\mathrm{RB}}=2\kappa_K\eta+C\kappa_K\Vmax
 \left[\mathsf{o}_{\rm ret}+\sqrt{\{H_Y(\eta)+\log(8/\delta)+1\}/\mathcal N_{\rm ret}}\right],\\
 v_n^{\mathrm{RB}}=C\kappa_K^2
 \left[\mathsf{o}_{\rm ret}+\sqrt{\{\log(8/\delta)+1\}/\mathcal N_{\rm ret}}\right].
\end{gather*}
With these substitutions, the radius, prediction and common evaluation bounds hold for the full-sample KRR fit with probability at least $1-\delta-\delta_{\mathrm{coup}}(b_{\rm ret},q_{\rm ret})$.
Boundary tuning from Corollary~\ref{cor:rev-krr-boundary}, with $N$ replaced by $\mathcal N_{\rm ret}$ in its verified full-label entropy and radii, gives a bounded fitted norm if additionally $\mathsf{o}_{\rm ret}=O(\varrho_n)$ and $\delta_{\mathrm{coup}}(b_{\rm ret},q_{\rm ret})\to0$. The stricter $\mathsf{o}_{\rm ret}=o(\varrho_n)$ retains the interior conclusion ${B}_{K,n}={B}_F+o(1)$ when the other interior premises hold.

Apply the Hilbert-valued score and second-moment concentration arguments of Proposition~\ref{thm:krr-common-prediction} to the independent copied blocks. Their envelopes are $\kappa_K\Vmax$ and $\kappa_K^2$, and their independent-unit count is $\mathcal N_{\rm ret}$. On the one shared coupling event, decompose the original full-sample statistics into retained and omitted averages:
\[
 \mathbf z_y=(1-\mathsf{o}_{\rm ret})\mathbf z_{{\rm ret},y}+\mathsf{o}_{\rm ret}\mathbf z_{{\rm omit},y},\qquad
 \widehat\Sigma_K=(1-\mathsf{o}_{\rm ret})\widehat\Sigma_{K,{\rm ret}}+\mathsf{o}_{\rm ret}\widehat\Sigma_{K,{\rm omit}} .
\]
Bounded sections give
$\sup_y\|\mathbf z_y-\mathbf z_{{\rm ret},y}\|\le2\kappa_K\Vmax \mathsf{o}_{\rm ret}$ and
$\|\widehat\Sigma_K-\widehat\Sigma_{K,{\rm ret}}\|_{\mathrm{op}}+\|\Sigma_K-\Sigma_{K,{\rm ret}}\|_{\mathrm{op}}
\le4\kappa_K^2\mathsf{o}_{\rm ret}$.
These are precisely the additive terms in the stated radii. The full-sample ridge normal equation and basic inequality are unchanged, so \eqref{eq:krr-fitted-radius}--\eqref{eq:krr-new-prediction-radius} apply with the substituted radii. Dividing the omission terms by $\varrho_n$ proves the boundary or interior conclusion, respectively. The failure is $\delta+\delta_{\mathrm{coup}}(b_{\rm ret},q_{\rm ret})$, without multiplying the coupling cost by the label cover.

\subsection{Dependence scales and value consequences}\label{subsec:mix-rate}\label{subsec:mix-decay}
For bounded joint entropy cost $h_n$, the proved concentration inputs are
$u_n=O(\sqrt{h_n/\mathcal N_b})$ and a squared prediction remainder $O(h_n/\mathcal N_b)$.
Together with pair membership, the numerical-input and policy-optimization bounds in~\eqref{eq:main-uniform-error-event}, and $\mathcal R_n$, these uniform statistical events imply $\mathcal E_n$. Theorem~\ref{thm:finite-value-general} therefore accepts these errors on the same deterministic localized policy class, with dependence failure $\delta_b^{\rm mix}$ included in the bound for $\Pr(\mathcal E_n^c)$; under stable propagation, Theorem~\ref{thm:main-value} applies.
For either the finite-state product class or the complete continuous-state RKHS ball of Appendix~\ref{sec:structural-rates}, replace $N$ by $\mathcal N_b$ in the entropy balance, retune the deterministic critic classes and tolerances, and choose $M\ge c\log\mathcal N_b$ for sufficiently large $c$. The corresponding exponent is
$p=\min\{\vartheta_*,1/2\}/(2+D)$, with $D=d_{\rm app}$ or ${D}_{\rm ct}$; retain $\sqrt{\Delta_{\varepsilon_{\mathrm{app}}}}$ in the coupled case. Corollary~\ref{cor:growing-transfer} gives the analogous exponent when critic factors grow. All information, state-transfer and approximation conditions are imposed at this retuned scale for the same pooled laws.
This substitution follows from the full-sample objective and prediction bounds above. The deterministic policy penalty requires no empirical-to-population comparison; the fixed-radius norm-localization condition is retuned at the same effective sample size.

If $\beta_{\rm mix}(k)\le C_m\exp(-c_m k^{p_{\rm mix}})$, $p_{\rm mix}>0$, choose
\[
 b=\min\{T,\lceil[c\log(en)]^{1/p_{\rm mix}}\rceil\}
\]
with sufficiently large $c$. If uncapped, $\delta_b^{\rm mix}\le(n/b)\beta_{\rm mix}(b)\le n^{-2}$; if capped, each subject has a single block and $\delta_b^{\rm mix}=0$.
Thus $\mathcal N_b\asymp N\max\{1,T/\log^{1/p_{\rm mix}}(en)\}$, up to constants and integer rounding.
Assumption~\ref{ass:main-mixing} is the case $p_{\rm mix}=1$, for which $b=O(\log(en))$. Hence the approximation scale $\mathcal N_b^{-1/(2+D)}$ equals $b^{1/(2+D)}(NT)^{-1/(2+D)}$. All $NT$ transitions remain in the fitted objectives.

For Example~\ref{subsec:main-joint-rate} and its continuous-state extension below, take $D=D_{\rm Ada}$ and $D=D_{\rm ct}$, respectively. Choose approximation accuracy $\varepsilon_{\mathrm{app},n}$ of order $(NT)^{-1/(2+D)}$ times a sufficiently large logarithmic factor incorporating $b^{1/(2+D)}$, the covering terms, and the confidence level. Require critic squared-loss fitting error $O(\varepsilon_{\mathrm{app},n}^2)$, numerical input error $O(\varepsilon_{\mathrm{app},n})$, and policy optimization error and penalty coefficient $O(\varepsilon_{\mathrm{app},n})$, with the full-domain approximation, covering, localization, and bounded-transfer conditions of the corresponding structural corollary. The uniform logged prediction error $\delta_{\mathrm{Ada},n}$ and critic-update error then satisfy
\[
 \delta_{\mathrm{Ada},n}=\widetilde O((NT)^{-1/(2+D)}),\qquad
 \varepsilon_{Q,n}=\widetilde O((NT)^{-\vartheta_*/(2+D)}).
\]
The policy error is $\widetilde O(\varepsilon_{\mathrm{app},n}^{1/2}+\sqrt{\Delta_n})$. With $M\ge c\log(NT)$ for sufficiently large fixed $c$, Theorem~\ref{thm:main-value} gives the value rate in Example~\ref{subsec:main-joint-rate} and the continuous-state bound below. The one-score choices $d=s_h=s_b=1$ and $p_{\rm Lip}=0$ give $D_{\rm Ada}=1$; when $\vartheta_*\ge1/2$ and $\Delta_n=0$, the finite-state excess rate is $(NT)^{-1/6}$ up to logarithms.

\paragraph{Continuous-state value rate.}
For $s\in[0,1]$ and $\mathcal H=H^1([0,1];H^2(0,1))$, let the outer targets depend on the state and $d$ functional scores. With the smoothness orders $s_h,s_b$ and fitted-class Lipschitz exponent $p_{\rm Lip}$ specified in Appendix~\ref{subsec:continuous-full-rate}, the joint complexity exponent is
\[
 D_{\rm ct}=\max\{(1+d)/s_h,\ 1/s_b,\ 5+5p_{\rm Lip}/2\}.
\]
Under the full-policy-ball approximation, covering, localization, and state-transfer conditions of that appendix, Assumption~\ref{ass:main-mixing}, and the preceding computational calibration,
\[
 J^\star-J(\widehat\pi_M)
 \le a_{\mathcal H}+\widetilde O_p\!\left(
 (NT)^{-\min\{\vartheta_*,1/2\}/(2+D_{\rm ct})}+\sqrt{\Delta_n}\right).
\]
For $d=s_h=s_b=1$, $p_{\rm Lip}=0$, $\vartheta_*\ge1/2$, and $\Delta_n=0$, one has $D_{\rm ct}=5$ and excess value rate $\widetilde O_p((NT)^{-1/14})$. Its slower exponent reflects the larger policy-class complexity. These rates are sufficient and are not asserted to be optimal.

If $\beta_{\rm mix}(k)\le C_m k^{-p_{\rm poly}}$, $p_{\rm poly}>0$, take
\[
 b=\min\{T,\lceil[C_m n\log^2(en)]^{1/(1+p_{\rm poly})}\rceil\}.
\]
Then $\delta_b^{\rm mix}\le\log^{-2}(en)$ when uncapped, and
\[
 \mathcal N_b\asymp
 \max\{N,\ n^{p_{\rm poly}/(1+p_{\rm poly})}\log^{-2/(1+p_{\rm poly})}(en)\}.
\]
For prescribed failure $\epsilon_n$, replace $\log^2(en)$ by $1/\epsilon_n$.
The confidence level determines the corresponding scale explicitly.
With $T\asymp N^{p_{\rm len}}$, a verified power $p$ of this scale becomes, up to logarithms,
$N^{-\max\{1,(1+p_{\rm len})p_{\rm poly}/(1+p_{\rm poly})\}p}$.
For unequal lengths the same block proof uses $T_{\max}$ as the cap and $\mathcal N_{\infty}=n/T_{\max}$ in place of $N$ in these maxima.
All approximation, information and state-transfer conditions continue to refer to the corresponding pooled laws.

For the retained-block alternative, its proved prediction radius additionally includes
$\Vmax\sqrt{\mathsf{o}_{\rm ret}}$ and its objective bound includes $O(\mathsf{o}_{\rm ret})$.
Apply the main theorem with these actual inputs and the separately verified norm-localization event; no equality between retained and full design laws is needed.
The temporal calculation in Appendix~\ref{subsec:mix-scope} verifies the stronger main-text rate for Example~\ref{ex:two-state-control} using the full-sample construction.

\section{Verification of Example~\ref{ex:two-state-control}}\label{app:verification-examples}
All calculations in this section concern Example~\ref{ex:two-state-control}, with the same MDP, behavior law, Sobolev-valued policy RKHS, fixed comparison ball, and full functional linear critic kernel.

\subsection{The model, policy class, behavior support and moments}
Take $\mathcal S=\{1,2\}$ and $\mathcal A=\{a\in L^2(0,1):\|a\|_{L^2}\le B_A\}$. For the discount factor $\gamma\in(0,1)$, define
\begin{gather*}
 q_\circ(u)=e^u-(e-1),\qquad
 v_\circ=\frac{q_\circ}{\|q_\circ\|_{L^2}},\\
 z(a)=\langle v_\circ,a\rangle_{L^2},\qquad
 \iota=\frac{\min\{1,\sqrt{\gamma^{-2}-1}\}}{8B_A}.
\end{gather*}
The smooth function $v_\circ\in H^2(0,1)$ has unit $L^2$ norm and zero integral. Specify the reward and next-state laws by
\begin{gather*}
 R\mid(S=1,A=a)\sim\operatorname{Bernoulli}(1/4),\qquad
 R\mid(S=2,A=a)\sim\operatorname{Bernoulli}(3/4),\\
 \Pnext(2\mid1,a)=\tfrac12+\iota z(a),\qquad
 \Pnext(2\mid2,a)=\tfrac12-\iota z(a),\qquad
 \Pnext(1\mid s,a)=1-\Pnext(2\mid s,a).
\end{gather*}
Use fresh reward and transition randomization at each time, independently conditional on the current state and action.

Let $\mathcal G=H^2(0,1)$ with $\|a\|_{\mathcal G}^2=\|a\|_{L^2}^2+\|a''\|_{L^2}^2$ and
$\mathcal H=\mathcal G\oplus\mathcal G$.
With $T_{\mathcal G}$ defined by
$\langle T_{\mathcal G}a,b\rangle_{\mathcal G}=\langle a,b\rangle_{L^2}$, the $L^2$-valued operator kernel is $\Gamma(s,t)=\ind\{s=t\}T_{\mathcal G}$. The base class is the product of the feasible $L^2$ action balls within $\mathcal H$. The analytical comparison class is
\[
 \Pi_{\mathcal H}(B_{\mathcal H})
 =\left\{\pi\in\Pi:\sum_{s=1}^2\|\pi(s)\|_{\mathcal G}^2
                         \le B_{\mathcal H}^2\right\}.
\]
Consequently every selected action has $\mathcal G$ norm at most $B_{\mathcal H}$. This localized class is compact for statewise $L^2$ convergence: bounded $H^2$ sets are weakly compact and embed compactly into $L^2$, and the constraints are closed.

Choose $\varsigma_A>0$ so that
$\varsigma_A^2\sum_{j\ge1}j^{-9/2}\le{B}_A^2$.
With independent symmetric signs $\xi_j$ and the normalized shifted Legendre basis $\{\phi_j\}$, draw $A=\varsigma_A\sum_{j\ge1}j^{-9/4}\xi_j\phi_j$. This behavior action series converges in $L^2$ and is feasible.
Its mean is zero and covariance is $\varsigma_A^2\BaseWeightOp_{9/4}$.
Fresh draws are independent of current states and past records.
The behavior-averaged transition row is $(1/2,1/2)$, so uniform initial states give
$\nu_X=\frac12(\delta_1+\delta_2)\otimes\law(A)$ and $\nu_+=\frac12(\delta_1+\delta_2)$ at every time.

For every support point of $\law(A)$, its $j$th coefficient has magnitude $\varsigma_A j^{-9/4}$.
Indeed the series is a continuous map from the compact sign product space into $L^2$, since its squared tails converge uniformly; its image is closed and is the support of the pushforward product-sign law.
Thus
$\sum_jj^4|\langle A,\phi_j\rangle|^2=\varsigma_A^2\sum_jj^{-1/2}=\infty$ on the entire support.
Every action supplied by $\Pi_{\mathcal H}(B_{\mathcal H})$ instead satisfies the finite bound
$\|a\|_{\zeta}\le\sqrt{c_{\mathcal G}}B_{\mathcal H}$ at $\zeta=2$.

To show a positive distance, choose a finite $m$ with
$\varsigma_A(\sum_{j\le m}j^{-1/2})^{1/2}>2\sqrt{c_{\mathcal G}}B_{\mathcal H}$.
For a support point $x$ and a policy-selected action $a$, the triangle inequality in the weighted first-$m$ coordinates gives
\[
 m^2\|x-a\|_{L^2}
 \ge\left(\sum_{j\le m}j^4|\langle x-a,\phi_j\rangle|^2\right)^{1/2}
 \ge \varsigma_A\left(\sum_{j\le m}j^{-1/2}\right)^{1/2}
                   -\sqrt{c_{\mathcal G}}B_{\mathcal H}>0.
\]
This one bound applies to all actions selected by the localized policy class and all support points.

\subsection{Full linear critics and spectral transfer}
The critic RKHS of
$K\{(s,a),(s',a')\}=\ind\{s=s'\}(1+\langle a,a'\rangle_{L^2})$
is isometric to $(\mathbb R\oplus L^2)^2$.
Each critic has an intercept and arbitrary $L^2$ slope at each state.
The bounded action ball makes $K(x,x)\le1+{B}_A^2$.
The feature span is the whole coefficient space, and
\[
 \Sigma=\tfrac12\operatorname{diag}(1,\varsigma_A^2\BaseWeightOp_{9/4},1,\varsigma_A^2\BaseWeightOp_{9/4}),
 \qquad
 \WeightOp=\operatorname{diag}(1,\BaseWeightOp_{9/4},1,\BaseWeightOp_{9/4}).
\]
For every coefficient direction,
$\langle d,\Sigma d\rangle\ge\frac12\min\{1,\varsigma_A^2\}\langle d,\WeightOp d\rangle$.
This is a quadratic-form comparison in the same coordinates.
For $\vartheta=8/9$, $2(9/4)\vartheta=4$ and
\[
 \int\sup_{\pi\in\Pi_{\mathcal H}(B_{\mathcal H})}\|\WeightOp^{-\vartheta/2}
            K_{(s,\pi(s))}\|^2\nu_+(ds)
 \le1+c_{\mathcal G}B_{\mathcal H}^2.
\]
Proposition~\ref{thm:smooth-spectral-transfer} therefore gives uniform H\"older transfer for raw coefficient secants in any fixed radius ball.
When the target is bounded in $[0,\Vmax]$, clipping a fitted raw critic does not increase its pointwise error to the target.

The same family has no finite uniform linear error ratio.
For $j\ge3$ put $p_j(u)=\int_0^u(u-t)\phi_j(t)\,dt$.
Writing $c_j=\{2\sqrt{(2j-1)(2j+1)}\}^{-1}$, Legendre integration gives
\begin{equation}
 p_j=c_jc_{j+1}\phi_{j+2}
       -(c_j^2+c_{j-1}^2)\phi_j+c_{j-1}c_{j-2}\phi_{j-2}.
 \label{eq:double-primitive}
\end{equation}
For $j\ge2$, the Legendre recurrence and derivative identities
\citep[Section~14.10]{NISTDLMFLegendre} give
$\int_0^u\phi_j(t)\,dt=c_j\phi_{j+1}(u)-c_{j-1}\phi_{j-1}(u)$;
the constant is zero at $u=0$ by endpoint values.
Integrating again proves~\eqref{eq:double-primitive}.
The coefficients have order $j^{-2}$, $p_j''=\phi_j$, and hence
$\|p_j\|_{L^2}=O(j^{-2})$, $\|p_j\|_{\mathcal G}=O(1)$ and
$|\langle p_j,\phi_j\rangle|\asymp j^{-2}$.

A single sufficiently small positive multiple of every $p_j$ belongs to $\Pi_{\mathcal H}(B_{\mathcal H})$, taking the other state's action to be zero. Indeed, $\sup_j\|p_j\|_{\mathcal G}<\infty$, so one fixed multiplier can satisfy both action feasibility and the policy-norm bound.
For a fixed amplitude $\kappa_0>0$ chosen small enough that clipping remains inactive, the critic differences
$\Delta_j(s,a)=\kappa_0\ind\{s=1\}\langle a,\phi_j\rangle$ then satisfy
\[
 \epsilon_j=\kappa_0\varsigma_Aj^{-9/4}/\sqrt2,\qquad
 c j^{-2}\le\upsilon_j\le C j^{-2}.
\]
The upper bound uses the Sobolev spectral embedding and the lower bound the preceding feasible actions.
Thus $\upsilon_j/\epsilon_j\asymp j^{1/4}\to\infty$.
Also $\upsilon_j/\epsilon_j^{\vartheta'}\to\infty$ for every $\vartheta'>8/9$, so this exponent cannot be improved uniformly over these bounded coefficient secants.
They can be realized between bounded critics by adding a fixed interior intercept and using the stated small amplitude $\kappa_0$.

\subsection{State stability and the globally optimal policy}
The smooth normalized function $v_\circ$ has unit $L^2$ norm, zero first Legendre coefficient and infinitely many nonzero coefficients; otherwise the exponential minus its mean would be a polynomial.
For all feasible actions, $|z(a)|\le{B}_A$ and $\iota{B}_A\le1/8$, so all transition probabilities lie in $[3/8,5/8]$.
Uniform $\nu_+$ gives, for every probability occupancy,
\[
 C_{\rm occ}\le\sqrt2,\qquad
 \sum_{t=1}^2\frac{\Pnext(t\mid s,a)^2}{1/2}
       =1+4\iota^2z(a)^2\le1+4\iota^2{B}_A^2.
\]
The state-density calculation in Appendix~\ref{sec:bellman} gives
$C_B\le\sqrt{1+4\iota^2{B}_A^2}$ for every residual, including those involving $Q^\dagger$.
The chosen $\iota$ implies
$\chi^2\le\gamma^2[1+\min\{1,\gamma^{-2}-1\}/16]<1$.

For a continuation vector $c\in[0,\Vmax]^2$, the full-domain targets are
\begin{align}
 Q_c(1,a)&=\tfrac14+\gamma(c_1+c_2)/2
                      +\gamma\iota(c_2-c_1)\langle v_\circ,a\rangle,\nonumber\\
 Q_c(2,a)&=\tfrac34+\gamma(c_1+c_2)/2
                      -\gamma\iota(c_2-c_1)\langle v_\circ,a\rangle.
 \label{eq:running-targets}
\end{align}
At the global optimum the two maximized continuation terms coincide, so
$V^\star(2)-V^\star(1)=1/2$.
The actions $\pi^\star(1)={B}_Av_\circ$ and
$\pi^\star(2)=-{B}_Av_\circ$ maximize those terms by Cauchy--Schwarz.
This gives a solution of the global Bellman equations, which is unique by discounted contraction.
The actions lie in the base class $\Pi$. Moreover,
$\|\pi^\star\|_{\mathcal H}^2=2B_A^2\|v_\circ\|_{\mathcal G}^2$, so the condition
$B_{\mathcal H}^2\ge2B_A^2\|v_\circ\|_{\mathcal G}^2$ places them in $\Pi_{\mathcal H}(B_{\mathcal H})$.
Thus $J_{\mathcal H}^\star=J^\star$ equals the global value and $a_{\mathcal H}=0$ in this example.

\subsection{Full-linear ridge fitting for coverage and value}
Every possible continuation from any clipped critic and any policy is a vector $c\in[0,\Vmax]^2$.
The labels $Y_c=R+\gamma c_{S'}$ have log covering number at most
$2\log(1+C\Vmax/\epsilon)$, independently of the fitted critic radius.
Formula~\eqref{eq:running-targets} places every target in a fixed ball of the full critic RKHS: intercepts and slopes are bounded by constants depending only on $\gamma,\iota$.
All numerical directions of the fitted critic remain available; the known $v_\circ$ only defines the target dynamics.

Use $\delta_N=N^{-2}$, net accuracy $N^{-1}$ and ridge
$\varrho_n=c_\varrho\sqrt{\log(eN)/N}$ for a sufficiently large fixed $c_\varrho$.
The score and operator bounds of Proposition~\ref{thm:krr-common-prediction} are $O(\varrho_n)$.
A KRR objective gap $O(\varrho_n)$ gives
${B}_{K,n}=O(1)$ and $\delta_{K,n}=\widetilde O(N^{-1/4})$ on one full-label event.
No positive minimum eigenvalue of the infinite covariance is used.
Its bounded raw coefficients also give a common action-Lipschitz constant for all fitted clipped critics on this event.
Put
\begin{equation}
       \eta_N:=C\delta_{K,n}^{8/9}=\widetilde O(N^{-2/9}),
       \label{eq:running-eta}
\end{equation}
where $C$ is a sufficiently large fixed constant covering the uniform coefficient radii.
The spectral calculation above gives $\varepsilon_{Q,n}\le\eta_N$.

The same calculation also supplies a policy-norm-dependent bound before localization is known. Write a raw fitted critic and its target at state $s$ as
$\widehat b_{s,c}+\langle\widehat d_{s,c},a\rangle$ and
$b_{s,c}+\langle d_{s,c},a\rangle$, respectively.
Because the logged signs are centered and have diagonal covariance, the prediction bound implies, uniformly in $s$ and $c$,
\[
 |\widehat b_{s,c}-b_{s,c}|\le C\delta_{K,n},\qquad
 \|\BaseWeightOp_{9/4}^{1/2}(\widehat d_{s,c}-d_{s,c})\|_{L^2}
      \le C\delta_{K,n}.
\]
The fitted and target RKHS radii also give
$\|\widehat d_{s,c}-d_{s,c}\|_{L^2}=O(1)$.
Interpolating these two slope norms at exponent $8/9$, pairing against the order-two spectral norm of $a$, and using the Sobolev--spectral embedding yields
 \begin{equation}
 |\widehat Q_c(s,a)-Q_c(s,a)|
 \le C\eta_N\{1+\|a\|_{\mathcal G}\},
 \qquad \|a\|_{L^2}\le B_A,\quad a\in H^2(0,1).
 \label{eq:running-norm-dependent-error}
\end{equation}
Clipping does not enlarge the left side because $Q_c\in[0,\Vmax]$.
Unlike a supremum over an arbitrary critic ball,~\eqref{eq:running-norm-dependent-error} uses the prediction-error bound for the actual KRR fits and therefore becomes sharper as $N$ grows.

\subsection{Policy-norm self-localization and the value rate}
For each continuation vector $c$, the target in~\eqref{eq:running-targets} has a pointwise maximizing action
$a_c^\star(s)\in\{0,\pm B_Av_\circ\}$ and
\[
       \|a_c^\star(s)\|_{\mathcal G}^2
       \le B_0:=B_A^2\|v_\circ\|_{\mathcal G}^2.
\]
A conforming spline action $\bar a_{c,d_u}(s)$ can be chosen with $\mathcal G$ norm at most $B_0^{1/2}$ and
$\|\bar a_{c,d_u}(s)-a_c^\star(s)\|_{L^2}\le Cd_u^{-2}$.

Let $\widehat\mu_s=n^{-1}\sum_{i,t}\ind\{S_{i,t+1}=s\}$.
Independent bounded subject averages imply
$\Pr\{\min_s\widehat\mu_s<1/4\}\le4e^{-N/8}$.
On the complementary event the empirical policy objective is statewise separable. The critic term at state $s$ has factor $\widehat\mu_s$, whereas the product-space norm contributes the unweighted term $\lambda_{\pi,n}\|a\|_{\mathcal G}^2$.
Let $R(a)=\|a\|_{\mathcal G}^2$.
Comparing a computed statewise update $\widehat a_s$ with $\bar a_{c,d_u}(s)$, using target optimality and~\eqref{eq:running-norm-dependent-error}, gives
\begin{equation}
 \lambda_{\pi,n}R(\widehat a_s)
 \le C\eta_N\{1+\sqrt{R(\widehat a_s)}\}
      +C\{d_u^{-2}+\lambda_{\pi,n}B_0
                    +\tau_{\rm solver}+\epsilon_{\rm num}\}.
 \label{eq:running-self-localization}
\end{equation}
Here statewise solver errors are obtained directly, or from a total nonnegative gap by dividing by $\min_s\widehat\mu_s$; a uniform numerical objective discrepancy is charged twice and absorbed into $\epsilon_{\rm num}$.

Choose $\lambda_{\pi,n}\asymp\eta_N$, $d_u^{-2}=O(\eta_N)$ and
$\tau_{\rm solver}+\epsilon_{\rm num}=O(\eta_N)$.
The conforming product-spline resolution $d_u\gtrsim N^{1/9}$ suffices for this calibration with the logarithmic majorant in~\eqref{eq:running-eta}.
Writing $x=\sqrt{R(\widehat a_s)}$ and dividing~\eqref{eq:running-self-localization} by $\lambda_{\pi,n}$ gives $x^2\le C_1x+C_2$ with fixed constants. Hence all returned actions have one deterministic $\mathcal G$-norm bound, uniformly over labels and iterations. Taking $B_{\mathcal H}$ large enough for the two-state product norm and for the global optimum places both in the same fixed $\Pi_{\mathcal H}(B_{\mathcal H})$, so $a_{\mathcal H}=0$.

For policy improvement, every action occurring in the pointwise envelope of $\Pi_{\mathcal H}(B_{\mathcal H})$ has squared $\mathcal G$ norm at most $B_{\mathcal H}^2$. Regularized near-optimality, conforming spline approximation and the bounded fitted-critic Lipschitz constant therefore give directly
\begin{equation}
 \sup_{\pi\in\Pi_{\mathcal H}(B_{\mathcal H})}
 \widehat Q_c\{s,\pi(s)\}
       -\widehat Q_c(s,\widehat a_s)
 \le C\{d_u^{-2}+\lambda_{\pi,n}B_{\mathcal H}^2
                    +\tau_{\rm solver}+\epsilon_{\rm num}\}
 =O(\eta_N).
 \label{eq:running-greedy-rate}
\end{equation}
This is a pointwise bound, so it supplies $\varepsilon_{\pi,n}=O(\eta_N)$ without the generic square-root conversion in Lemma~\ref{lem:functional-objective-transfer}. The initialization $Q=0$ obeys the same bounds by comparison with the zero action.

Equations~\eqref{eq:running-eta}, \eqref{eq:running-self-localization} and~\eqref{eq:running-greedy-rate}, together with the verified state constants and $M\ge c\log N$, establish the independent-subject rate $\widetilde O_p(N^{-2/9})$. The following subsection sharpens this to the rate stated in the continuation of Example~\ref{ex:two-state-control}. Thus the full linear critic in Example~\ref{ex:two-state-control} supplies the coverage, localization and value calculations.

\subsection{Temporal refinement for Example~\ref{ex:two-state-control}}\label{subsec:mix-scope}
The independent action draws and reward--transition randomization verify the dependence condition used for the stronger main-text rate.
Given complete records through time $t$, $S_{t+1}$ is known, but averaging the independent action draw and transition at $t+1$ gives
$\law(S_{t+2}\mid Z_1,\ldots,Z_t)=(1/2,1/2)$.
All future draws depend on the past only through this reset state.
Consequently the complete-record absolute regularity coefficient is zero at lags $k\ge2$.
The comparison in Lemma~\ref{lem:mix-comparison} uses the coefficient at lag $b+1$, so block length $b=1$ gives effective sample size $n=NT$ and zero comparison failure.

Apply Proposition~\ref{prop:mix-krr} with label resolution $n^{-1}$, confidence level $n^{-2}$, ridge coefficient $\varrho_n$ of order $\sqrt{\log(en)/n}$, and KRR regularized empirical squared-loss objective gap $O(\varrho_n)$. The fitted norm stays bounded and $\delta_{K,n}=\widetilde O(n^{-1/4})$. Replacing the scale $N$ by $n$ in~\eqref{eq:running-eta}, write
\[
 \eta_n=C\delta_{K,n}^{8/9}=\widetilde O(n^{-2/9}).
\]
Use $\lambda_{\pi,n}\asymp\eta_n$, $d_u^{-2}=O(\eta_n)$, and statewise solver and numerical objective errors of order $O(\eta_n)$. The full-sample objective bound in Proposition~\ref{thm:mix-three-inputs}, applied to state indicators, also places $\min_s\widehat\mu_s\ge1/4$ on an event with probability tending to one at this scale. Intersecting it with the full-label KRR event allows~\eqref{eq:running-self-localization} and~\eqref{eq:running-greedy-rate} to give the same fixed-radius policy localization and
\[
 \varepsilon_{Q,n}+\varepsilon_{\pi,n}=\widetilde O(n^{-2/9}).
\]
The intersected statistical events establish $\mathcal R_n$ and the required uniform bounds, giving $\Pr(\mathcal E_n)\to1$. Theorem~\ref{thm:main-value}, with $a_{\mathcal H}=0$ and $M\ge c\log n$ for sufficiently large fixed $c$, therefore gives
\[
 J^\star-J(\widehat\pi_M)=\widetilde O_p((NT)^{-2/9}),
\]
as stated in the continuation of Example~\ref{ex:two-state-control}. The preceding independent-subject calculation remains available without using the temporal refinement.

\subsection{Numerical observations without imposing smooth observed paths}\label{subsec:rough-numerical-observations}

Finite point evaluations do not define a bounded observation map on arbitrary $L^2$ equivalence classes. Accordingly, the exact-score result places no smoothness requirement on observed actions, whereas an implemented numerical score retains its discrepancy from observing and integrating the action function. A vanishing uniform point-quadrature error over the whole $L^2$ ball is not available.

Here is a useful sufficient condition that does not require a smooth observed action.  Suppose the supplied observation record provides cell averages
$\bar A_k=|I_k|^{-1}\int_{I_k}A(u)du$ and the implemented score is
$\sum_k|I_k|\bar A_k\beta(t_k)$, $t_k\in I_k$.
If the learned basis function is $\mathrm{Lip}_{\beta,n}$-Lipschitz and the largest cell width is $h$, then
\[
 \left|\sum_k|I_k|\bar A_k\beta(t_k)
                   -\langle A,\beta\rangle\right|
 \le\|A\|_{L^2}\|\beta_h-\beta\|_{L^2}
 \le {B}_A\mathrm{Lip}_{\beta,n}h.
\]
Here $\beta_h(u)=\beta(t_k)$ on $I_k$. The cell-average definition and linearity of integration give
\[
 \sum_k|I_k|\bar A_k\beta(t_k)-\langle A,\beta\rangle
 =\sum_k\int_{I_k}A(u)\{\beta(t_k)-\beta(u)\}du
 =\langle A,\beta_h-\beta\rangle.
\]
Cauchy--Schwarz bounds its absolute value by $\|A\|_{L^2}\|\beta_h-\beta\|_{L^2}$. The Lipschitz property gives
\[
 \|\beta_h-\beta\|_{L^2}^2
 \le \mathrm{Lip}_{\beta,n}^2\sum_k\int_{I_k}|u-t_k|^2du
 \le \mathrm{Lip}_{\beta,n}^2h^2\sum_k|I_k|=\mathrm{Lip}_{\beta,n}^2h^2.
\]
Taking square roots proves the bound without any derivative of $A$. This calculation applies only to records supplying cell averages or a justified equivalent reconstruction; raw point samples are not converted into averages. Other observation schemes retain their separately verified numerical discrepancy.

\section{Further related work}\label{app:related-work}
Classical FQI alternates regression and greedy improvement \citep{ErnstEtAl2005}, with finite-sample value analysis based on Bellman error propagation \citep{MunosSzepesvari2008,FarahmandMunosSzepesvari2010}. \citet[Sections~3.2--3.3]{LinEtAl2026} use KRR and average B-spline policy optimization for functional-action learning. Their physical-activity application treats a distribution of daily steps as the action and establishes optimal-policy existence alongside empirical results. We analyze the finite-sample value of the learned policy. Interactive functional control provides another setting: \citet{PanEtAl2018} use descriptor regularity to reduce policy covering numbers while collecting new transitions during training, and \citet{PirmoradEtAl2021SPDE} control spatial forcing functions. Our setting uses fixed offline data.

In the conference version of \citet{AntosEtAl2007}, Assumptions~A1--A2 require a compact Euclidean action set and a positive behavior-density lower bound. Their equation~(3) fits labels with inverse-density-weighted squared loss. Assumption~A6 lower-bounds the fraction of every $\ell^1$ pyramid around an action that remains inside the action set, while Assumptions~A7--A8 impose action-Lipschitz transition, reward and critic functions. Combined with density positivity, this geometry gives polynomial-order behavior probability to small $\ell^1$ neighborhoods and lets nearby logged actions control a selected action. We instead regress under the observed law and use directional information plus RKHS smooth evaluation, so no logged action need be metrically close to a policy action. Their Theorem~3.2 has a sample term proportional, up to logarithms and other constants, to $N^{-1/\{4(d_A+1)\}}$. Substituting $d_A=d_{u,N}\sim c\log N$ gives $N^{-1/\{4(d_{u,N}+1)\}}\to e^{-1/(4c)}$. This only shows that the stated fixed-dimensional bound does not directly give functional consistency; it does not rule out growing-dimensional FQI. Their final value theorem uses an exactly greedy output policy, whereas our output is the learned average-objective policy. Functional neighborhood probabilities also differ from finite-dimensional volume calculations \citep{DelaigleHall2010,Mas2012SmallBall}, although density ratios relative to specified probability laws can exist.

Pessimistic nonlinear fitted-Q analysis uses structural conditions beyond differentiability \citep{YinWangWang2023Differentiable,DiEtAl2024PNLSVI}. Simultaneous control of critics and data-dependent Bellman targets is central to deep nonparametric offline-RL bounds \citep{NguyenTangEtAl2022Besov}; here the labels also range over the policy class and the sampling units are subjects or separated blocks. Intrinsic-dimensional fixed-policy evaluation \citep{JiEtAl2023Manifolds} provides related function-class-relative mismatch ideas, but does not include policy approximation and optimization. AdaFNN learns the basis functions forming its functional scores \citep{YaoMuellerWang2021AdaFNN}, complementing classical functional representations \citep{RamsaySilverman2005,HsingEubank2015}. We impose an explicit label-uniform $L^2$ factorization for approximation and separately prove the learned-basis secant identity; neither step alone supplies the information condition.

The model-based coefficient of \citet[Definition~1]{UeharaSun2022PartialCoverage} is the worst, over candidate transition models, ratio between squared total-variation transition error under a comparator-policy occupancy and under the offline law. Definition~1 of \citet{XieEtAl2021BellmanPessimism} is the model-free analogue formed from squared Bellman-residual norms under a target law and the data law. Both coefficients can be smaller than a state--action density-ratio bound, and our construction follows the same function-class-relative principle. The difference is the transferred error family and the target functional: we use actual critic secants, take a policy supremum inside a next-state integral over one fixed RKHS policy class, and allow the information constant to depend on the prescribed evaluation accuracy. Consequently, a uniform nonlinear bound can yield vanishing evaluation error even when the corresponding uniform linear error ratio is infinite. The low-rank specialization of \citet{UeharaSun2022PartialCoverage} also includes an action-probability ratio. The importance of dynamics beyond static feature coverage appears in linear off-policy evaluation: \citet{WangFosterKakade2021Limits} show exponential horizon dependence despite realizability and feature coverage, \citet{PerdomoEtAl2023OPE} characterize distinct conditions for FQI and LSTD to succeed, and \citet{AmortilaEtAl2026Coverage} develop feature-dynamics coverage for LSTDQ. These results motivate distinguishing target approximation, directional information, and state propagation in our functional-action analysis.

\input{experiment_source17.tex}

%% file: experiment_source17.tex
\section{Experimental design and implementation}\label{app:experiment-design}
\paragraph{Cohorts and estimands.}
The experiment uses 20 independent datasets and their associated training runs, with random seeds 0 through 19 for reproducibility, at each of five sample sizes, $n=NT\in\{2000,4000,8000,16000,32000\}$ with $T=20$ and $N\in\{100,200,400,800,1600\}$. FQI with functional actions uses either AdaFNN or Nystr\"om KRR. For each approximator, sample size, and seed, five candidate policies use $\lambda_{\Omega,n}/s_\Omega\in\{10^{-4},10^{-3},10^{-2},10^{-1},1\}$ at $M=20$ and $\gamma=.95$, where the deterministic conversion factor $s_\Omega$ is defined in Appendix~\ref{app:implemented-penalty}. The sample-size panels use one selected candidate per approximator, sample size, and seed. The AdaFNN state-dependent constant-torque comparator is a separate, full $M=20$ FQI fit for every seed and sample size. It uses the corresponding offline dataset and common random numbers for evaluation.

\paragraph{Environment and behavior.}
Our Pendulum extension has $S=(\cos\theta,\sin\theta,\omega)$ and bounded torque-function action $a:[0,1]\to[-2,2]$, executed over $0.5$ physical time units. The entire function is chosen at the start of each decision interval; the constant-torque comparator restricts it to $a(u)\equiv c(S)$. In normalized within-action time $u$,
\[
 \frac{d\theta}{du}=0.5\omega,\qquad
 \frac{d\omega}{du}=0.5\left\{\frac{3g}{2}\sin\theta+3a(u)\right\},
 \qquad g=10,
\]
where the pendulum mass and length are both one.
Velocity is restricted to $[-8,8]$. Fourth-order Runge--Kutta uses 128 substeps. Stochastic transitions add wrapped angle noise with standard deviation $0.03$ and truncated angular-velocity noise with standard deviation $0.05$. Reward is one minus the integrated angle, velocity, and torque cost,
\[
 \frac{\int_0^1\{\operatorname{wrap}(\theta(u))^2+0.1\omega(u)^2+0.001a(u)^2\}\,du}
 {\pi^2+0.1(8)^2+0.001(2)^2},
\]
clipped to $[0,1]$ for numerical tolerance. The behavior action combines a baseline controller, one persistent subject-level Mat\'ern-$5/2$ Gaussian-process realization, and an independent step-level Mat\'ern-$5/2$ realization. The step amplitude is $0.35$. Subjects are independent while decisions within a subject share its persistent component. Sample-size variants are nested subject prefixes of the same seed-specific data pool.

\paragraph{Critic and policy optimization.}
AdaFNN uses four learned global bases, 128-point action quadrature, and a three-hidden-layer width-256 ReLU outer network. Its critic is trained with early stopping. The functional policy is a bounded, state-dependent cubic B-spline with $p_u=12$ coefficients; full-covariance CMA-ES \citep{HansenOstermeier2001} optimizes the policy objective for AdaFNN. The state-dependent constant-torque comparator uses the common critic family and offline dataset but restricts each decision to $a_s(u)\equiv c(s)$. Its within-action curvature vanishes, so the curvature penalty does not affect its policy objective. KRR uses the Nystr\"om approximation \citep{WilliamsSeeger2000}, fixed Gaussian state lengthscales $(1,1,2)$, an action lengthscale equal to the training-action median quadrature-weighted $L^2$ distance, and landmark ranks $512,1024,1536,2304,3584$ across the five sample sizes. The KRR critic ridge coefficient is selected within each FQI iteration by fivefold subject-grouped cross-validation of the fixed regression target, without policy-value information. For KRR policy optimization, we use Adam \citep{KingmaBa2015} following the implementation of \citet{LinEtAl2026}, with bounded outputs enforced by our policy parameterization. KRR and AdaFNN use a common policy-coefficient grid and environment; their critic approximators remain distinct.

\subsection{Dimensionless empirical total-curvature policy objective}\label{app:implemented-penalty}
For the cubic B-spline action $\pi_C(s)(u)=B_{p_u}(u)^\top c_C(s)$, write
\[
 W_{p_u}=\int_0^1 B_{p_u}''(u)B_{p_u}''(u)^\top\,du,\qquad
 \widehat\Omega_{\mathcal D}(\pi_C)=\frac1n\sum_{j=1}^n c_C(S_j^+)^\top W_{p_u}c_C(S_j^+).
\]
This is the uncentered integral of squared second derivatives, averaged over all observed next states; it is not a centered across-state functional variance. Let $\mathcal I_\Phi$ index a fixed subset of the training next states and define the fitted value term at FQI iteration $m$ by
\[
 \widehat\Phi_m(\pi_C)
 :=\frac{1}{|\mathcal I_\Phi|}\sum_{j\in\mathcal I_\Phi}
 \widehat Q_m\{S_j^+,\pi_C(S_j^+)\}.
\]
The implemented dimensionless objective is
\begin{align*}
 \mathcal J_{\lambda_{\Omega,n},m}(C)
 &=
 \frac{\widehat\Phi_m(\pi_C)}{Q_{\rm scale}}
 -\frac{\lambda_{\Omega,n}}{s_\Omega}\frac{\widehat\Omega_{\mathcal D}(\pi_C)}{\Omega_{\rm scale}(p_u)},\\
 Q_{\rm scale}&=(1-\gamma)^{-1}=20,\qquad
 \Omega_{\rm scale}(p_u)=4p_u\lambda_{\max}(W_{p_u}),\qquad
 s_\Omega=\frac{Q_{\rm scale}}{\Omega_{\rm scale}(p_u)}.
\end{align*}
The value term averages the fitted critic over this fixed subset, whereas empirical total curvature averages over all $n$ training next states. The scales are deterministic for fixed $p_u$ and $\gamma$ and do not depend on realized curvature or a particular seed. Because the policy coefficients lie in $[-2,2]^{p_u}$, $\widehat\Omega_{\mathcal D}(\pi_C)\le\Omega_{\rm scale}(p_u)$. Multiplying the objective by $Q_{\rm scale}$ gives $\widehat\Phi_m(\pi_C)-\lambda_{\Omega,n}\widehat\Omega_{\mathcal D}(\pi_C)$. The experiment therefore replaces direct computation of the complete policy-RKHS norm in Algorithm~\ref{alg:ffqi-main} by empirical total curvature as a surrogate for its within-action regularity; it does not add a second penalty. The tuning grid is expressed in the dimensionless ratio $\lambda_{\Omega,n}/s_\Omega$. Appendix~\ref{app:closed-balanced-rates} proves that the resulting spline class satisfies the same spectral smoothness envelope used in coverage transfer and gives conditions for sharper penalty-localized envelopes.

At $n=8000$, increasing the curvature-penalty coefficient from the smallest to the largest candidate value reduces the mean AdaFNN policy curvature from about $13{,}625$ to numerical zero, while the number of selected policies outside the constant-within-interval subclass falls from 7 to 0 across the 20 datasets.

\paragraph{Tuning, independent evaluation, and uncertainty.}
For each candidate fit, 1,000 length-100 Monte Carlo trajectories estimate unpenalized normalized policy value for validation. The largest validation mean selects one positive coefficient, with exact ties assigned to the smaller coefficient. Because the experiment concerns policy improvement, simulator rollouts provide these value estimates without fitting a separate policy evaluator such as FQE. The selected policy is evaluated on a disjoint set of 1,000 trajectories. Validation and evaluation use independent random numbers; within each set, candidate coefficients and approximators share common random numbers for a given data/training seed. The state-dependent constant-torque comparator is evaluated using the same random numbers as the selected functional policy. All reported values estimate $(1-\gamma)\mathbb E\sum_{t=0}^{99}\gamma^tR_t$ without subtracting the policy penalty. Monte Carlo sampling error remains, and at $\gamma=.95$ the normalized reward tail beyond 100 decisions is at most $\gamma^{100}\simeq0.00592$. Pointwise 95\% percentile intervals use 10,000 resamples of the 20 independent datasets and their associated training runs; trajectories and held-out diagnostic pairs are not resampled as independent training datasets. The figures show means for value and medians for critic-relative diagnostics.

\section{Additional AdaFNN and KRR results}\label{app:additional-ada}
\paragraph{Selected-policy sample-size curves.}
The mean AdaFNN values at $n=2000,4000,8000,16000,32000$ are $0.587,0.593,0.605,0.657,0.702$, compared with $0.531,0.532,0.535,0.559,0.560$ for the state-dependent constant-torque comparator. Seed-level mean differences are descriptive rather than a claim about every dataset. For KRR, the selected-policy means on the common grid are $0.710,0.695,0.682,0.658,0.645$.

At the smallest sample size, KRR's mean normalized policy value is higher than AdaFNN's ($0.710$ versus $0.587$). As sample size grows, their value trends differ: AdaFNN improves while KRR declines. The KRR diagnostics in Appendix~\ref{app:empirical-spectral} show increasing relative feature leverage and an overall increase in critic-energy ratios over the same sample sizes. The latter trend is not monotone, with a small dip from $n=4000$ to $n=8000$ in Figure~\ref{fig:krr-policy-action-diagnostics}(c). This coexistence motivates examining critic representation and its relevant difference directions when interpreting evaluation stability. The curves describe performance under the respective tuning procedures, with the two critic families reported separately.

\begin{figure}[H]
\centering
\includegraphics[width=\linewidth]{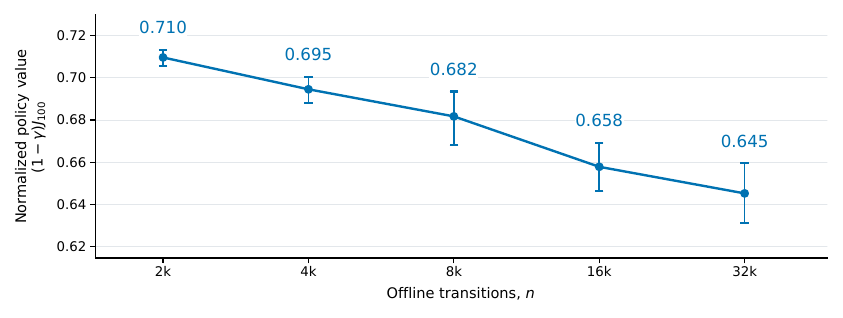}
\caption{Nystr\"om KRR normalized policy value across the sample sizes in Figure~\ref{fig:main-sample-behavior}(a), with independent evaluation after Monte Carlo validation of the curvature coefficient. Points are means across the 20 datasets with pointwise 95\% percentile-bootstrap intervals. Landmark rank also increases with $n$ as specified in Appendix~\ref{app:experiment-design}; the curve describes the implemented learning procedure, not an isolated effect of sample size.}
\label{fig:krr-sample-value}
\end{figure}

\paragraph{Paired policy-value differences.}
Figure~\ref{fig:paired-return-difference} compares the selected functional policy and the separately trained constant-torque policy within each data/training seed. Mean differences increase from $0.056$ at $n=2000$ to $0.141$ at $n=32000$. Their pointwise intervals lie above zero at all five sample sizes. These paired comparisons quantify the policy-value gain of the functional-policy learning procedure under the common offline-data and evaluation design.
\begin{figure}[H]
\centering
\includegraphics[width=\linewidth]{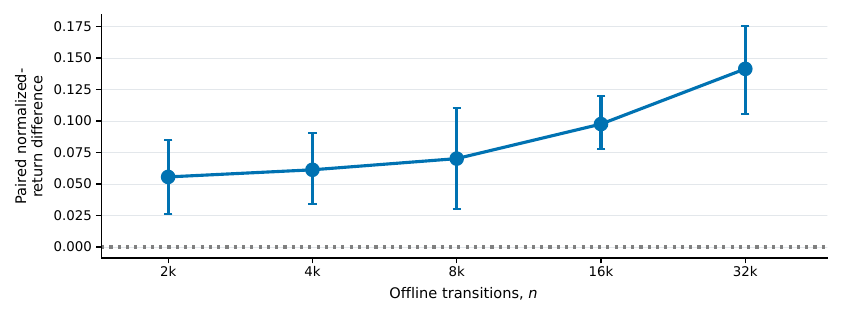}
\caption{AdaFNN functional-minus-constant difference in normalized policy value, paired within seed using common random numbers for evaluation. Points are mean differences; bars are pointwise 95\% percentile intervals from 10,000 whole-seed bootstrap resamples. All sample sizes from a seed are resampled together, but the displayed intervals are not a simultaneous confidence band.}
\label{fig:paired-return-difference}
\end{figure}

\subsection{Empirical locality and critic differences}\label{app:empirical-spectral}
\paragraph{Independent diagnostic sample and measurement.}
The diagnostics evaluate the fitted critics $\widehat Q_{19}$ and $\widehat Q_{20}$ on an independent evaluation cohort generated for each fit from the same environment and behavior policy as the offline data. At sample size $n=NT$, the cohort contains $N$ trajectories with $T=20$ decisions, yielding $q=NT=n$ evaluation states: $q\in\{2000,4000,8000,16000,32000\}$. Both the learned action and the observed behavior action are evaluated at each of these $q$ states. The training and evaluation trajectories are independent. Write $\mathcal N_k(s)$ for the indices of the $k$ nearest training states to $s$, ordered by
\[
 d_S(s,s')^2=(s_1-s_1')^2+(s_2-s_2')^2+(s_3-s_3')^2/4,
\]
with ties broken by training-record order. For a query action $a$ at state $s$, the nearest-action $L^2$ distance is
\[
 d_k(s,a)=\min_{(i,t)\in\mathcal N_k(s)}
 \left[\int_0^1\{a(u)-A_{i,t}(u)\}^2\,du\right]^{1/2},
\]
computed by quadrature on the common 128-point action grid. At each evaluation state, the learned action $\widehat\pi_{20}(s)$ and the observed behavior action are compared with the candidate set $\mathcal N_k(s)$. Each fit contributes the median distance for each action class; Figure~\ref{fig:main-sample-behavior}(b) reports the median across the 20 datasets, with $k=32$. The behavior-action baseline is not the constant-torque comparator. This finite-sample distance describes closeness to recorded actions at nearby states, not probability mass or a population support gap.

\paragraph{Sensitivity to neighborhood size.}
Figure~\ref{fig:adafnn-neighbor-sensitivity} reports the comparison for $k\in\{16,32,64,128\}$, keeping the fitted policies, selected penalty coefficients, and evaluation queries fixed. Since $\mathcal N_k(s)\subseteq\mathcal N_{k+1}(s)$, we necessarily have $d_{k+1}(s,a)\le d_k(s,a)$. A decrease with $k$ is therefore a consequence of minimizing over more candidates, not evidence of improved coverage. Larger neighborhoods also weaken state locality: at $k=n$, the distance is to the entire action sample regardless of state, and need not be zero. The relevant sensitivity contrast is between learned and behavior actions at a fixed $k$. The median learned-action distance remains larger at every sample size and every displayed $k$. This supports the qualitative contrast over the tested neighborhood sizes, without asserting that it persists for arbitrary $k$ or establishes a population support gap.
\begin{figure}[H]
\centering
\includegraphics[width=\linewidth]{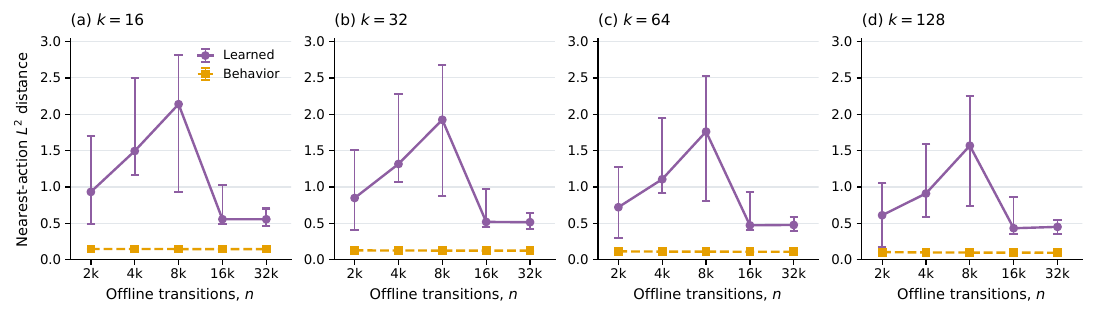}
\caption{AdaFNN nearest-action $L^2$ distances across sample sizes for (a) $k=16$, (b) $k=32$, (c) $k=64$, and (d) $k=128$ nearest training states. Each fit contributes the median nearest-action $L^2$ distance over its independent $q=n$ evaluation states; points are medians across 20 independent datasets, with pointwise 95\% percentile intervals from 10,000 whole-seed bootstrap resamples. Panel (b) reproduces Figure~\ref{fig:main-sample-behavior}(b). All panels use the same vertical scale.}
\label{fig:adafnn-neighbor-sensitivity}
\end{figure}

\paragraph{Adjacent critic differences.}
For each fit, the adjacent critic difference is $h=\widehat Q_{20}-\widehat Q_{19}$. The critic-energy ratio $\mathsf{ER}_h$ defined in Section~\ref{sec:main-experiments} can be written as
\[
 \mathsf{ER}_h
 =\frac{\operatorname{mean}_{(s,a)\in \mathcal X_{\rm pol}}h(s,a)^2}
        {\operatorname{mean}_{(s,a)\in\mathcal D}h(s,a)^2},
\]
where $\mathcal X_{\rm pol}$ denotes the held-out learned-policy graph and $\mathcal D$ the logged design, with every logged record counted. The displayed center is the median of the ratios from the 20 datasets, not a ratio of aggregated energies. All selected fits have positive design-energy denominators; no denominator regularization is applied. Values near one mean that this realized adjacent critic difference has comparable empirical squared energy on the two finite designs. This is a diagnostic for one adjacent difference per fit, rather than uniform control of the fitted-versus-target family in Section~\ref{sec:main-spectral}.

\paragraph{KRR representation diagnostics.}
Figure~\ref{fig:krr-policy-action-diagnostics} gives the KRR sample-size counterparts of the locality and adjacent-critic summaries, together with ridge leverage in the fitted Nystr\"om representation. Learned-policy leverage is normalized by the median behavior-action leverage from the corresponding fit and independent evaluation cohort; each fit contributes its median normalized learned-policy leverage. For each fitted KRR critic, the deterministic Nystr\"om feature map and the coefficients of $\widehat Q_{19}$ and $\widehat Q_{20}$ determine the adjacent-difference and leverage geometry.
\begin{figure}[H]
\centering
\includegraphics[width=\linewidth]{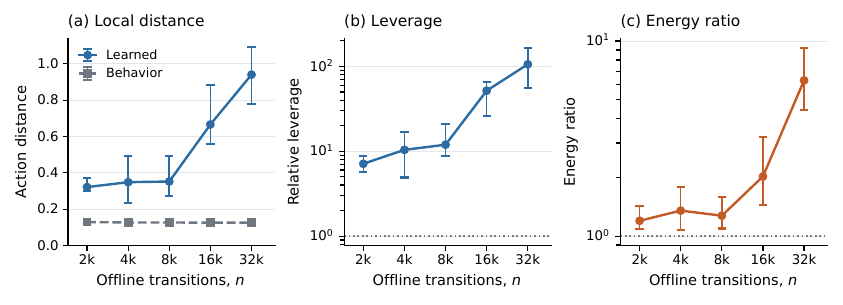}
\caption{Nystr\"om KRR policy-action diagnostics across the five sample sizes: (a) nearest-action $L^2$ distance, (b) relative feature leverage, and (c) critic-energy ratio. Every fit uses an independent cohort of $q=n$ evaluation states. Points are medians over the 20 datasets with pointwise 95\% percentile-bootstrap intervals across whole seeds. The offline state--action design and critic iterations 19 and 20 are common to AdaFNN and KRR; the critic representation and ridge geometry are approximator-specific. Panels (b)--(c) use logarithmic vertical scales.}
\label{fig:krr-policy-action-diagnostics}
\end{figure}

\paragraph{Absolute adjacent-critic energies.}
The ratio alone does not specify the scale of the critic update. Figure~\ref{fig:absolute-critic-energies} therefore displays its numerator and denominator separately for the selected fits. Each quantity is a mean squared difference for $\widehat Q_{20}-\widehat Q_{19}$, not a squared error relative to a true or Bellman-target critic.
\begin{figure}[H]
\centering
\includegraphics[width=\linewidth]{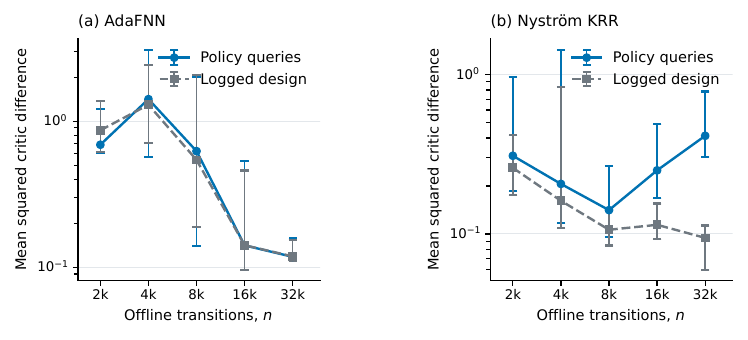}
\caption{Absolute adjacent-critic energies on independent held-out learned-policy state--action pairs and the logged design. The held-out mean evaluates the learned action at $q=n$ states for every fit. Points are medians of the energies from the 20 datasets with pointwise 95\% percentile-bootstrap intervals across whole seeds. Each panel uses a logarithmic vertical scale. The median of the per-fit ratios in Figure~\ref{fig:main-sample-behavior}(c) is generally not the ratio of these plotted medians.}
\label{fig:absolute-critic-energies}
\end{figure}